\documentclass[10pt]{amsart}
\usepackage{amssymb}
\usepackage{setspace}
\usepackage[mathscr]{eucal}
\usepackage{eufrak,verbatim}
\usepackage{amsthm}
\usepackage{graphicx}
\usepackage{fancyhdr}

\newtheorem{theorem}{Theorem}
\newtheorem{proposition}{Proposition}[subsection]
\newtheorem{lemma}[proposition]{Lemma}
\newtheorem{corollary}[proposition]{Corollary}

\theoremstyle{definition}
\newtheorem{remark}{Remark}[subsection]

\theoremstyle{remark}

\newtheorem*{notation}{Notation}
\numberwithin{equation}{subsection}

\newcommand{\snew}{\mathbf{s}}
\newcommand{\ringV}{\mathring{\mathbf{V}}}
\newcommand{\ringVm}{\mathring{\mathbf{V}}^m}
\newcommand{\Vm}{\mathbf{V}^m}

\newcommand{\Bm}{\mathcal{B}^m}
\newcommand{\Vb}{\bar{\mathbf{V}}}
\newcommand{\Mg}{\underline{g}{}}
\newcommand{{\Ent}}{{\eta}}

\evensidemargin=\oddsidemargin 
\begin{document}
\pagestyle{fancy}
\title{Well-Posedness for the Euler-Nordstr\"{o}m System with Cosmological Constant}
\author{Jared Speck}
\address{Department of Mathematics, Hill Center-Busch Campus, Rutgers University \\ 110 Frelinghuysen
Rd Piscataway, NJ 08854-8019, USA}
\email{jrspeck@math.princeton.edu}
\date{October $12^{th}$, 2008}

\begin{abstract}
		In this paper the author considers the motion of a relativistic perfect
    fluid with self-interaction mediated by
    Nordstr\"{o}m's scalar theory of gravity. The evolution of the fluid is determined by
    a quasilinear hyperbolic system of PDEs, and a cosmological
    constant is introduced in order to ensure the existence of
    non-zero constant solutions. Accordingly, the initial value problem
    for a compact perturbation of an infinitely extended quiet fluid is studied.
    Although the system is neither symmetric hyperbolic nor
    strictly hyperbolic, Christodoulou's constructive results on the existence of energy
    currents for equations derivable from a Lagrangian can be adapted to
    provide energy currents that can be used in place of the standard energy
    principle available for first-order symmetric
    hyperbolic systems. After providing such energy currents,
    the author uses them to prove that the Euler-Nordstr\"{o}m system with a cosmological constant
    is well-posed in a suitable Sobolev space.
\end{abstract}
\maketitle


\section{Introduction}                                                               \label{S:Introduction}
    It is well-known that for \emph{symmetric hyperbolic} systems of PDEs, an energy principle is available
    that implies well-posedness (local existence, uniqueness, and continuous dependence on initial data) for
    initial data belonging to an appropriate Sobolev space. Consult \cite{rCdH1966},
    \cite{cD2000}, \cite{kF1954}, \cite{pL2006}, \cite{aM1984}, or \cite{dS1999} for the definition of a symmetric hyperbolic system
    and a detailed proof of local existence in this case. A full proof of well-posedness is
    difficult to locate in the literature, but Kato \cite{tK1975}
    supplies one using a very general setup that applies to symmetric
    hyperbolic systems in a Banach space. Additionally, for \emph{strictly hyperbolic} (not necessarily
    symmetric) systems, well-posedness follows from the availability of a generalization
    of the energy principle for symmetric hyperbolic systems.  For strictly hyperbolic systems,
    there are a variety of methods due to
    Petrovskii, Leray, G$\mathring{\mbox{a}}$rding, and Calder\'{o}n for
    generating energy estimates; consult \cite{rCdH1966} or \cite{pL2006} for
    details on these methods.

    We consider here the Cauchy problem for the Lorentz covariant
    Euler-Nordstr\"{o}m (EN) system, which is a scalar caricature of the general covariant
    Euler-Einstein system describing a gravitationally self-interacting fluid.
    The EN system is a quasilinear hyperbolic system of PDEs that is not manifestly symmetric hyperbolic.
    Moreover, because of the repeated factors in the expression for $\mathcal{Q}(x;\cdot)$ in
    equation \eqref{E:NormalCone} below, and because the sheets of the characteristic
    subset of the cotangent space at $x$ intersect (see Fig. \ref{Fi:CotCharSubset}), it is not strictly hyperbolic. 
    Therefore, well-posedness for the EN system does not follow from either of these two well-known frameworks.

    Fortunately, alternate techniques recently developed by Christodoulou \cite{dC2000}, and which are
    applied to the study of relativistic fluids in Minkowski spacetime in particular in
    \cite{dC2007}, offer a viable approach to studying the Cauchy problem for the EN system.
    The central advantage afforded by Christodolou's techniques, which provide
    energy currents for equations derivable from a Lagrangian, is that they
    bypass the physically artificial requirement of symmetry in
    the equations: even though the EN system is not
    manifestly symmetric, its energy currents
    allow for precisely the same energy estimates to be made as in the
    theory of symmetric hyperbolic systems. Once one has these
    estimates, the proof of well-posedness for the EN system 
    mirrors the well known proof for symmetric hyperbolic
    systems. Our main goal is to use the method of
    energy currents to prove the following theorem (stated loosely here), which is
    divided into parts and stated rigorously in Section \ref{S:WellPosedness}:
    \vspace{\baselineskip}
		\begin{changemargin}{.25in}{.25in}    
			
				\noindent {\bf Main Theorem (Well-Posedness).} \ \ Let $N \geq 3$ be an integer. Assume that the initial data $\ringV$ for 
				the EN system are an $H^N$ perturbation of a constant background solution $\Vb.$
        Then these data launch a unique solution $\mathbf{V}$
        possessing the regularity property $\mathbf{V} - \Vb \in C_b^1([0,T] \times \mathbb{R}^3) \cap C^0([0,T],H^{N}) \cap 
        C^1([0,T],H^{N-1}).$ Furthermore, the map from the initial perturbation $\ringV -
        \Vb$ to $\mathbf{V} - \Vb$ is a continuous map from an open subset of $H^N$
        into $C^0([0,T],H^{N}).$
    
    \end{changemargin}
    \vspace{\baselineskip}

    While Christodoulou's methods are not the only techniques available for proving the well-posedness
    of the EN system, they are powerful and natural in the sense that they exploit
    the inherent geometry of the equations. In contrast, one may proceed
    by seeking a change of state-space variables that renders the system symmetric
    hyperbolic. For example, Makino applies this symmetrizing technique to the Euler-Poisson equations
    in \cite{tM1986}, and Makino and Ukai apply it to the
    relativistic Euler equations without gravitational interaction
    in \cite{tMsU1995} and \cite{tMsU1995b}. Further discussion of applications of
    symmetrization discussed in the literature can be found in sections \ref{SS:EPkappa} and
    \ref{SS:EEkappa}.\footnote{The references given are far from exhaustive; we merely wish to provide the reader
    with some examples of the application of well-known techniques.} Yet the symmetrizing method is not without disadvantages: one
    must solve a formally over-determined system of equations to find the
    symmetrizing variables\footnote{Consult chapter 3 of \cite{cD2000} for a discussion of symmetrization.}, and the resulting 
    state-space variables, if they exist, may place un-physical and/or mathematically unappealing
    restrictions on the function spaces with which one would like to work. However, it should be noted that
    Makino's symmetrization is currently capable of dealing with a
    restricted class of compactly supported data, while
    the techniques applied here cannot yet handle such data due to singularities in the
    energy current \eqref{E:EnergyCurrent} when the proper energy density $\rho$ of the
    fluid vanishes.

\section{Remarks on the Notation}
    We introduce here some notation that is used throughout this
    article, some of which is non-standard. We assume that the reader is familiar with standard notation for the
    $L^p$ spaces and the Sobolev spaces $H^k.$ Unless otherwise stated, the symbols $L^p$ and $H^k$ refer to
    $L^p(\mathbb{R}^3)$ and $H^k(\mathbb{R}^3)$ respectively.

    \subsection{Notation and assumptions regarding spacetime}           \label{SS:SpacetimeConventions}
    In the Euler-Poisson system with cosmological constant introduced below, we use
    $t \in \mathbb{R}$ to denote the time variable and $\snew \in \mathbb{R}^3$
    to denote the space variable. In the Euler-Einstein and EN systems (which we also equip with a cosmological constant
    below), we assume that spacetime is a 4-dimensional, time-orientable
    Lorentzian manifold $\mathcal{M}$ and use the notation
    \begin{align}                               \label{E:SpacetimePoint}
        x=(x^0,x^1,x^2,x^3)
    \end{align}
    to denote spacetime points. For the EN system with cosmological constant, we assume
    the existence of a global system of rectangular coordinates (an inertial frame),
    and for this preferred time-space splitting, we identify
    $t=x^0$ with time and $\snew=(x^1,x^2,x^3)$ with space and use the notation \eqref{E:SpacetimePoint}
    to denote the components of $x$ relative to this fixed coordinate
    system.

    \subsection{Notation regarding differential operators}

    If $F$ is a scalar or finite-dimensional array-valued function on
    $\mathbb{R}^{1 + 3},$ then
    $DF$ denotes the array consisting of all first-order spacetime partial derivatives (including
    the partial derivative with respect to time) of every component of
    $F,$ while $\nabla^{(a)} F$ denotes the array of
    consisting of all $a^{th}$ order \emph{spatial} partial derivatives of every component of $F;$ this
    should not be confused with $\nabla,$ which represents covariant differentiation.

    \subsection{Index conventions}                  \label{SS:IndexConventions}
    We adopt Einstein's notation that repeated Latin
    indices are summed from $1 \ \mbox{to} \ 3,$ while repeated Greek
    indices are summed from $0 \ \mbox{to} \ 3.$ Indices are
    raised and lowered using a spacetime metric, which varies according to context.

		\subsection{Notation regarding norms and function spaces}

    If $\bar{\mathbf{V}}$ is a constant array, we use the
    notation
    \begin{align} \notag
        & \|F\|_{L_{\Vb}^p(\mathscr{A})}
        \overset{\mbox{\tiny{def}}}{=} \|F - \Vb
        \|_{L^p(\mathscr{A})},
    \end{align}
    and we denote the set of all Lebesgue measurable functions $F$ such that \\ $\|F\|_{L_{\Vb}^p(\mathscr{A})} < \infty$
    by $L_{\Vb}^p(\mathscr{A}).$ Unless we indicate otherwise, we
    assume that $\mathscr{A} = \mathbb{R}^3$ when the set $\mathscr{A}$ is not
    explicitly written.

    If $F$ is a map from the interval $[T_1,T_2]$ into the normed function space $X,$ we use the notation
    \begin{equation}
        \mid\mid\mid F \mid\mid\mid_{X,[T_1,T_2]} \ \overset{\mbox{\tiny{def}}}{=} \underset{t \in
        [T_1,T_2]}{\sup} \|F(t)\|_{X} \notag.
    \end{equation}
    We often abbreviate $\mid\mid\mid F \mid\mid\mid_{X,T}$ in place of $\mid\mid\mid F \mid\mid\mid_{X,[0,T]}.$
    
    We also use the notation $C^k([0,T],X)$ to denote the set of $k$-times continuously differentiable maps from $(0,T)$ into
    $X$ that, together with their derivatives up to order $k,$ extend continuously to $[0,T].$

    If $\mathscr{A} \subset \mathbb{R}^d$ ($d$ frequently
    equals $3, 4,$ or $10$ in this article) and $\mathscr{A}$ is open, then
    $C^k_b(\bar{\mathscr{A}})$ denotes the set $k-$times continuously differentiable
    functions (either scalar or array-valued, depending on context) on $\mathscr{A}$ with bounded derivatives up to
    order $k$ that extend continuously to the closure of $\mathscr{A}.$ The norm of a function $F \in C^k_b(\bar{\mathscr{A}})$ is
    defined by
    \begin{equation}
        |F|_{k,\mathscr{A}} \overset{\mbox{\tiny{def}}}{=} \sum_{|\vec{\alpha}|\leq k} \sup_{z \in \mathscr{A}}
        |\partial_{\vec{\alpha}}F(z)|, \notag
    \end{equation}
    where $\partial_{\vec{\alpha}}$ represents differentiation
    with respect to the arguments $z$ of $F$ (which may be spacetime
    variables or state-space variables, depending on the context).

    \subsection{Notation regarding operators}
     If $X$ and $Y$ are normed function spaces, then
     $\mathcal{L}(X,Y)$ denotes the set of bounded linear maps from $X$ to $Y.$
     If $\mathcal{U} \in \mathcal{L}(X,Y),$ then we denote its operator norm by
     $\|\mathcal{U}\|_{X,Y}.$ If
     $X=Y,$ we write $\mathcal{L}(X)$ instead of $\mathcal{L}(X,Y)$
     and $\|\mathcal{U}\|_{X}$ instead of $\|\mathcal{U}\|_{X,X}.$ If $\mathcal{U}(t,t')$ is an operator-valued map from the
    triangle \\ $\triangle_T \overset{\mbox{\tiny{def}}}{=} \lbrace 0
    \leq t' \leq t \leq T \rbrace$ into $\mathcal{L}(X),$ then we adopt the notation

    \begin{align} \notag
        \mid\mid\mid \mathcal{U} \mid\mid\mid_{X,\triangle_T} \overset{\mbox{\tiny{def}}}{=}
        \sup_{(t,t') \in \triangle_T} \|\mathcal{U}(t,t')\|_{X}.
    \end{align}

   \subsection{Notation regarding constants}
    We use the symbol $C$ to denote a generic constant in the estimates
    below which is free to vary from line to line. If the constant
    depends on quantities such as real numbers $N',$ subsets $\mathscr{A}$ of
    $\mathbb{R}^d,$ functions $F$ of the state-space variables, etc., that are peripheral
    to the argument at hand, we sometimes
    indicate this dependence by writing $C(N',\mathscr{A},F),$
    etc. We frequently omit the dependence of $C$ on functions of
    the state-space variables below in order to conserve
    space, but we explicitly show the dependence when it is (in
    our judgment) illuminating. Occasionally, we shall use additional symbols
    such as $C_{\bar{\mathcal{O}}_2}, L, K,$ etc., to denote
    constants that play a distinguished role in the discussion below.

\section{The EN and EN$_\kappa$ Models in Context}
    
    The EN system is an intermediate model in between
    the Galilean covariant Euler-Poisson (EP) and the general covariant
    Euler-Einstein (EE) systems for self-gravitating classical fluids. Although
    it is the most fundamental of these models for self-gravitating Eulerian fluids,
    the EE system presents numerous technical difficulties
    that make a detailed analysis of the system's evolution,
    through either numerical or analytical methods, extremely difficult.
    For example, in General Relativity there is a coordinate gauge freedom due to the diffeomorphism covariance of the equations,
    and furthermore, there is no known law of local conservation of gravitational energy. 
    Our main motivations for studying the EN system are to bridge the gap between the EP and the
    EE systems and to provide a special relativistic primer for studying the EE system.

    Since it is based on Nordstr\"{o}m's theory of gravity,
    it should be stressed that the EN system is physically wrong.
    However, since both the EN and the EE systems are
    relativistic generalizations of the EP system, we
    expect, at least in some limiting cases, that there are some
    qualitative similarities between solutions to the three
    systems. Furthermore, in \cite{sSsT1993}, Shapiro and Teukolsky
    discuss numerical simulations of the EN system in the spherically symmetric
    case; they expect that the numerical schemes developed in their paper can
    be adapted to allow for the calculation of accurate wave
    forms in the EE model.

    Before discussing the EN system in detail, we briefly recall
    the EP and EE systems, endowing both
    with a \emph{cosmological constant}\footnote{We deviate from Einstein's notation; he denoted the cosmological
    constant by $\Lambda.$} denoted by
    $\kappa^2.$ We also briefly discuss some local existence proofs for these systems
    in the case $\kappa = 0,$ emphasizing their dependence on the symmetric hyperbolic setup or
    the method of Leray (strict) hyperbolicity.

    We introduce a positive cosmological constant out of mathematical
    necessity: the EN system fails to have non-zero constant
    solutions without it. Our reasoning is similar to the reasoning that led Einstein to introduce
    the cosmological constant into General Relativity; he sought a static universe, and General Relativity without
    a cosmological constant features only Minkowski space as a static homogeneous solution (see
    \cite{aE1917}). We emphasize the presence of the cosmological constant $\kappa^2$ in the
    models by referring to them as the EP$_{\kappa},$ EE$_{\kappa},$ and EN$_{\kappa}$
    systems; note that EP$=$EP$_0,$ and similarly for the other
    two models.

        \subsection{The Euler-Poisson system with cosmological constant (EP$_{\kappa}$)}    \label{SS:EPkappa}

    In units with Newton's universal gravitational constant equal to 1, the
    equations governing the dynamics in this case are
    \begin{align}
        \partial_t {\Ent} + v^k \partial_k {\Ent} &=  0                      \label{E:EP1}\\
        \partial_t{\rho} +  \partial_k(\rho v^k) &=  0              \label{E:EP2}\\
        \rho\left(\partial_t v_j + v^k \partial_k v^j\right)
        + \partial_j p + \rho \partial_j \phi &=  0 \qquad  (j=1,2,3),               \label{E:EP3}
    \end{align}
    where
    \begin{align}
        \Delta \phi - \kappa^2\phi &= 4 \pi \rho                                  \label{E:EP4}
    \end{align}
    and
    \begin{align}
        p &=\mathscr{P}(\rho,{\Ent}).
    \end{align}

    The unknowns in \eqref{E:EP1} - \eqref{E:EP4} are the cosmological
    Newtonian gravitational scalar potential $\phi(t,\snew),$ and
    the state-space variables mass density $\rho(t,\snew),$
    velocity $\mathbf{v}(t,\snew)=(v_1,v_2,v_3),$
    pressure $p(t,\snew),$ and entropy density\footnote{We are influenced by Boltzmann's notation
    in denoting the entropy density by $\Ent.$} ${\Ent}(t,\snew).$ We remark that in the EP$_{\kappa}$ system, $\phi$ is not a 
    state-space variable because it is uniquely determined by $\rho$ under the assumption
    of appropriate decay conditions on $\phi$ and $\rho$ at infinity.
    The equation that specifies $p$ as a function $\mathscr{P}$ of $\rho$ and ${\Ent}$ is known as the
    \emph{equation of state.} 
    
    This system of equations\footnote{Kiessling omits equation \eqref{E:EP1} from the system of equations he studies.
    See Section \ref{SS:EEkappa} for further discussion of this truncation.} is discussed in
    \cite{mK2003}, in which, under an \emph{isothermal} equation of
    state ($p=c_s^2 \rho,$ where the constant $c_s$ denotes the speed of sound), Kiessling derives the
    Jeans dispersion relation that arises from linearizing
    \eqref{E:EP2} - \eqref{E:EP4} about a static state in which the
    background mass density $\bar{\rho}$ is non-zero, followed by taking the limit $\kappa \to 0.$

    In \cite{tM1986}, Makino studies the Cauchy problem for the
    EP$_0$ system\footnote{Equation \eqref{E:EP1} is also omitted from Makino's paper.} with ``tame" compactly supported initial data 
    belonging to an appropriate Sobolev space. He studies \emph{adiabatic} equations of state $(p=K \rho^{\gamma},$ where
    $K$ is a positive constant) under the mathematical assumption $1 < \gamma < 3,$ and after finding symmetrizing
    variables, he proves local existence using the symmetric hyperbolic setup. 
    
    \begin{remark}
    Let us now make a few remarks about the ``tame" data. Vanishing mass densities typically 
    produce singularities in the expression for the energy, but Makino's choice of symmetrizing variables,
    which works for the class of adiabatic equations of state described in the previous paragraph, allows him to handle a class of 
    compactly supported data. The ``tame" data are constrained by the requirement that $\rho^{\delta}$ must belong to an appropriate 
    Sobolev space, where 
    $\delta$ is a positive constant depending on $\gamma.$ To the author's knowledge, a fully satisfactory treatment (i.e., without 
    unphysical mathematical
    restrictions on the data) of the evolution of compactly supported data in the EP$_0$ system remains an open problem.
   	\end{remark}   

    \subsection{The Euler-Einstein system with cosmological constant (EE$_{\kappa}$)}                       \label{SS:EEkappa}

        We work in units with Newton's universal gravitational constant and the speed of light both equal to 1.
        Given $T,$
    the energy-momentum tensor of the contemplated matter model,
    the gravitational spacetime with cosmological constant is determined by the Einstein field equations,
        \begin{align}
        G_{\mu \nu} + \kappa^2 g_{\mu \nu} = 8 \pi T_{\mu \nu}  \qquad (0 \leq \mu,\nu \leq 3),    \label{E:Einstein}
    \end{align}
    where $G$ is the \emph{Einstein tensor} of the
    \emph{spacetime metric} $g.$
        As a consequence of \eqref{E:Einstein}, $T$ has to satisfy the admissibility condition
    \begin{align}
        \nabla_{\mu} T^{\mu \nu} = 0 \qquad (0 \leq \nu \leq 3),         \label{E:EMContinuity}
    \end{align}
    where the $\nabla$ denotes the covariant derivative induced by the spacetime metric $g.$
    Equation \eqref{E:EMContinuity} follows from the twice contracted Bianchi identity, which implies that
    \begin{align}               \label{E:Bianchi}
        \nabla_{\mu} G^{\mu \nu} =0,
    \end{align}
    together with
    \begin{align}               \label{E:NablagCommute}
        {\nabla_{\lambda}} g^{\mu \nu} = 0 \qquad (0 \leq \lambda, \mu, \nu \leq 3),
    \end{align}
    which follows from the fact that $\nabla$ is the Levi-Civita connection on spacetime.

    We now briefly introduce the notion of a relativistic perfect fluid. Readers may consult \cite{nAgG2007} or
    \cite{dC1995} for more background. For a perfect fluid model, the components of the energy-momentum
    tensor of matter read
        \begin{align}
        T^{\mu \nu} \overset{\mbox{\tiny{def}}}{=} (\rho + p)u^{\mu}u^{\nu} + pg^{\mu \nu}.   \label{E:EMDef}
    \end{align}
    Here the scalar $\rho \geq 0$ is the \emph{proper energy density},
    the scalar $p \geq 0$ is the
    \emph{pressure,} and the vector $u$ is the \emph{4-velocity,} a
    future-directed timelike vectorfield which is subject to the normalization condition
        \begin{align}
        g_{\mu \nu}u^{\mu}u^{\nu}=-1.                                                                 \label{E:VelocityConstraint}
    \end{align}
   	
    We also introduce the additional thermodynamic scalar variables $n \geq 0$, the \emph{proper number density}, and ${\Ent} \geq 
    0,$ the \emph{proper entropy density}, and the following continuity equation:
    \begin{align}
        \nabla_{\mu}(n u^{\mu}) = 0.                                                                \label{E:Continuity}
    \end{align}

    When $g$ is given and $T$ is defined by \eqref{E:EMDef}, equations \eqref{E:EMContinuity} and \eqref{E:Continuity} together
    form the Euler equations for a general-relativistic perfect fluid. In general, when both $g$ and $T$ are unknowns, 
    \eqref{E:Einstein}, its consequence \eqref{E:EMContinuity}, and \eqref{E:EMDef} - \eqref{E:Continuity} form the EE$_{\kappa}$ 
    system for $u, \rho, p, n, \Ent,$ and $g$ (up to closure, for instance by providing two equations
    that relate $\rho, p, n,$ and $\Ent$). As in the EP$_{\kappa}$ system, the under-determined system
    consisting of \eqref{E:Einstein}, \eqref{E:EMContinuity}, and \eqref{E:EMDef} - \eqref{E:Continuity} may be closed
    (up to a choice of coordinate gauge) by providing 
    further relationships between the state-space variables. An example of a simple closure often
    discussed in the mathematical (consult e.g. \cite{nAgG2007}, \cite{tMsU1995}, \cite{tMsU1995b}) and astrophysical (consult e.g. 
    \cite{aR1992}) literature is to assume that $\rho$ is a function of $n$ alone, in which case equation \eqref{E:Continuity} is an 
    automatic consequence 
    of \eqref{E:EMContinuity}, \eqref{E:EMDef} and the 
    thermodynamic relation \eqref{E:Pressure} below. Equivalently, one may specify $p$ as a function of $\rho$ alone; such fluids
    are called \emph{barotropic}. If the fluid is barotropic, the variable $\Ent$ becomes passive in the sense that it satisfies the 
    equation $u^{\mu} \nabla_{\mu} \Ent =0,$ but does not otherwise enter into the dynamics; the remaining state-space variables
    (which we may take to be $u, g, p$) decouple from $\Ent.$
    
    Local existence for a closed relativistic fluid system has been discussed by several authors under 
    various assumptions. For example, in \cite{yCB1958}, Choquet-Bruhat showed that the EE$_0$ system 
    with pressure-free dust sources\footnote{The energy-momentum tensor for
    pressure-free dust has components $T^{\mu \nu} = \rho u^{\mu} u^{\nu}.$} forms a well-posed Leray-hyperbolic system, and in 
    \cite{aR1992}, Rendall adapted Makino's
    symmetrization (as discussed in Section \ref{SS:EPkappa}) of the EP$_0$ system to handle a subclass of 
    compactly supported initial data for the EE$_0$ system with perfect fluid sources under an adiabatic equation of state
    with $\gamma > 1.$ Similar results are also proved in \cite{uBlK2003},
    in which Brauer and Karp write the equations as a symmetric hyperbolic system
    in harmonic coordinates.

    \subsection{The Euler-Nordstr\"{o}m system with cosmological constant}       \label{SS:SettingUpENkappa}

    We base our discussion here on Calogero's derivation of the Nordstr\"{o}m-Vlasov system\footnote{Each of the three Eulerian fluid 
    models discussed in this article has a
            kinetic theory counterpart. Collectively known as the Vlasov
            models, these diffeo-integral systems describe a particle density function $f$ on
            physical space $\times$ momentum space that evolves due
            to gravitational self-interaction. In particular, the EN$_0$ system is the Eulerian
            counterpart of the previously studied Nordstr\"{o}m-Vlasov (NV) system (which
            does not feature a cosmological constant). See e.g., \cite{sC2003} or \cite{sC2006}.} \cite{sC2003}.
    Consult sections \ref{SS:SpacetimeConventions} and \ref{SS:IndexConventions} for
    some remarks on our assumptions concerning spacetime and our use of index notation. As in the
    EE$_{\kappa}$ model, we work in units with the speed of light and Newton's
    universal gravitational constant
    both equal to 1.

    Like the EE$_{\kappa}$ system, the EN$_{\kappa}$ system subsumes equations \eqref{E:EMContinuity}, \eqref{E:EMDef},
    \eqref{E:VelocityConstraint},
    and \eqref{E:Continuity}, where $\rho \geq 0, p \geq 0, n \geq 0, {\Ent} \geq 0,$ and $u$ are defined
    as in the EE$_{\kappa}$ system.  In contrast to the EE$_{\kappa}$
    model, we do \emph{not} assume Einstein's field equations \eqref{E:Einstein};
    instead we turn to Nordstr\"{o}m's theory of gravity. We postulate that in
    our global rectangular coordinate system, the conformally flat metric is given
    by
    \begin{align}
        g_{\mu \nu}\overset{\mbox{\tiny{def}}}{=}e^{2\phi}\Mg_{\mu \nu},       \label{E:NordstromMetric}
    \end{align}
    where $\phi$ is the \emph{Nordstr\"{o}m scalar potential}, and
    $\Mg=$ diag$(-1,1,1,1)$ are the components of the Minkowski metric in the rectangular coordinate system. 
    
    Nordstr\"{o}m's theory of gravity \cite{gN1913} belongs to the class
    of theories known as scalar metric theories of gravity.
    For theories in this class, gravitational forces are mediated
    by a scalar field (or ``potential") $\phi$ that affects the spacetime metric.
    Furthermore, it is assumed that the effect of $\phi$ is to
    modify the otherwise flat metric by a scaling factor that depends on
    $\phi.$ Therefore, the physical metric in such a theory is
    given by $g_{\mu \nu} = \chi^2(\phi) \Mg_{\mu \nu},$
    where $\Mg$ is the Minkowski metric. A metric of this form
    is said to be \emph{conformally flat}. Strictly speaking, the
    scalar theory of gravity we study in this paper is not identical to
    the one published by Nordstr\"{o}m in \cite{gN1913}. In his
    paper, Nordstr\"{o}m makes the choice $\chi(\phi)= \phi,$ while in
    our paper, we make the choice $\chi(\phi)= e^{\phi},$ a theory
    that appears as a homework exercise in the well-known text
    ``Gravitation" by Misner, Thorne, and
    Wheeler \cite{cMkTjW1973}. See \cite{sC2003} or \cite{tDgEF1992} concerning the significance
    of the choice $\chi(\phi)= e^{\phi},$ which has the property of \emph{scale invariance}
    of the gravitational interaction. Also consult \cite{fR2004} for a discussion of scalar theories of
    gravity, including the two mentioned here.
    
    Following Nordstr\"{o}m's lead \cite{gN1913}, we also introduce
    the auxiliary energy-momentum tensor $T_{\mbox{\tiny{aux}}}$ with components
    \begin{align}
        T_{\mbox{\tiny{aux}}}^{\mu \nu} \overset{\mbox{\tiny{def}}}{=} e^{6\phi}T^{\mu \nu}     \label{E:EMAuxDef}
    \end{align}
    and postulate that $\phi$ is a solution to
    \begin{align}
        \square \phi - {\kappa}^2 \phi = -\Mg_{\mu \nu}T_{\mbox{\tiny{aux}}}^{\mu \nu} = -e^{4\phi}(3p-\rho).        \label{E:Phi}
    \end{align}
    Note that $\square \phi \overset{\mbox{\tiny{def}}}{=} - \partial^2_t \phi + \Delta \phi$
    is the wave operator on flat spacetime applied to $\phi.$ The virtue of the postulate \eqref{E:Phi} is that it
    provides us with continuity equations for an energy-momentum tensor in Minkowski
    space which we label $\Theta$ and discuss below; see equations \eqref{E:ENEMDef} and \eqref{E:ENEMContinuity}.

    As in the EP$_{\kappa}$ and EE$_{\kappa}$ models, we may close the EN$_{\kappa}$ system by supplying relationships
    between the state-space variables. The basic postulates we adopt are as follows (see e.g. \cite{yGsTZ1999}):
    
    \begin{list}{}{\setlength{\leftmargin}{.5in}}
        \item $\mathbf{1) \ }$ $\rho \geq 0$ is a function of $n \geq 0$ and ${\Ent} \geq 0.$

    \item $\mathbf{2) \ }$ $p \geq 0$ is defined by
        \begin{align}
            p\overset{\mbox{\tiny{def}}}{=}n \left. \frac{\partial \rho}{\partial n} \right|_{\Ent} -
            \rho,                                                                                   \label{E:Pressure}
        \end{align}
        where the notation $\left. \right|_{\cdot}$ indicates partial differentiation with $\cdot$ held
        constant.
    \item $\mathbf{3) \ }$ A perfect fluid satisfies
        \begin{align}
            \left. \frac{\partial \rho}{\partial n} \right|_{\Ent} >0, \left. \frac{\partial p}{\partial
            n} \right|_{\Ent}>0, \left. \frac{\partial \rho}{\partial {\Ent}} \right|_n \geq 0 \
            \mbox{with} \ ``=" \ \mbox{iff} \ {\Ent}=0.                             \label{E:EOSAssumptions}
        \end{align}
        As a consequence, we have that $\sigma,$ the speed of sound in the fluid, is always real:
            \begin{align}
                \sigma^2 \overset{\mbox{\tiny{def}}}{=} \left.\frac{\partial p}{\partial
                \rho}\right|_{\Ent} = \frac{{\partial p / \partial n}|_{\Ent}}{{\partial \rho / \partial
                n}|_{\Ent}} > 0.                                                                         \label{E:SpeedofSound}
            \end{align}
     \item  $\mathbf{4) \ }$ We also demand that the speed of sound is positive and less than the speed of light whenever $n,\Ent > 0$:
            \begin{align} \label{E:Causality}
                0 < \sigma < 1.
            \end{align}
    \end{list}
    
    Postulates $\mathbf{1}$ - $\mathbf{3}$ express the laws of thermodynamics and fundamental thermodynamic assumptions, while as discussed 
    in detail in Section \ref{S:Geometry}, postulate $\mathbf{4}$ ensures that vectors that are timelike with respect to the sound cone are 
    necessarily timelike with respect to the light cone. 
    
    \begin{remark}
    We note that the assumptions $\rho \geq 0, p \geq 0$ together imply that
    the energy momentum tensor \eqref{E:EMDef} satisfies both the \emph{weak energy 
    condition} ($T_{\mu \nu} X^{\mu} X^{\nu} \geq 0$ holds whenever $X$ is future-directed and timelike) 
    and the \emph{strong energy condition}
    ($[T_{\mu \nu} - 1/2 g^{\alpha \beta}T_{\alpha \beta} g_{\mu \nu}]X^{\mu}X^{\nu} \geq 0$ holds whenever $X$ 
    is future-directed and timelike). Furthermore, if we assume that the equation of state is such that $p=0$ when
    $\rho = 0,$ then \eqref{E:Causality} guarantees that $p \leq \rho.$ It is then easy to check
    that $0 \leq p \leq \rho$ implies the \emph{dominant energy condition}
    ($-T^{\mu}_{\ \nu} X^{\nu}$ is future-directed and causal whenever $X$ is future-directed and causal).
		\end{remark}
    
    \begin{remark}
        By \eqref{E:EOSAssumptions}, we can solve for $\sigma$ and $\rho$ as
        functions of $p$ and ${\Ent}:$
        \begin{align}
            \sigma &= \mathscr{S}({\Ent},p)                                                         \label{E:Sigma}\\
            \rho &= \mathscr{R}({\Ent},p).                                                       \label{E:rho}
        \end{align}
    \end{remark}
    
    \begin{remark}
        We will make use of the following identity implied by \eqref{E:SpeedofSound}, \eqref{E:Sigma}, and \eqref{E:rho}:
        \begin{align} \label{E:partialRpartialpandStoNegativeTwoRelationship}
            \left. \frac{\partial \mathscr{R}}{\partial p}(\Ent,p)\right|_{\Ent} = \mathscr{S}^{-2}(\Ent,p).
        \end{align}
    \end{remark}

    As a typical example, we mention a \emph{polytropic} equation
    of state, that is, an equation of state of the form (see e.g.
    \cite{yGsTZ1999})
    \begin{align}                                                                           \label{E:Polytropic}
        \rho = n + \frac{A({\Ent})}{\gamma -1}n^{\gamma},
    \end{align}
    where $1 < \gamma < 2,$ and $A$ is a positive,
    increasing function of ${\Ent}.$ In this case $p=An^{\gamma},$
    $\left. \partial p / \partial \rho \right|_{\Ent}$ is
    increasing in $\rho,$ and the speed of sound $\sigma$ is bounded
    from above by $\sqrt{\gamma - 1}.$

    \begin{remark}
        We note here a curious discrepancy that arises when, for the polytropic equation of state under the isentropic condition
        ${\Ent} \equiv {\Ent}_0$, we consider the Newtonian limit, that is, the limit as the speed of light $c$ goes to $\infty.$
        In dimensional units, \eqref{E:Polytropic} becomes $\rho = m_0c^2n + \frac{A_c({\Ent})}{\gamma -1}n^{\gamma},$ and
        $p=A_c({\Ent})n^{\gamma},$
        where $m_0$ is the mass per fluid element, and $A_c({\Ent})$ is a positive, increasing function of ${\Ent}$ indexed by the parameter $c.$ The
        speed of sound squared is given by $\sigma^2 \overset{\mbox{\tiny{def}}}{=}c^2\left.\frac{\partial p}{\partial
        \rho}\right|_{\Ent}
        =\frac{\gamma c^2 A_c({\Ent}_0)n^{\gamma -1}}{c^2 m_0 + (\gamma /{\gamma -1}) A_c({\Ent}_0) n^{\gamma-1}}.$
        Assuming that $\lim_{c \to \infty} A_c({\Ent}_0) \overset{\mbox{\tiny{def}}}{=}A_{\infty}({\Ent}_0)$ exists, we may
        consider the Newtonian limit $c \to \infty$ of $\sigma^2$ and $p,$ obtaining in the limit that $\sigma^2 = \gamma
        m_0^{-1}A_{\infty}({\Ent}_0)n^{\gamma - 1}$ and
        $p=A_{\infty}({\Ent}_0)n^{\gamma},$ Newtonian formulas that make mathematical sense
        and have physical interpretations for $1 \leq \gamma < \infty.$ In the Newtonian case,
        $\gamma=1$ corresponds to isothermal conditions, while $\gamma \to \infty$ yields the rigid body dynamics.
        However, for finite values of $c,$ not all values of the parameter $\gamma$ make mathematical or physical sense: there is a
        mathematical singularity in the formula for $\rho$ at $\gamma = 1.$ This is physically reasonable since
        isothermal conditions require the instantaneous transfer of heat
        energy. Thus, for finite $c,$ the polytropic equations of state do not allow for the case
        corresponding to the instantaneous transfer of heat energy over finite distances, a feature which we find desirable in a
        relativistic model. Additionally, we have that $\lim_{n \to \infty} \sigma^2 = c^2
        (\gamma - 1),$ so that for $\gamma > 2,$ there is a $\gamma-$dependent critical threshold for the number
        density above which the speed of sound exceeds the speed of light. Since larger values of $\gamma$ correspond to
        ``increasing rigidity" of the fluid, and the concept of rigidity violates the spirit of the framework of relativity,
        we are not surprised to discover that large values of
        $\gamma$ may lead to superluminal sound speeds. However, we find ourselves at the moment unable to attach a
        physical interpretation to the fact that the mathematical borderline case is $\gamma=2.$
    \end{remark}

    We summarize this section by stating that equations \eqref{E:EMContinuity}, \eqref{E:EMDef}, \eqref{E:VelocityConstraint},
    \eqref{E:Continuity}, \eqref{E:NordstromMetric}, \eqref{E:EMAuxDef}, \eqref{E:Phi}, \eqref{E:Pressure},
    and \eqref{E:rho} constitute the EN$_{\kappa}$ system.

\section{Reformulation of the EN$_{\kappa}$ System, the Linearized EN$_{\kappa}$ System, and the Equations of Variation}

    Because it is mathematically advantageous, in this section we reformulate the EN$_{\kappa}$ system as a
    fixed-background theory in flat Minkowski space. This is a mathematical reformulation only; the ``physical"
    metric in the EN$_{\kappa}$ system is $g$ from \eqref{E:NordstromMetric} rather than the Minkowski metric $\Mg.$
    We also discuss the linearization of the EN$_{\kappa}$ system and the related equations of variation, systems that are central to
    the well-posedness arguments.

    \subsection{Reformulating the EN$_{\kappa}$ system}

    For the remainder of this article, indices are raised and lowered
    with the Minkowski metric, so for example, $\partial^{\lambda} \phi = \Mg^{\mu \lambda}
    \partial_{\mu} \phi.$ To begin, we use the form of the metric \eqref{E:NordstromMetric}
    to compute that in our fixed rectangular coordinate system (see
    Section \ref{SS:SpacetimeConventions}),
    the continuity equation \eqref{E:EMContinuity} for the
    energy-momentum tensor \eqref{E:EMDef} is given by
    \begin{align}
        0  = \nabla_{\mu} T^{\mu \nu} &= \partial_{\mu} T^{\mu \nu} + 6
                T^{\mu \nu} \partial_{\mu}\phi  -e^{-2\phi}
                g_{\alpha \beta}T^{\alpha \beta}\partial^{\nu} \phi \notag\\
         &= \partial_{\mu} T^{\mu \nu} + 6 T^{\mu \nu}
                \partial_{\mu}\phi- e^{-6 \phi}\Mg_{\alpha \beta}
                T_{\mbox{\tiny{aux}}}^{\alpha \beta} \partial^{\nu}\phi \qquad
                (\nu=0,1,2,3),                                                                                    \label{E:TDiv}
    \end{align}
    where $T_{\mbox{\tiny{aux}}}^{\mu \nu}$ is given by \eqref{E:EMAuxDef}. For this calculation we made use of the
    explicit form of the Christoffel symbols in our rectangular coordinate
    system:
    \begin{align}
        \Gamma_{\mu \nu}^{\alpha} = \delta^{\alpha}_{\nu}
        \partial_{\mu}\phi + \delta^{\alpha}_{\mu}
        \partial_{\nu}\phi - \Mg_{\mu \nu}\Mg^{\alpha \beta} \partial_{\beta}
        \phi.                                                                                               \label{E:Christoffel}
    \end{align}

    Under the postulate \eqref{E:Phi} for $\phi,$
    \eqref{E:TDiv} can be rewritten as

    \begin{align}
        0 = e^{6\phi} \nabla_{\mu} T^{\mu \nu} = \partial_{\mu} \big(T^{\mu
        \nu}_{\mbox{\tiny{aux}}}
        + \partial^{\mu} \phi \partial^{\nu} \phi -\frac{1}{2} \Mg^{\mu \nu} \partial^{\alpha} \phi
        \partial_{\alpha} \phi - \frac{1}{2} \Mg^{\mu\nu} \kappa^2 \phi^2 \big).                         \label{E:Divergence}
    \end{align}
    Equation \eqref{E:Divergence} now
    illustrates the divergence-free energy-momentum tensor
    $\Theta$ mentioned in Section \ref{SS:SettingUpENkappa}. Its components
    $\Theta^{\mu \nu}$ consist of the terms from \eqref{E:Divergence} that are inside the
    parentheses; we are thus afforded with local conservation laws in Minkowski space.

    To simplify the notation, we make the change of state-space
    variables (recalling equation \eqref{E:rho} for the definition
    of the function $\mathscr{R}$)
    \begin{align}
        U^{\nu} &\overset{\mbox{\tiny{def}}}{=}e^{\phi} u^{\nu}  \qquad \qquad (\nu = 0,1,2,3) \label{E:U}  \\
        R &\overset{\mbox{\tiny{def}}}{=} e^{4\phi} \rho  = e^{4\phi}\mathscr{R}(p,{\Ent})             \label{E:R} \\
        P &\overset{\mbox{\tiny{def}}}{=} e^{4\phi} p                                                 \label{E:P}
    \end{align}
    throughout the EN$_{\kappa}$ system, noting that $U$ is subject to the constraint
    \begin{align}
        U^0 = (1 + U^kU_k)^{1/2}.                                                       \label{E:U0Def}
    \end{align}
    Following the above substitutions, $\Theta$ has components
    \begin{align}
        \Theta^{\mu \nu} \overset{\mbox{\tiny{def}}}{=} (R+P)U^{\mu}U^{\nu} + P{\Mg}^{\mu \nu} + \partial^{\mu}\phi
        \partial^{\nu}\phi
        -\frac{1}{2} \Mg^{\mu \nu}\partial^{\alpha}\phi\partial_{\alpha} \phi
        - \frac{1}{2} \Mg^{\mu\nu} \kappa^2 \phi^2,                                            \label{E:ENEMDef}
    \end{align}
    and \eqref{E:Divergence} becomes
    \begin{align}
        \partial_{\mu} \Theta^{\mu \nu}=0 \qquad (\nu =0,1,2,3). \label{E:ENEMContinuity}
    \end{align}

    We perform the same changes of variables in the equation \eqref{E:Continuity}
    and expand the covariant differentiation in terms of coordinate
    derivatives and the Christoffel symbols \eqref{E:Christoffel}, arriving at the
    equation
    \begin{align}
        \partial_{\mu}\left(n e^{3\phi} U^{\mu} \right)=0.                                                 \label{E:ENContinuity}
    \end{align}

    For our purposes below, we take as our equations the projections of
    \eqref{E:ENEMContinuity} onto the orthogonal complement of $U$
    and in the direction of $U.$ In this formulation, the mathematical form of the EN$_{\kappa}$
    system is that of the relativistic Euler equations in Mikowski space without gravitational interaction (as presented in
    \cite{dC2007}), with inhomogeneous terms involving $D\phi,$ and supplemented
    by the linear Klein-Gordon equation \eqref{E:Phi} for $\phi.$ Thus, we introduce $\Pi,$ the projection
    onto the orthogonal complement of $U,$ given by
    \begin{align}
    \Pi^{\mu \nu} \overset{\mbox{\tiny{def}}}{=} U^{\mu}U^{\nu} + \Mg^{\mu \nu}.   \label{E:Projection}
    \end{align}

    Considering first the projection of \eqref{E:ENEMContinuity} in the
    direction of $U,$ we remark that one may use \eqref{E:Pressure} and
    \eqref{E:ENContinuity} to conclude that for $C^1$ solutions,
    $U_{\nu}\partial_{\mu} \Theta^{\mu \nu}=0$ is equivalent to
    \begin{align}
        U^{\mu} \partial_{\mu}{\Ent} = 0,                                                            \label{E:Entropy}
    \end{align}
    which implies that the entropy density ${\Ent}$ is constant along
    the integral curves of $U.$

    The projection of \eqref{E:ENEMContinuity} onto the
    orthogonal complement of $U$ gives the 4 equations (only 3 of which are independent)
    \begin{align}
        (R+P) U^{\mu} \partial_{\mu} U^{\nu} + \Pi^{\mu \nu} \partial_{\mu}P = -(\square
        \phi - \kappa^2 \phi)\Pi^{\mu \nu}\partial_{\mu} \phi\qquad (\nu=0,1,2,3). \label{E:ENEMContinutiyProjection}
    \end{align}

    By \eqref{E:rho}, \eqref{E:R} and \eqref{E:P}, we may solve for $R$ as a function $\mathfrak{R}$ of ${\Ent}, P$ and $\phi:$
    \begin{align}
        R=\mathfrak{R}({\Ent},P,\phi)\overset{\mbox{\tiny{def}}}{=} e^{4\phi}\mathscr{R}({\Ent},e^{-4\phi}P).           \label{E:RfunctionDef}
    \end{align}

    We also the nameless quantity $Q$ and make use of
    \eqref{E:Pressure},
    \eqref{E:SpeedofSound}, \eqref{E:Sigma}, \eqref{E:rho}, \eqref{E:partialRpartialpandStoNegativeTwoRelationship}, \eqref{E:R}, and 
    \eqref{E:P} to express it as a function $\mathfrak{Q}$ of ${\Ent}, P$ and $\phi:$
    \begin{align}   \label{E:QfunctionDef}
        Q = \mathfrak{Q}({\Ent},P,\phi) \overset{\mbox{\tiny{def}}}{=} n \left. \frac{\partial P}{\partial n}\right|_{{\Ent},\phi} &=
        \left. \frac{\partial P}{\partial \rho} \right|_{\Ent,\phi} \cdot n \left. \frac{\partial \rho}{\partial n} 
        \right|_{{\Ent}} = e^{4\phi}{\mathscr{S}}^2({\Ent},p)(\rho+p)  \\
        &={\mathscr{S}}^2({\Ent},e^{-4\phi}P)[\mathfrak{R}({\Ent},P,\phi) + P].   \notag
    \end{align}

    We also we use the chain rule together with \eqref{E:ENContinuity},
    \eqref{E:Entropy}, and \eqref{E:QfunctionDef} to derive
    \begin{align}
        U^{\mu}\partial_{\mu}P + Q \partial_{\mu} U^{\mu} = (4P-3Q)U^{\mu} \partial_{\mu}\phi,
    \end{align}
    which we may use in place of \eqref{E:ENContinuity}.

    Deleting the redundant equation from
    \eqref{E:ENEMContinutiyProjection}, using \eqref{E:U0Def} to derive the relation
    \begin{align}
        \partial_{\lambda}U^0 = \frac{U_k}{U^0} \partial_{\lambda} U^k,
    \end{align}
    and rewriting \eqref{E:Phi} as an equivalent
    first order system, the working form of the EN$_{\kappa}$ system that we adopt is

    \begin{align}
            U^\mu \partial_\mu {\Ent} &= 0                                           \label{E:ENkappa1}\\
            U^\mu \partial_\mu P + Q\frac{U_k}{U^0}
                \partial_0U^k+ Q \partial_k U^k  &= (4P - 3Q)U^\mu
                \psi_{\mu}                                                      \label{E:ENkappa2}\\
            (R + P)U^\mu \partial_\mu U^j + \Pi^{\mu j} \partial_\mu P
            &= (3P-R) \Pi^{\mu j} \psi_{\mu} \qquad (j =1,2,3) \label{E:ENkappa3}\\
            -\partial_0 \psi_0 + \partial^j \psi_j &= \kappa^2 \phi + R - 3P                                             \label{E:ENkappa4}\\
            \partial_0 \psi_j - \partial_j \psi_0 &= 0 \qquad (j =1,2,3)               \label{E:ENkappa5}\\
            \partial_0 \phi &= \psi_0.      \label{E:ENkappa6}
    \end{align}
    Here, $U^0,R,$ and $Q$ are expressed in terms of the
    unknowns through the relations
    \begin{align}
        U^0 &= (1 + U^k U_k)^{1/2} \label{E:ENkappa7} \\
        Q   &=\mathfrak{Q}({\Ent},P,\phi) \label{E:ENkappa8}\\
        R   &= \mathfrak{R}({\Ent},P,\phi) \label{E:ENkappa9},
    \end{align}
    where the function $\mathfrak{Q}$ is defined in \eqref{E:QfunctionDef}, and the function $\mathfrak{R}$ is defined in
    \eqref{E:RfunctionDef}. In our rewriting of \eqref{E:Phi} as a first order
    system, we treat $\psi_{\nu} \overset{\mbox{\tiny{def}}}{=} \partial_{\nu} \phi$
    as separate unknowns for $\nu = 0,1,2,3.$ 
    
    To simplify the notation, we collect the unknowns
    $\mathbf{V}$ together into
    an array\footnote{Although every array appearing in this
        article is a $q \times 1$ column vector, we write them as
        if they were row vectors to save space.} given by
    \begin{align}
        \mathbf{V} & \overset{\mbox{\tiny{def}}}{=}
        ({\Ent},P,U^1,U^2,U^3,\phi,\psi_0,\psi_1,\psi_2,\psi_3),
        \label{E:Vdef}
    \end{align}
    and we refer to the first five components of $\mathbf{V}$ as
    \begin{align}
        \mathbf{W} & \overset{\mbox{\tiny{def}}}{=} ({\Ent},P,U^1,U^2,U^3).                                  \label{E:W}
    \end{align}

    \subsection{Linearization and the Equations of Variation (EOV)}                                                     \label{SS:EOV}
        The standard techniques for proving well-posedness
        require the linearization of the EN$_{\kappa}$ system around a known
        background solution, which we refer to as a ``bgs." Each
        bgs $\widetilde{\mathbf{V}}: \mathcal{M} \rightarrow \mathbb{R}^{10}$
        we consider is of the form $\widetilde{\mathbf{V}} = (\widetilde{{\Ent}},\widetilde{P},
        \cdots, \widetilde{\psi}_2, \widetilde{\psi}_3).$ The resulting system is known as the \emph{equations of variation} (EOV).
        Thus, given such a $\widetilde{\mathbf{V}}$ and inhomogeneous terms
        $f,g,\cdots,l^{(4)},$ we define the EOV by

    \begin{align}
            \widetilde{U}^\mu \partial_\mu \dot{{\Ent}} &= f                                 \label{E:EOV1} \\
            \widetilde{U}^\mu \partial_\mu \dot{P} +
                \widetilde{Q}\frac{\widetilde{U}_k}{\tilde{U}^0}
                \partial_0\dot{U}^k+ \widetilde{Q} \partial_k \dot{U}^k  &= g           \label{E:EOV2} \\
            (\widetilde{R} + \widetilde{P})\widetilde{U}^\mu \partial_\mu
                \dot{U}^j + \widetilde{\Pi}^{\mu j} \partial_\mu \dot{P}   &= h^{(j)} \qquad
                (j =1,2,3)                                                 \label{E:EOV3} \\
            -\partial_0 \dot{\psi}_0 + \partial^j \dot{\psi}_j &= l^{(0)}               \label{E:EOV4}\\
            \partial_0 \dot{\psi}_j - \partial_j \dot{\psi}_0 &= l^{(j)} \qquad (j =1,2,3)   \label{E:EOV5}\\
              \partial_0 \dot{\phi} &= l^{(4)},                                         \label{E:EOV6}
    \end{align}
    where
    \begin{align}
        \widetilde{U}^0 &\overset{\mbox{\tiny{def}}}{=} (1 + \widetilde{U}^k \widetilde{U}_k)^{1/2}  \\
        \widetilde{\Pi}^{\mu \nu} &\overset{\mbox{\tiny{def}}}{=} \widetilde{U}^{\mu}\widetilde{U}^{\nu} + \Mg^{\mu \nu}    \\
        \widetilde{Q}   &\overset{\mbox{\tiny{def}}}{=} \mathfrak{Q}(\widetilde{{\Ent}},\widetilde{P},\widetilde{\phi})\\
        \widetilde{R}   &\overset{\mbox{\tiny{def}}}{=}
        \mathfrak{R}(\widetilde{{\Ent}},\widetilde{P},\widetilde{\phi}). \label{E:WidetildeR}
    \end{align}
    Here, the function $\mathfrak{Q}$ is defined in \eqref{E:QfunctionDef}, and
    the function $\mathfrak{R}$ is defined in \eqref{E:RfunctionDef}. The unknowns are the components of $\dot{\mathbf{V}} 
    \overset{\mbox{\tiny{def}}}{=} (\dot{{\Ent}},\dot{P},\cdots, \dot{\psi}_2,\dot{\psi}_3)$,
    and we label the first five components of $\dot{\mathbf{V}}$ by
    $\dot{\mathbf{W}} \overset{\mbox{\tiny{def}}}{=} (\dot{{\Ent}},\dot{P},\dot{U}^1,\dot{U}^2,\dot{U}^3).$

    The EOV play multiple roles in this article.
    Except when discussing the space of variations $\dot{\mathbf{V}}$
    as an abstract vector space isomorphic to $\mathbb{R}^{10},$ we use the
    symbol $\dot{\mathbf{V}}$ to represent a quantity that solves the EOV.
    The quantity represented by $\dot{\mathbf{V}}$, the bgs $\widetilde{\mathbf{V}},$ and the inhomogeneous terms will vary
    from application to application, but we will always be clear about their definitions in the relevant sections.

    In the case that we are discussing the linearization of the EN$_{\kappa}$ system around a bgs $\widetilde{\mathbf{V}},$ the 
    inhomogeneous terms take the form
    \begin{align}
        f&=\mathfrak{F}(\widetilde{\mathbf{V}})\overset{\mbox{\tiny{def}}}{=}0 \label{E:Inhomogeneousf} \\
        g&= \mathfrak{G}(\widetilde{\mathbf{V}})\overset{\mbox{\tiny{def}}}{=}(4\widetilde{P} - 3\widetilde{Q})\widetilde{U}^\mu
        \widetilde{\psi}_{\mu} \label{E:Inhomogeneousg}\\
        h^{(j)}&=\mathfrak{H}^{(j)}(\widetilde{\mathbf{V}})\overset{\mbox{\tiny{def}}}{=}(3\widetilde{P}-\widetilde{R}) \widetilde{\Pi}^{\mu j} \widetilde{\psi}_{\mu} \qquad  (j=1,2,3)\label{E:Inhomogeneoushj}\\
        l^{(0)}&=\mathfrak{L}^{(0)}(\widetilde{\mathbf{V}})\overset{\mbox{\tiny{def}}}{=}\kappa^2 \widetilde{\phi} + \widetilde{R} -
        3\widetilde{P}\label{E:Inhomogeneousl0}\\
        l^{(j)}&=\mathfrak{L}^{(j)}(\widetilde{\mathbf{V}})\overset{\mbox{\tiny{def}}}{=}0 \qquad (j =1,2,3)\label{E:Inhomogeneouslj}\\
        l^{(4)}&=\mathfrak{L}^{(4)}(\widetilde{\mathbf{V}})\overset{\mbox{\tiny{def}}}{=}\widetilde{\psi}_0\label{E:Inhomogeneousl4},
    \end{align}
    where $\mathfrak{F},\mathfrak{G},\cdots,\mathfrak{L}^{(4)}$
    are functions of $\widetilde{\mathbf{V}}.$

    It is quite important that the coordinate derivatives of
    solutions to \eqref{E:EOV1} - \eqref{E:WidetildeR} also satisfy \eqref{E:EOV1} -
    \eqref{E:WidetildeR} with different inhomogeneous terms. This may be seen
    by differentiating the equations and relegating all but the principal terms
    to the right-hand side. Similarly, the difference of two solutions to \eqref{E:EOV1} -
    \eqref{E:WidetildeR} also satisfies \eqref{E:EOV1} - \eqref{E:WidetildeR}. Thus, the ``$\cdot$'' is a suggestive placeholder
    that will frequently represent ``derivative" or ``difference" depending on the
    application.

    \begin{notation}
        In reference to the inhomogeneous terms on the right-hand side of \eqref{E:Inhomogeneousf} - \eqref{E:Inhomogeneousl4},
        we often use vector notation including but not limited to
        \begin{align}
            \mathbf{b} &= (f,g,h^{(1)},h^{(2)},h^{(3)})                                 \label{E:bdef} \\
            \mathbf{l} &= (l^{(0)},l^{(1)},l^{(2)},l^{(3)},l^{(4)}).                    \label{E:ldef}
        \end{align}
            When it is convenient, we will use different vector notation to refer to the inhomogeneous terms, but we always use
            the notation $f,g,\cdots,l^{4}$ to refer to the inhomogeneous terms in scalar form; our use of
            notation for the inhomogeneous terms will always be made clear in the relevant sections.
    \end{notation}

            \noindent {\it Terminology:} If $\dot{\mathbf{V}}$ is a solution to the system \eqref{E:EOV1} - \eqref{E:WidetildeR}, we say
            that $\dot{\mathbf{V}}$ is a solution to the EOV
            defined by the bgs $\widetilde{\mathbf{V}}$ with inhomogeneous terms
            $(\mathbf{b},\mathbf{l}).$

    \begin{notation}                                                                                                                                                    \label{N:MatrixNotation}
            We will often find it advantageous to abbreviate the
            ``upper half" of the various systems in this article using matrix notation. For
            example, we sometimes write \eqref{E:EOV1} - \eqref{E:EOV3} as
            \begin{align}
                A^\mu(\widetilde{\mathbf{V}})\partial_\mu \dot{\mathbf{W}} &=\mathbf{b},           \label{E:MatrixDef}
            \end{align}
            where each $A^{\mu}(\widetilde{\mathbf{V}})$ is a $5 \times 5$ matrix
            with entries that are functions of the bgs $\widetilde{\mathbf{V}},$ while $\mathbf{b}$ is
            defined by \eqref{E:bdef}.
            For instance,
            \begin{equation}                                        \label{E:A0Def}
                A^0(\widetilde{\mathbf{V}}) =
                    \begin{pmatrix}
                        \widetilde{U}^0 & 0 & 0 & 0 &0\\
                        0 & \widetilde{U}^0 & \widetilde{Q} \widetilde{U}^1/\widetilde{U}^0 &
                            \widetilde{Q} \widetilde{U}^2/\widetilde{U}^0 &
                            \widetilde{Q} \widetilde{U}^3/\widetilde{U}^0 \\
                        0 & \widetilde{\Pi}^{01} & (\widetilde{R} +
                            \widetilde{P})\widetilde{U}^0 & 0 & 0\\
                            0 & \widetilde{\Pi}^{02} & 0 & (\widetilde{R} +
                            \widetilde{P})\widetilde{U}^0 & 0\\
                        0 & \widetilde{\Pi}^{03} & 0 & 0 &(\widetilde{R} +
                            \widetilde{P})\widetilde{U}^0 \\
                    \end{pmatrix},
            \end{equation}
            and similarly for the $A^k(\widetilde{\mathbf{V}}),$ for $k=
            1,2,3.$
        \end{notation}
        
        \begin{remark}
        		We reserve the use of matrix notation
            for the ``upper half" for two reasons. The first is that the ``lower half" involves
            constant coefficient differential operators, so when differentiating the ``lower half" equations, we don't
            have to worry about commutator terms, which are easily expressed using matrix notation as in
            \eqref{E:kalphaDef}, arising from differential operators acting on the
            coefficients. The second reason is that in future work, we plan to study the ``lower-half" in its original form as
            an inhomogeneous Klein-Gordon equation, but we will
            still use matrix notation for the ``upper-half."
        \end{remark}

       	\begin{remark}                                                              \label{R:Invertible}
            A calculation gives that $\det \big(A^0(\widetilde{\mathbf{V}})\big) =
            -\widetilde{Q}(\widetilde{R}+\widetilde{P})^2(\widetilde{U}^0)^3,$
            and in the Cauchy problem studied below, this formula will
            ensure that $A^0$ is invertible.
        \end{remark}

\section{The Geometry of the EN$_{\kappa}$ System} \label{S:Geometry}

        In this section, we discuss the geometry of the
        characteristics of the EN$_{\kappa}$ system and relate the
        geometry to the speeds of propagation.

        \subsection{The symbol and the characteristic subset of $T^*_x \mathcal{M}$}

        The \emph{symbol} $\sigma_{\xi}$ of the equations of
        variation at a given covector $\xi \in T^*_x \mathcal{M},$ the cotangent
        space of $\mathcal{M}$ at $x,$ is a linear operator
        on the space of variations $\dot{\mathbf{V}}.$ This operator
        is obtained by making the replacements $\partial_\lambda \mathscr{U}
        \longrightarrow \xi_{\lambda} \dot{\mathscr{U}}$ on the left-hand side of the
        system \eqref{E:EOV1} - \eqref{E:EOV6}. Here, $\mathscr{U}$ stands for any of
        the unknowns. The \emph{characteristic
        subset of the cotangent space} at $x$ is defined to be the
        set of all covectors $\xi \in T^*_x \mathcal{M}$ such that
        $\sigma_{\xi}$ has a nontrivial null space. Thus, $\xi$
        lies in the characteristic subset of $T^*_x \mathcal{M}$ iff the following
        \emph{algebraic} system has non-zero solutions $\dot{\mathbf{V}}\subset \mathbb{R}^{10}:$
        \begin{align}
            \widetilde{U}^\mu \xi_{\mu} \dot{{\Ent}} &= 0                                                    \label{E:Characteristics1}\\
            \widetilde{U}^\mu \xi_{\mu} \dot{P} +
            \widetilde{Q}\frac{\widetilde{U}_k}{\widetilde{U}^0}
                \xi_0 \dot{U}^k + \widetilde{Q} \xi_k \dot{U}^k  &= 0                                   \label{E:Characteristics2}\\
            (\widetilde{R} + \widetilde{P})\widetilde{U}^\mu \xi_\mu \dot{U}^j + \widetilde{\Pi}^{\mu j} \xi_\mu \dot{P}
            &= 0 \qquad (j =1,2,3)                                             \label{E:Characteristics3}\\
            \xi_{\mu} \dot{\psi}^{\mu} &= 0                                                   \label{E:Characteristics4}\\
            \xi_0 \dot{\psi}_j - \xi_j \dot{\psi}_0 &=0 \qquad (j=1,2,3)       \label{E:Characteristics5}\\
            \xi_0 \dot{\phi} &= 0.                                                 \label{E:Characteristics6}
        \end{align}

        The determinant of the linear operator $\sigma_{\xi}$ at
        $x,$ known as the \emph{characteristic form} of the
        EOV and denoted by $\mathcal{Q}(x;\xi),$ is given by
        \begin{align}
            \mathcal{Q}(x;\xi) \overset{\mbox{\tiny{def}}}{=}
            (\xi_0)^3\big(\widetilde{U}^{\lambda}\xi_{\lambda}\big)^3 (\widetilde{h}^{-1})^{\mu \nu}\Mg^{\alpha \beta} \xi_{\mu}
            \xi_{\nu}\xi_{\alpha}\xi_{\beta},                                                                       \label{E:NormalCone}
        \end{align}
        where $\widetilde{h}^{-1}$ is the \emph{reciprocal acoustical metric}, a non-degenerate
        quadratic form on $T^*_x \mathcal{M}$ defined by
        \begin{align}
            (\widetilde{h}^{-1})^{\mu \nu} &\overset{\mbox{\tiny{def}}}{=} \widetilde{\Pi}^{\mu \nu} - \widetilde{\sigma}^{-2} \widetilde{U}^{\mu}
            \widetilde{U}^{\nu}= \Mg^{\mu \nu} - \left(\widetilde{\sigma}^{-2} - 1\right) \widetilde{U}^{\mu}
            \widetilde{U}^{\nu}, \\
                \widetilde{\sigma} &\overset{\mbox{\tiny{def}}}{=}\mathscr{S}(e^{-4\widetilde{\phi}}\widetilde{P},\widetilde{{\Ent}}),
        \end{align}
        and the function $\mathscr{S}$ is defined by \eqref{E:Sigma}.
        The characteristic subset of $T_x^*$
        is therefore equal to the level set
        \begin{align}
            \lbrace \xi \in T^*_x \mathcal{M}|\mathcal{Q}(x;\xi)=0 \rbrace.
        \end{align}
        Consequently, $\xi$ is an element of the characteristic subset of $T^*_x \mathcal{M}$ iff one of the following four
        conditions holds:
        \begin{align}
            \xi_{\mu} \widetilde{U}^{\mu} &= 0                                                  \label{E:CotCharSubset1}\\
            (\widetilde{h}^{-1})^{\mu \nu} \xi_{\mu} \xi_{\nu} &= 0                              \label{E:CotCharSubset2}\\
            \Mg^{\mu \nu} \xi_{\mu} \xi_{\nu} &= 0                                 \label{E:CotCharSubset3}\\
            \xi_0 & = 0.                                                            \label{E:CotCharSubset4}
        \end{align}

        Condition \eqref{E:CotCharSubset1} defines a plane $P^*_{x,\widetilde{U}}$ in $T^*_x
        \mathcal{M},$ while conditions \eqref{E:CotCharSubset2} and
        \eqref{E:CotCharSubset3} define cones $C^*_{x,s(ound)}$ and $C^*_{x,l(ight)},$ respectively, in $T^*_x
        \mathcal{M}.$ Condition \eqref{E:CotCharSubset4} also defines a
        plane $P^*_{x,0}$ in $T^*_x \mathcal{M},$ and its presence is a consequence of our
        choice of $\partial_t \phi$ as a state-space variable in our rewriting of the linear Klein-Gordon
        equation as a first order system. We refer to \eqref{E:CotCharSubset1} - \eqref{E:CotCharSubset4} as the four
        \emph{sheets}
        of the characteristic subset of $T^*_x \mathcal{M}.$
        Fig. \ref{Fi:CotCharSubset} illustrates the characteristic subset of $T^*_x \mathcal{M}.$ In the illustration, we
        masquerade as if the domain of solutions to the EOV is $R^{1+2},$ with the vertical direction representing positive
        values of $\xi_0.$
        \begin{figure}[htp]
        		\centering
            \includegraphics{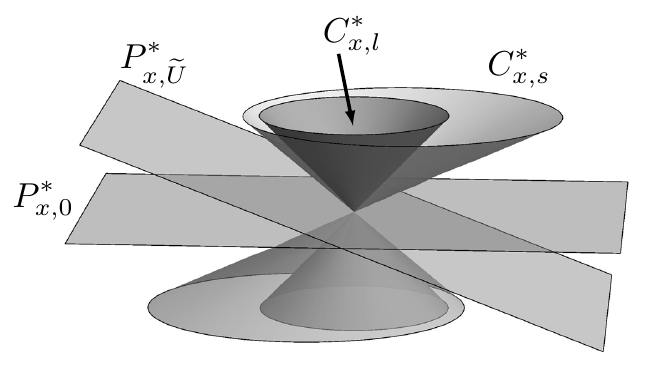}
            \caption{The Characteristic Subset of $T^*_x
            \mathcal{M}$}\label{Fi:CotCharSubset}
        \end{figure}

        \subsection{Characteristic surfaces and the
        characteristic subset of $T_x \mathcal{M}$} \label{SS:TangentSpaceCharactersisticSubset}
        A $C^1$ surface $S \subset \mathcal{M}$ that is given as a level set of a function $\Phi$
        is said to be a \emph{characteristic surface} if at each point
        $x \in S,$ the covector $\xi$ with components
        $\xi_{\nu} = \partial_{\nu} \Phi \ \mbox{for} \ \nu=0,1,2,3,$ is an
        element of the characteristic subset of $T^*_x \mathcal{M}.$ It is well-known (consult e.g. \cite{rCdH1966}) that jump 
        discontinuities
        in weak solutions can occur across characteristic surfaces, and
        that characteristic surfaces play a role in determining
        a domain of influence of a region of spacetime.

        There is an alternative characterization of
        characteristic surfaces in terms of the duals of the sheets $P^*_{x,\widetilde{U}},$ $P^*_{x,0},$ $C^*_{x,s},$
        and $C^*_{x,l}.$ The notion of duality we refer to is as follows
        (consult e.g. \cite{rCdH1966}): To each covector $\xi$ in the
        characteristic subset of $T_x^* \mathcal{M}$ there
        corresponds the \emph{null space} of $\xi,$ which we denote by $N_{\xi}.$
        This 3-dimensional plane is a subset of $T_x \mathcal{M}$,
        the tangent space of $\mathcal{M}$ at $x,$ and is described in coordinates as
        $N_{\xi}\overset{\mbox{\tiny{def}}}{=}\lbrace X \in T_x
        \mathcal{M} | \xi_{\mu} X^{\mu}=0 \rbrace.$ We define the dual to a sheet
        of the characteristic subset of $T_x^* \mathcal{M}$ to be the envelope
        in $T_x \mathcal{M}$ generated by the $N_{\xi}$ as $\xi$ varies over the sheet. The \emph{characteristic
        subset of the tangent space} at $x$ is defined to be the union of the duals to
        the sheets \eqref{E:CotCharSubset1} - \eqref{E:CotCharSubset4}. A calculation of the envelopes implies that the respective
        duals to \eqref{E:CotCharSubset1}, \eqref{E:CotCharSubset2},
        \eqref{E:CotCharSubset3}, and \eqref{E:CotCharSubset4} are the sets of $X \in T_x \mathcal{M}$ such that
        in our fixed rectangular coordinate system (see Section \ref{SS:SpacetimeConventions}),

        \begin{align}
            X &= \lambda \widetilde{U} \ \mbox{for some} \ \lambda \in \mathbb{R}  \label{E:TanCharSubset1} \\
            \widetilde{h}_{\mu \nu}X^{\mu}X^{\nu}&=0                             \label{E:TanCharSubset2} \\
            \Mg_{\mu \nu}X^{\mu}X^{\nu}&=0                        \label{E:TanCharSubset3}\\
            X &= \lambda (1,0,0,0) \ \mbox{for some} \ \lambda \in \mathbb{R}, \label{E:TanCharSubset4}
        \end{align}
        where

        \begin{align}
            \widetilde{h}_{\mu \nu} \overset{\mbox{\tiny{def}}}{=} \Mg_{\mu \nu} + (1- \widetilde{\sigma}^2)\widetilde{U}_{\mu}\widetilde{U}_{\nu}
        \end{align}
        is the \emph{acoustical metric}, a non-degenerate quadratic
        form on $T_x \mathcal{M}.$ The dual to $P^*_{x,\widetilde{U}},$
        given by \eqref{E:TanCharSubset1}, is the linear span of $\widetilde{U},$ and
        the dual to the plane $P^*_{x,0},$ given by \eqref{E:TanCharSubset4}, is the linear
        span of $(1,0,0,0).$ The dual to $C^*_{x,s},$ given by \eqref{E:TanCharSubset2} and
        labeled as $C_{x,s},$ is the sound cone in $T_x \mathcal{M},$
        while the dual to $C^*_{x,l},$ given by \eqref{E:TanCharSubset3} and
        labeled as $C_{x,l},$ is the light cone in $T_x \mathcal{M}.$
        We refer to these subsets of $T_x \mathcal{M}$ as the four sheets of the
        characteristic subset of the $T_x \mathcal{M}$ (noting that
        the degenerate cases \eqref{E:TanCharSubset1} and \eqref{E:TanCharSubset4} are
        lines rather than ``sheets"). See Fig. \ref{Fi:TanCharSubset}
        for the picture in $\mathbb{R}^{1+2},$ where the
        vertical direction represents positive values of $X^0.$

        \begin{figure}[htp]
        		\centering
            \includegraphics{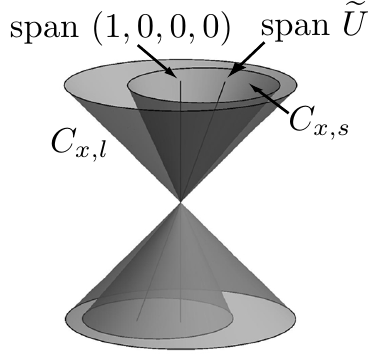}
            \caption{The Characteristic Subset of $T_x \mathcal{M}$} \label{Fi:TanCharSubset}
        \end{figure}

        It follows from the above description that for each $\xi$ belonging to
        a fixed sheet of the characteristic subset of $T_x^* \mathcal{M},$
        $N_{\xi}$ is tangent to the corresponding sheet of the characteristic subset
        of $T_x \mathcal{M}.$ Therefore, we may equivalently define a
        characteristic surface as a $C^1$ surface $S$ such that the tangent
        plane at each of its points $x$ is tangent to any of the
        four sheets of the characteristic subset of $T_x \mathcal{M}.$

        \begin{remark}
            Note that $C_{x,s}$ lies inside $C_{x,l},$ but
            $C_{x,l}^*$ lies inside $C_{x,s}^*.$
        \end{remark}

        \subsection{Inner characteristic core, strict hyperbolicity, spacelike surfaces}    \label{SS:Hyperbolicity}
        The \emph{inner characteristic core} of the cotangent space at
        $x,$ denoted $\mathcal{I}^*_x,$ is the subset of $T_x^* \mathcal{M}$ lying \emph{strictly} inside
        the innermost sheet $C^*_{x,l}.$ $\mathcal{I}^*_x$ comprises two components,
        and we refer to the component such that each co-vector $\xi$ belonging to it
        has $\xi_0 > 0$ as the \emph{positive component,} denoted by $\mathcal{I}^{*+}_x:$
        \begin{align}
            \mathcal{I}^{*+}_x \overset{\mbox{\tiny{def}}}{=} \lbrace \xi \in T^*_x \mathcal{M} | \xi_{\mu} \xi^{\mu} < 0
            \ \mbox{and} \ \xi_0 > 0 \rbrace.                                                   \label{E:ICCPC}
        \end{align}

        A co-vector $\xi \in T^*_x \mathcal{M}$ is said to be
        \emph{hyperbolic for $\mathcal{Q}$ at $x$} iff for any co-vector
        $\upsilon$ not parallel to $\xi,$ $\mathcal{Q}(x;\lambda \xi + \upsilon)=0$
        has real roots in $\lambda,$ where $\mathcal{Q}$ is given
        in \eqref{E:NormalCone}. The set of hyperbolic co-vectors at
        $x$ is equal to $\mathcal{I}^*_x;$ see Fig. \ref{Fi:CotCharSubset}. A co-vector $\xi \in T^*_x \mathcal{M}$ is said to be
        \emph{strictly hyperbolic}\footnote{For PDEs derivable from a Lagrangian, the notions of hyperbolicity,
        characteristic subsets, etc., have been generalized by
        Christodoulou \cite{dC2000} in a manner that allows one to handle characteristic
        forms that feature multiple roots.} \emph{for $\mathcal{Q}$ at $x$} iff for any co-vector
        $\upsilon$ not parallel to $\xi,$ $\mathcal{Q}(x;\lambda \xi + \upsilon)=0$
        has \emph{distinct} real roots in $\lambda.$ As mentioned in Section \ref{S:Introduction},
        the EOV (and hence the EN$_{\kappa}$ system) are (is) not strictly hyperbolic because
        of the repeated factors in the expression \eqref{E:NormalCone} for $\mathcal{Q}(x;\cdot),$
        and because two of the sheets of the characteristic subset of $T^*_x \mathcal{M}$ intersect.

        A $C^1$ surface $S \subset \mathcal{M}$ is said to be \emph{spacelike} (with respect
        to the light cones $C^*_{x,l}$) if at each $x \in S,$
        there is a co-vector $\xi$ belonging to $\mathcal{I}^*_x$
        such that the tangent plane to $S$ at $x$ is equal to
        $N_{\xi}.$ Based on the discussion above,
        it follows that $S$ is spacelike at $x$ iff the tangent plane to $S$ at $x$
        is the null space of a co-vector $\xi$ that is hyperbolic for $\mathcal{Q}$ at $x.$

        \subsection{Speeds of propagation}                                                \label{SS:Speeds of Propagation}

        It is well-known that for first order symmetric hyperbolic systems,
        the speeds of propagation are locally governed by the
        characteristic subsets. For example, in the case that the
        characteristic subset of $T^*_x \mathcal{M}$ at each
        $x$ includes an innermost sheet, the \emph{domain of influence} of a
        spacetime point $x'$ is contained in the interior of the forward conoid in $\mathcal{M}$ traced out
        by the set of all curves emanating from $x'$ and remaining
        tangent to the sheets of the characteristic subsets of the $T_x
        \mathcal{M}$ that are dual to the innermost sheets of the characteristic subsets
        of the $T_x^* \mathcal{M}$ as the curve parameter varies; consult \cite{pL2006}
        for a detailed discussion of this fact.

        We will later illustrate the occurrence of similar phenomena in the EN$_{\kappa}$ system. In
        this case, the innermost sheet at $x$ is $C^*_{x,l},$ the
        dual of which is $C_{x,l},$ the light cone in $T_x \mathcal{M}.$ Therefore, the forward conoid emanating from a spacetime point $x'$
        is the forward light cone in $\mathcal{M}$ with vertex at $x'.$ Thus, one would expect that the fastest speed of
        propagation in the EN$_{\kappa}$ system is the speed of light. This claim is given rigorous meaning
        below in the uniqueness argument (see Section \ref{SS:UniqunessandHNMinusOneDependence}) which shows, for example, that a
        solution that is constant in the Euclidean sphere of radius $r$ centered at the point $\snew \in \mathbb{R}^3$ at $t=0$ remains
        constant in the Euclidean sphere of radius $r - t$ centered at $\snew$ at time $t > 0;$ see Remark \ref{R:DOI}.

        We contrast this to the case of the special-relativistic Euler equations without gravitational interaction, in which there is 
        no Klein-Gordon equation governing the
        propagation of gravitational waves at the speed of light, and the set
        $C^*_{x,l}$ does not belong to the characteristic subset of $T_x^* \mathcal{M}.$
        The inner sheet at $x$ in this case is $C^*_{x,s},$ the dual of which is $C_{x,s},$ the sound cone
        in $T_x \mathcal{M},$ and the methods applied below can be used to
        show that the fastest local speed of propagation is dictated by the sound cones
        $C_{x,s}.$ This case is studied in detail in \cite{dC2000} and \cite{dC2007}.

    \subsection{Energy currents} \label{SS:EnergyCurrents}

        The role of energy currents in the well-posedness proof
        is to replace the energy principle available for symmetric
        hyperbolic systems. After providing the definition of an energy current, we
        illustrate its two key properties, namely that it has the positivity property
        \eqref{E:PositiveDefinite} below, and that its divergence is lower order in the variation
        $\dot{\mathbf{V}}.$

        \subsubsection{The definition of an energy current}

        Given a variation $\dot{\mathbf{V}}: \mathcal{M} \rightarrow \mathbb{R}^{10}$
        and a bgs $\widetilde{\mathbf{V}}: \mathcal{M} \rightarrow \mathbb{R}^{10}$
        as defined in Section \ref{SS:EOV},
        we define the energy current to be the vectorfield $\dot{J}$ with
        components $\dot{J}^0,$ $\dot{J}^{j},$ $j=1,2,3,$ in the global rectangular coordinate system given by
        \begin{align}                                                                                       
            \hspace{-2cm} \dot{J}^0 \overset{\mbox{\tiny{def}}}{=} \widetilde{U}^0 \dot{{\Ent}}^2  & +
                    \frac{\widetilde{U}^0}{\widetilde{Q}}\dot{P}^2 +
                    2\frac{\widetilde{U}_k\dot{U}^k}{\widetilde{U}^0} \dot{P}
                 + (\widetilde{R} + \widetilde{P})\widetilde{U}^0 \Big[\dot{U}^k
                    \dot{U}_k - \frac{(\widetilde{U}_k
                    \dot{U}^k)^2}{(\widetilde{U}^0)^2} \Big] \notag \\
                 & + \frac{1}{2} \big[(\dot{\phi})^2
                    + (\dot{\psi}_0)^2 + (\dot{\psi}_1)^2 + (\dot{\psi}_2)^2 + (\dot{\psi}_3)^2
                    \big] \notag
        \end{align}
            \begin{align}
            \dot{J}^j &\overset{\mbox{\tiny{def}}}{=}\widetilde{U}^j \dot{{\Ent}}^2  +
                    \frac{\widetilde{U}^j}{\widetilde{Q}} \dot{P}^2 + 2 \dot{U}^j \dot{P}
                 + (\widetilde{R} + \widetilde{P})\widetilde{U}^j \Big[\dot{U}^k
                    \dot{U}_k - \frac{(\widetilde{U}_k \dot{U}^k)^2}{(\widetilde{U}^0)^2} \Big]  -
                    \dot{\psi}_0 \dot{\psi}_j . \label{E:EnergyCurrent}
                \end{align}

        \begin{notation}
            In an effort to avoid cluttering the notation,
            we sometimes suppress the direct dependence of $\dot{J}$ on $\dot{\mathbf{V}}$
            and $\widetilde{\mathbf{V}}$ and instead emphasize the indirect dependence of
            $\dot{J}$ on $(t,\snew)$ through $\dot{\mathbf{V}}$ and $\widetilde{\mathbf{V}}$ by writing
            $``\dot{J}(t,\snew)."$
        \end{notation}

            \noindent {\it Terminology:} We say that $\dot{J}$ is the energy current for the variation
            $\dot{\mathbf{V}}$ with coefficients defined by the bgs $\widetilde{\mathbf{V}}.$

        \begin{remark}
            The theory of hyperbolic PDEs derivable from a Lagrangian, and in particular the derivation of energy
            currents, is developed by Christodoulou in \cite{dC2000}. For readers interested in studying
            Christodoulou's techniques, we remark that the Lagrangian density for
            \eqref{E:ENkappa1} - \eqref{E:ENkappa3} (the
            first 5 scalar equations of the EN$_{\kappa}$ system) is
            expressed in the original variables as $\rho
            e^{4 \phi}.$ The energy current \eqref{E:EnergyCurrent} is the sum of an energy
            current for the linear Klein-Gordon equation, which supplies
            the terms involving $(\dot{\phi})^2$ and $(\dot{\psi_{\nu}})^2,$ and an energy current used by Christodoulou in
            \cite{dC2007} to study the special-relativistic Euler equations without gravitational interaction.
        \end{remark}

    \subsubsection{The positive definiteness of $\xi_{\mu}\dot{J}^{\mu}$ for $\widetilde{P}>0$ and $\xi \in \mathcal{I}^{*+}_x$}

        Given an energy current as defined by \eqref{E:EnergyCurrent} and a co-vector $\xi \in T^*_x
        \mathcal{M},$ the quantity $\xi_{\mu}\dot{J}^{\mu}$ may be viewed as a quadratic form in
        the variations $\dot{\mathbf{V}}$
        with coefficients defined by the bgs $\widetilde{\mathbf{V}}.$
        We emphasize this quadratic dependence on the variations by writing
        $\xi_{\mu}\dot{J}^{\mu}(\dot{\mathbf{V}},\dot{\mathbf{V}}).$ One of the two key
        features of the energy current is that $\widetilde{P}>0$ and $\xi \in
        \mathcal{I}^{*+}_x$ together imply that the form $\xi_{\mu}\dot{J}^{\mu}(\dot{\mathbf{V}},\dot{\mathbf{V}})$
        is positive definite in $\dot{\mathbf{V}}:$
        \begin{equation}                                                                                    \label{E:PositiveDefinite}
           \xi_\mu \dot{J}^\mu(\dot{\mathbf{V}},\dot{\mathbf{V}})> 0 \  \mbox{if} \ \xi \in
            \{ \zeta \in T^*_x(\mathcal{M}) \ | \ \zeta_\mu \zeta^\mu < 0 \ \mbox{and} \ \zeta_0 > 0
            \} \ \mbox{and} \ \dot{\mathbf{V}} \neq \mathbf{0}.
        \end{equation}
        A direct verification of this fact can be carried out, for example,
        by calculating the eigenvalues of the matrix of the quadratic form 
        $\xi_{\mu}\dot{J}^{\mu}(\dot{\mathbf{V}},\dot{\mathbf{V}}).$ The
        eigenvalues depend on $\xi$ and are positive whenever $\widetilde{P}>0$ and $\xi \in \mathcal{I}^{*+}_x.$
        As we shall soon see, inequality \eqref{E:PositiveDefinite} will allow us to use the form
        $\xi_{\mu}\dot{J}^{\mu}(\dot{\mathbf{V}},\dot{\mathbf{V}})$
        to estimate the $L^2$ norms of the variations, provided that we
        estimate the bgs $\widetilde{\mathbf{V}}.$

        \begin{remark}
                Although later in this article we make use of the fact that $\dot{\mathbf{V}}$ is a
                solution to the EOV, the inequality in \eqref{E:PositiveDefinite} does
                not rely on this fact; it is an algebraic statement about $\xi_{\mu} \dot{J}^{\mu}(\dot{\mathbf{V}},\dot{\mathbf{V}})$
                viewed as a quadratic form on $\mathbb{R}^{10}.$
        \end{remark}

    \subsubsection{The divergence of the energy current}
        If the variations $\dot{\mathbf{V}}$ are solutions of the EOV, then we can compute $\partial_\mu \dot{J}^\mu$ and use the 
        equations \eqref{E:EOV1} - \eqref{E:EOV6} for substitution to eliminate the terms containing the derivatives of 
        $\dot{\mathbf{V}}:$

        \begin{align}
            \begin{split}                                                                           \label{E:EnergyCurrentDivergence}
                \partial_\mu \dot{J}^\mu & =
                    (\partial_\mu \widetilde{U}^\mu) \dot{{\Ent}}^2 +
                    \partial_\mu \left(\frac{\widetilde{U}^\mu}{\widetilde{Q}}\right) \dot{P}^2
                   +  2 \partial_0 \left(\frac{\widetilde{U}_k}{\widetilde{U}^0}\right) \dot{U}^k \dot{P}                \\
                &  + \partial_\mu[(\widetilde{R} + \widetilde{P})\widetilde{U}^\mu]
                    \left[\dot{U}^k \dot{U}_k - \frac{(\widetilde{U}_k
                    \dot{U}^k)^2}{(\widetilde{U}^0)^2} \right]
                    - 2\widetilde{U}_k \dot{U}^k(\widetilde{R}
                    + \widetilde{P})\left(\frac{\widetilde{U}^\mu}{\widetilde{U}^0}\right)
                    \partial_\mu \left(\frac{\widetilde{U}_j}{\widetilde{U}^0}\right)\dot{U}^j                  \\
                &   + 2 \dot{{\Ent}}f + 2\frac{\dot{P}g}{\widetilde{Q}} + 2
                    \dot{U}_k h^{(k)} - 2 \frac{\widetilde{U}_j h^{(j)} \widetilde{U}_k
                    \dot{U}^k}{(\widetilde{U}^0)^2} - \ {\dot{\psi}}_0 l^{(0)}  \ + \ {\dot{\psi}}_k l^{(k)} \ + \ \dot{\phi}l^{(4)} .
            \end{split}
        \end{align}
        That the right-hand side of \eqref{E:EnergyCurrentDivergence} does not contain any derivatives of the
        variations is the second key property announced at the beginning of Section \ref{SS:EnergyCurrents}.

        \begin{remark}                                                              \label{R:HigherOrderEnergyCurrents}
        Given a spatial derivative multi-index $\vec{\alpha}$ and
        an energy current $\dot{J}$ as defined in
        \eqref{E:EnergyCurrent} such that the variation $\dot{\mathbf{V}}$ is a solution of
        \eqref{E:EOV1} - \eqref{E:WidetildeR} with inhomogeneous terms $(\mathbf{b},\mathbf{l})$, where $\mathbf{b}$ and
        $\mathbf{l}$ are defined by \eqref{E:bdef} and \eqref{E:ldef} respectively,
        we define the \emph{higher-order energy current} $\dot{J}_{\vec{\alpha}}$ to be the energy
        current for the variation $\partial_{\vec{\alpha}}\dot{\mathbf{V}}$ with
        coefficients defined by the same bgs $\widetilde{\mathbf{V}}.$ The
        variations $\partial_{\vec{\alpha}}\dot{\mathbf{V}}$ are solutions
        of \eqref{E:EOV1} - \eqref{E:WidetildeR} with inhomogeneous terms $(\mathbf{b}_{\vec{\alpha}},
        \partial_{\vec{\alpha}}\mathbf{l}),$ where $\mathbf{b}_{\vec{\alpha}}$ is defined
        in terms of $\mathbf{b}$ below through \eqref{E:balphaDef}.
        Consequently, the expression for $\partial_\mu \dot{J}^\mu_{\vec{\alpha}}$ is
        given by taking the formula \eqref{E:EnergyCurrentDivergence}
        for $\partial_\mu \dot{J}^{\mu}$ and making the replacements
        $\dot{\mathbf{V}} \rightarrow
        \partial_{\vec{\alpha}}\dot{\mathbf{V}}$ and
        $(\mathbf{b},\mathbf{l}) \rightarrow
        (\mathbf{b}_{\vec{\alpha}},\partial_{\vec{\alpha}}\mathbf{l}).$
        \end{remark}

\section{Assumptions on the Initial Data}                                     \label{S:IVP}
        We now describe a class of initial data to which
        the energy methods for showing well-posedness can be applied.
        The Cauchy surface we consider is $\{(t,\snew) \in \mathcal{M} \ | \ t=0\}.$
        
        \subsection{An $H^N$ perturbation of a quiet fluid}
        The initial data for the EN$_{\kappa}$ system are denoted by
        $\mathring{\mathbf{V}}=\mathring{\mathbf{V}}(\snew)\overset{\mbox{\tiny{def}}}{=}(\mathring{{\Ent}},\mathring{P},
        \mathring{U}^1,\cdots,\mathring{\psi}_3),$
        where $\mathring{\psi}_j \overset{\mbox{\tiny{def}}}{=} \partial_j \mathring{\phi}$ for $j=1,2,3.$ We assume that the initial
        data $\mathring{\mathbf{V}}$ for the EN$_{\kappa}$ system are constructed from initial data
        $(\mathring{{\Ent}},\mathring{p}, \mathring{u}^1,\cdots,\mathring{\psi}_3)$
        in the original state-space variables $({\Ent},p,u^1,\cdots, \psi^3)$ according to the substitutions \eqref{E:U}, \eqref{E:R},
        and \eqref{E:P}. Additionally, we assume that outside of the unit ball centered at the origin in the Cauchy surface
        
        \begin{equation}                                                                                \label{E:InitialConstant}
            \mathring{\mathbf{V}} \equiv \bar{\mathbf{V}}\overset{\mbox{\tiny{def}}}{=}(\bar{{\Ent}},\bar{P},0,0,0,\bar{\phi},0,0,0,0),
        \end{equation}
        where $\bar{\phi}$ is the unique solution to
        \begin{align}                                                                               \label{E:phiBar}
            {\kappa}^2 \bar{\phi} + e^{4\bar{\phi}} \left(\mathscr{R}(\bar{p},\bar{{\Ent}}) - 3\bar{p}\right) = 0,
        \end{align}
        $\bar{{\Ent}}$ and $\bar{p}$ are positive constants
        denoting the initial entropy and pressure of the fluid outside of
        the unit ball, $\bar{P}\overset{\mbox{\tiny{def}}}{=}e^{4\bar{\phi}} \bar{p},$ and the function $\mathscr{R}$ is defined in 
        \eqref{E:rho}. An initial
        state of this form is a perturbation of an infinitely extended quiet fluid, such that
        the perturbation is initially contained in the unit ball. Here we
        need the cosmological constant $\kappa^2 >0$ in order to ensure that the
        EN$_{\kappa}$ system has non-zero constant solutions of the form $\bar{\mathbf{V}}.$
        
        Because the standard energy methods require that the initial data belong to a
        Sobolev space of high enough order, we assume that
        \begin{align}
            \|\mathring{{\Ent}} - \bar{{\Ent}} \|_{H^N} + \|\mathring{p} -
            \bar{p}\|_{H^N} + \|\mathring{u}^k \|_{H^N} + \| \mathring{\phi} -
            \bar{\phi}\|_{H^{N+1}} + \|\mathring{\psi}_0 \|_{H^N} < \infty,              \label{E:IDSobolevConditions}
        \end{align}
        where $N \in \mathbb{N}$ satisfies
        \begin{align}
            N \geq 3. \                                                                                                                                         \label{E:Ndef}
        \end{align}
        Note that \eqref{E:IDSobolevConditions} implies that $\|\mathring{\psi}_j \|_{H^N}
        < \infty \ (j=1,2,3).$ By Proposition \ref{P:CompositionProductSobolevMoser} and Remark \ref{R:SobolevCalculusRemark},
        it follows from \eqref{E:IDSobolevConditions} that
        \begin{align}
            \|\mathring{\mathbf{V}}\|_{H_{\Vb}^N} < \infty.
        \end{align}

        \begin{remark}
        It is not necessary to assume that the initial deviation from the constant state has compact
        support. It is sufficient to consider initial data $\mathring{\mathbf{V}}$
        that differ from $\bar{\mathbf{V}}$ by a perturbation belonging to $H^N,$ such that that
        $\mathring{\mathbf{V}}(\mathbb{R}^3)$ is contained in a compact subset of $\mathcal{O},$ where $N$ is given by \eqref{E:Ndef}
        and $\mathcal{O}$ is defined in Section \ref{SS:AdmissibleStateSpace}. We make the compactness assumption
        because it is useful for illustrating the speeds of propagation as discussed in Section \ref{SS:Speeds of Propagation},
        and because we plan to make use of this setup in future work.
				\end{remark}
        
        \subsection{The admissible subset of state space and the uniform positive definiteness of $\dot{J}^{0}$} \label{SS:AdmissibleStateSpace}
        
        In this section we discuss a further positivity restriction that we place on the initial data.
        We will see in Section \ref{SSS:UniformTime} that this positivity condition is propagated for short times 
        during an iterative construction of solutions to the
        linearized EN$_{\kappa}$ system. Since it plays a key role in our future 
        analysis, we discuss here the implications of this positivity
        restriction regarding the uniform positive definiteness of the energy current,
        viewed as a quadratic form in the variations.
            
            \subsubsection{The definition of the admissible subset of state-space}

						In order to avoid studying the free boundary problem and in order
            to avoid singularities in the energy current, we assume that
            the initial pressure, energy density, and speed of sound are
            uniformly bounded from below by a positive constant. According
            to our assumptions \eqref{E:EOSAssumptions} on the equation of state, to 
            satisfy these requirements, it is
            sufficient to consider initial data for the EN$_{\kappa}$ system such that $\mathring{\mathbf{V}}(\mathbb{R}^3)$
            is contained in a compact subset of the following
            open subset $\mathcal{O}$ of the state-space $\mathbb{R}^{10},$ the
            \emph{admissible subset of state-space:}

            \begin{equation}                                                                \label{E:AdmissibleSubset}
                \mathcal{O}= \{\mathbf{V} \in \mathbb{R}^{10}|\Ent >0, P > 0\}.
            \end{equation}
            We therefore assume that $\Vb \in \mathcal{O}_1$ and
            $\mathring{\mathbf{V}}(\mathbb{R}^3) \subset \mathcal{O}_1,$
            where $\mathcal{O}_1$ is a precompact open set with $\bar{\mathcal{O}}_1
            \Subset \mathcal{O}.$ We then fix a precompact open subset
            $\mathcal{O}_2$ with convex\footnote{Proposition
            \ref{P:SobolevTaylor} requires the convexity of $\bar{\mathcal{O}}_2.$ Without loss of
            generality, we may choose it to be a cube.} closure satisfying $\bar{\mathcal{O}}_1 \Subset \mathcal{O}_2
            \subset \bar{\mathcal{O}}_2 \Subset \mathcal{O};$ our
            goal is to show the existence of a solution that remains in $\bar{\mathcal{O}}_2$ for short times.

            \subsubsection{The uniform positive definiteness of $\dot{J}^{0}$} \label{SS:JdotUniformPositiveDefinite}

            Most of the technical exposition below is devoted to obtaining control over $\|\dot{\mathbf{V}}(t)\|_{H^N},$
            where $\dot{\mathbf{V}}$ is a solution to the EOV defined by a bgs
            $\widetilde{\mathbf{V}}.$ Instead of trying
            to estimate $\|\dot{\mathbf{V}}(t)\|_{L^2}$ directly, it is advantageous to
            estimate $\|\dot{J}^{0}(t)\|_{L^1},$ where $\dot{J}$ is
            an energy current for $\dot{\mathbf{V}}$ with coefficients defined by the bgs
            $\widetilde{\mathbf{V}},$
            since the divergence of $\dot{J}$ is lower order in $\dot{\mathbf{V}}.$ A similar remark applies to estimating 
            $\|\partial_{\vec{\alpha}}\dot{\mathbf{V}}\|_{L^2}$ using higher-order energy currents $\dot{J}_{\vec{\alpha}}.$
          	We shall see that $\|\dot{J}^{0}(t)\|_{L^1}$ can be used to
            estimate $\|\dot{\mathbf{V}}(t)\|^2_{L^2}$ from above and below provided that
            $\dot{J}^{0}$ is uniformly positive definite independent of the bgs
            $\widetilde{\mathbf{V}}.$ More precisely, we claim that there exists
            a $C_{\bar{\mathcal{O}}_2}$ with $0 < C_{\bar{\mathcal{O}}_2} < 1$ such that for any variation
            $\dot{\mathbf{V}}$ and any bgs $\widetilde{\mathbf{V}}$
            contained in $\bar{\mathcal{O}}_2,$ we have
            \begin{equation}                                                                        \label{E:UniformPositivity}
                C_{\bar{\mathcal{O}}_2}|\dot{\mathbf{V}}|^2 \leq
                \dot{J}^0(\dot{\mathbf{V}},\dot{\mathbf{V}})
                \leq \frac{1}{C_{\bar{\mathcal{O}}_2}}|\dot{\mathbf{V}}|^2.
            \end{equation}
            To prove \eqref{E:UniformPositivity}, recall that $\dot{J}$ is
            defined by \eqref{E:EnergyCurrent} and note
            that $(1,0,0,0) \in \mathcal{I}^{*+}_x$ by \eqref{E:PositiveDefinite}. The uniform continuity of
            $\dot{J}$ (which we momentarily view as a function of
            $(\widetilde{\mathbf{V}},\dot{\mathbf{V}})$)
            on the compact set $\bar{\mathcal{O}}_2 \times \lbrace |\dot{\mathbf{V}}|=1 \rbrace$ implies that
            there exists a $C_{\bar{\mathcal{O}}_2}$ with $0 < C_{\bar{\mathcal{O}}_2} < 1$
            such that \eqref{E:UniformPositivity} holds whenever $\widetilde{\mathbf{V}}(t,\snew) \in \bar{\mathcal{O}}_2$ and
            $|\dot{\mathbf{V}}|=1.$ Since the inequalities in \eqref{E:UniformPositivity} are invariant under any rescaling of
            $\dot{\mathbf{V}},$ it follows that we may remove the restriction $|\dot{\mathbf{V}}|=1.$

    \section{The Well-Posedness Theorems}                                                                   \label{S:WellPosedness}
    In this section, we state and indicate how to prove our two main theorems. We have separated
    the proof of well-posedness into two theorems since the techniques used in proving each are
    different. Statements of the technical estimates involving the Sobolev-Moser calculus have been
    placed in the Appendix so as to not interrupt the flow of the main argument.

		\newpage

    \begin{theorem} {\bf(Local Existence and Uniqueness)}                       \label{T:LocalExistence}
        Let $\mathring{\mathbf{V}}(\snew)$ be initial data for the EN$_{\kappa}$ system
        \eqref{E:ENkappa1} - \eqref{E:ENkappa9}
        that are subject to the conditions described in Section \ref{S:IVP}. Then there
        exists a $T > 0$ such that \eqref{E:ENkappa1} - \eqref{E:ENkappa9}
        has a unique classical solution $\mathbf{V}(t,\snew)$ on
        $[0,T] \times \mathbb{R}^3$ satisfying $\mathbf{V}(0,\snew)=\mathring{\mathbf{V}}(\snew).$ The solution is of the form $\mathbf{V}
        =({\Ent},P,U^1,U^2,U^3,\phi,\partial_0 \phi,\partial_1
        \phi,\partial_2 \phi,\partial_3 \phi)$ and satisfies \\ $\mathbf{V}([0,T] \times \mathbb{R}^3) \subset \bar{\mathcal{O}}_2.$
        Furthermore, \\ $\mathbf{V} \in C_b^1([0,T] \times \mathbb{R}^3) \cap 
        C^0([0,T],H_{\Vb}^{N}) \cap C^1([0,T],H_{\Vb}^{N-1}),$
        and consequently \\ $\phi \in C_b^2([0,T] \times \mathbb{R}^3) \cap C^0([0,T],H_{\bar{\phi}}^{N+1}) \cap
        C^1([0,T],H_{\bar{\phi}}^{N}) \cap C^2([0,T],H_{\bar{\phi}}^{N-1}).$
    \end{theorem}
    
    \begin{proof}
    	As discussed in Section \ref{SSS:ProofDescription}, our abbreviated proof of Theorem \ref{T:LocalExistence} is located in 
    	Section \ref{SS:AbbreviatedProof}.
  	\end{proof}
  	
  	\begin{remark}
        In the discussion below, we sometimes denote the solution from Theorem
        \ref{T:LocalExistence} by $\mathbf{V}_{sol}$ for clarity.
    \end{remark}

    \begin{corollary}                                      \label{C:ExistenceInterval}
        The interval of
        existence $[0,T]$ supplied by the Theorem \ref{T:LocalExistence} depends only on
        the set $\bar{\mathcal{O}}_2$ from Section \ref{S:IVP},
        $\|{^{(0)}\mathring{\mathbf{V}}}\|_{H_{\Vb}^{N+1}},$ and the constant $\Lambda$ chosen in
        \eqref{E:SmoothingApproximation} - \eqref{E:IterateIntialValueEstimate}
        below. Here, ${^{(0)}\mathring{\mathbf{V}}}$ denotes the mollified initial data
        as described in Section \ref{SS:AbbreviatedProof}. Furthermore, the set $\bar{\mathcal{O}}_2,$
        the mollified initial data ${^{(0)}\mathring{\mathbf{V}}},$
        and constant $\Lambda$ can be chosen to be independent of all
        initial data varying in a small $H^N$ neighborhood of
        $\ringV.$ Therefore, if we define
        $B_{y}(\ringV) \overset{\mbox{\tiny{def}}}{=} \lbrace \mathring{\widetilde{\mathbf{V}}} \in H_{\Vb}^N  \ | \
        \|\mathring{\widetilde{\mathbf{V}}} - \mathring{\mathbf{V}}\|_{H^N} <
        y\rbrace,$ then there exist $\delta >0$ and $T'>0$ (depending on $\ringV$) such that any initial data
        $\mathring{\widetilde{\mathbf{V}}}$ belonging to $B_{\delta}(\ringV)$
        launch a unique classical solution $\widetilde{\mathbf{V}}$ that exists on the common time interval
        $[0,T']$ and that has the property $\widetilde{\mathbf{V}}([0,T']\times \mathbb{R}^3) \subset \bar{\mathcal{O}}_2.$
    \end{corollary}
    
    \begin{proof}
        The corollary follows from the proof of Theorem
        \ref{T:LocalExistence}. See in particular Remark \ref{R:ExtraDerivativeNeeded} and Remark \ref{R:LambdaRemark} below.
    \end{proof}

    \begin{corollary}                                       \label{C:UniformNormBound}
        The norms $\mid\mid\mid \mathbf{V}
        \mid\mid\mid_{H_{\Vb}^N,T}$ and $\mid\mid\mid \partial_t \mathbf{V}
        \mid\mid\mid_{H^{N-1},T}$ of the solution from Theorem \ref{T:LocalExistence}
        depend only $\bar{\mathcal{O}}_2,$ $\|{^{(0)}\mathring{\mathbf{V}}}\|_{H_{\Vb}^{N+1}},$
        and $\Lambda.$ Furthermore, there exists a $K>0$ such that any initial data
        $\mathring{\widetilde{\mathbf{V}}}$ belonging to the set $B_{\delta}(\ringV)$
        defined in Corollary \ref{C:ExistenceInterval} launch a unique solution
        $\widetilde{\mathbf{V}}$ that satisfies the uniform bound
        \begin{align}                                           \label{E:UniformNormBound}
            \mid\mid\mid \widetilde{\mathbf{V}}
            \mid\mid\mid_{H_{\Vb}^N,T'},\mid\mid\mid \partial_t
            \widetilde{\mathbf{V}} \mid\mid\mid_{H^{N-1},T'} <
            K(N,\bar{\mathcal{O}}_2,\|{^{(0)}\mathring{\mathbf{V}}}\|_{H_{\Vb}^{N+1}}, \Lambda, \delta),
        \end{align}
        where $T$ and $T'$ are as in Corollary \ref{C:ExistenceInterval}.
 		\end{corollary}    
        
        \begin{proof}
        The estimates for $\mid\mid\mid \mathbf{V}
        \mid\mid\mid_{H_{\Vb}^N,T}$ and $\mid\mid\mid \widetilde{\mathbf{V}}
        \mid\mid\mid_{H_{\Vb}^N,T'}$ follow from Corollary \ref{C:ExistenceInterval},
        Proposition \ref{P:UniformTime}, and the fact that the sequence of
        iterates $\lbrace{^{(m)}\mathring{\mathbf{V}}}(t)\rbrace$ constructed below
        converges strongly in $L_{\Vb}^2$ and weakly in $H_{\Vb}^N$ to $\mathbf{V}(t);$
        consult \cite{aM1984} for the missing details. We then use the EN$_{\kappa}$ equations to solve for the
        time derivatives together with Proposition \ref{P:CompositionProductSobolevMoser} and Remark \ref{R:SobolevCalculusRemark}
        to obtain the estimates for \\
        $\mid\mid\mid \partial_t
        \mathbf{V} \mid\mid\mid_{H^{N-1},T}$ and $\mid\mid\mid \partial_t
        \widetilde{\mathbf{V}} \mid\mid\mid_{H^{N-1},T'}.$
        \end{proof}

    \begin{theorem} (\textbf{Continuous Dependence on Initial Data})                        \label{T:ContinuousDependence}
        Let $\mathring{\mathbf{V}}(\snew)$ be initial data for the EN$_{\kappa}$ system
        \eqref{E:ENkappa1} - \eqref{E:ENkappa9}
        that are subject to the conditions described in Section
        \ref{S:IVP}, and let $\mathbf{V}$ be the solution existing on the time interval
        $[0,T]$ furnished by Theorem \ref{T:LocalExistence}.
        Let $B_{\delta}(\ringV)$ be as in Corollary \ref{C:ExistenceInterval}.
        Let $\lbrace \ringVm \rbrace \subset
        B_{\delta}$ be a sequence of initial data with $\lim_{m \to \infty}
        \|\ringVm - \ringV \|_{H^N} =
        0,$ and let $\Vm$ denote the solution to \eqref{E:ENkappa1} - \eqref{E:ENkappa9} launched by
        $\ringVm.$ Then for all large $m,$
        the solutions $\Vm$ exist on $[0,T],$
        and $\lim_{m \to \infty} \mid\mid\mid \Vm - \mathbf{V} \mid\mid\mid_{H^N,T} = 0.$
    \end{theorem}
    
    \begin{proof}
    	Our proof of Theorem \ref{T:ContinuousDependence} is located in Section \ref{SS:ContinuousDependence}.
    \end{proof}

    \begin{remark}                                                                  \label{R:ContinuousDependenceHolder}
        It is unknown to the author whether or not the continuity
        statement from Theorem \ref{T:ContinuousDependence} can be
        strengthened to one of Lipschitz continuity or H\"{o}lder
        continuity. However, using Burger's equation $\partial_t u + u \partial_x u = 0,$
        Kato \cite{tK1975} provides a counterexample in which the map from the initial data $u_0 \in H^a(\mathbb{R})$ to the 
        solution $u(t) \in C([0,T],H^a)$
        is not H\"{o}lder continuous with any positive
        exponent; such a counterexample is explicitly constructed
        for $a \geq 2.$ On the other hand, inequality
        \eqref{E:ContinuousDependenceInequality1} below shows that
        for the EN$_{\kappa}$ system, the map from the initial
        data to the solution is a Lipschitz-continuous map from
        $H_{\Vb}^N$ into $C([0,T],H_{\Vb}^{N-1}).$
    \end{remark}

    \subsection{A discussion of the structure of the proof of the theorems}                  \label{SSS:ProofDescription}
        We prove local existence by following a standard method
        described in detail in Majda's book \cite{aM1984}:
        we construct a sequence of iterates $\lbrace{^{(m)}\mathbf{V}}(t,\snew)\rbrace$
        that converges to the solution $\mathbf{V}_{sol}(t,\snew).$ To construct the iterates, we
        first define a sequence of $C^{\infty}$ initial data
        $\lbrace{^{(m)}\mathring{\mathbf{V}}}\rbrace$ such that ${^{(m)}\mathring{\mathbf{V}}}(\mathbb{R}^3) \Subset \mathcal{O}_2$
        and $\lim_{m \to \infty}\mid\mid{^{(m)}\mathring{\mathbf{V}}}-\mathring{\mathbf{V}}\mid\mid_{H^N}=0.$
        The advantage of smoothing the data is that all of the iterates are $C^{\infty},$ thus allowing us to
        work with classical derivatives
        during the approximation process. Then beginning with
        ${^{(0)}\mathbf{V}}(t,\snew)\overset{\mbox{\tiny{def}}}{=}{^{(0)}\mathring{\mathbf{V}}}(\snew),$ we inductively
        define ${^{(m+1)}\mathbf{V}}(t,\snew)$
        as the unique solution to the linearization of the EN$_{\kappa}$
        system around ${^{(m)}\mathbf{V}}(t,\snew)$ with initial data
        ${^{(m+1)}\mathbf{V}}(0,\snew)={^{(m+1)}\mathring{\mathbf{V}}}(\snew).$
        As a consequence of the theory of linear\footnote{The exposition on linear theory in \cite{rCdH1966} makes use of the
        symmetric
        hyperbolic setup to obtain energy estimates for the linear
        systems. We may obtain similar energy estimates for the
        linearized EN$_{\kappa}$ equations by using energy currents of the form \eqref{E:EnergyCurrent};
        the proof of Proposition \ref{P:UniformTime} below illustrates the relevant techniques.}  PDEs (consult
        \cite{rCdH1966}), each iterate ${^{(m)}\mathbf{V}}$ is known to possess a
        classical solution on a strip $[0,T_m] \times \mathbb{R}^3,$ on which it
        satisfies, for every real number $N',$ ${^{(m)}\mathbf{V}} \in C^{0}([0,T_m],H_{\Vb}^{N'}).$
        Here, $T_m$ is any real number such that ${^{(m-1)}\mathbf{V}}([0,T_m] \times \mathbb{R}^3) \subset
        \bar{\mathcal{O}}_2,$ which ensures that the sequence of proper energy densities is bounded from below by
        a uniform constant and therefore precludes singularities in energy the
        currents we use during the linearization process.

        In order for the limiting function $\mathbf{V}_{sol}$ to be defined on a strip, it is obviously
        necessary that we show that the sequence of time values $\lbrace T_m \rbrace$ can be
        bounded from below by a positive constant $T_*.$
        To this end, we examine the EOV
        satisfied by ${^{(m)}\mathbf{V}} - {^{(0)}\mathring{\mathbf{V}}}$ and its partial derivatives, and we control the
        growth in $T_*$ of $\mid\mid\mid {^{(m)}\mathbf{V}} - {^{(0)}\mathring{\mathbf{V}}}\mid\mid\mid_{H^N,T_*}$ uniformly in $m$ 
        using energy currents. According to the above paragraph and the Sobolev embedding result $H^2(\mathbb{R}^3) \subset 
        C^0_b(\mathbb{R}^3),$
        it follows that if $\mid\mid\mid {^{(m)}\mathbf{V}} - {^{(0)}\mathring{\mathbf{V}}}\mid\mid\mid_{H^N,T_{*}}$ is small
        enough, uniformly in $m,$ then $T_{*}$ may be selected as a uniform lower
        bound on the $T_m.$ Our detailed proof of the control of the terms
        $\mid\mid\mid {^{(m)}\mathbf{V}} -
        {^{(0)}\mathring{\mathbf{V}}}\mid\mid\mid_{H^N,T_*}$
        is given in Proposition \ref{P:UniformTime} below and
        uses the Sobolev-Moser calculus inequalities, which are refined versions of the fact that
        for $N' > \frac{3}{2},$ $H^{N'}(\mathbb{R}^3)$ is a Banach
        algebra. Their purpose is to control the $L^2$ norms of terms
        of a product form, based on known Sobolev regularity of each factor in the
        product. We state the relevant Sobolev-Moser estimates in the Appendix and give references
        for readers interested in the proofs.

        Our proof of Proposition \ref{P:UniformTime} illustrates
        the relevant techniques for obtaining Sobolev estimates
        from the method of energy currents. Instead of completing the existence proof, which requires
        arguments similar to the ones used in proving this proposition, we refer the reader
        to Majda's local existence proof for symmetric hyperbolic systems \cite{aM1984};
        the only necessary modification to Majda's proof
        is to use the method of energy currents in place of the energy
        principle for symmetric hyperbolic systems.

        In Section \ref{SS:UniqunessandHNMinusOneDependence} we show uniqueness and $H^{N-1}-$Lipschitz-continuous dependence on the 
        initial data. The methods used in this argument are similar to the 
        methods used to prove Proposition \ref{P:UniformTime}, so we provide fewer details.
        We consider the EOV satisfied by the difference of two solutions $\mathbf{V}$ and
        $\widetilde{\mathbf{V}}$ to the EN$_{\kappa}$ system, and then use an appropriately defined energy current to bound the 
        growth of $\mid\mid\mid \mathbf{V} - \widetilde{\mathbf{V}} \mid\mid\mid_{H^{N-1},T}$ by
        a constant times exponential growth in $T.$ We show that the constant
        depends on the initial data and is bounded from above by another constant times $\| \mathbf{V}(0) - 
        \widetilde{\mathbf{V}}(0) \|_{H^{N-1}},$ thus implying uniqueness and $H^{N-1}-$Lipschitz-continuous dependence 
        on the initial data. Our abbreviated proof of Theorem \ref{T:LocalExistence} is complete at the end of this section.

        Our proof of Theorem \ref{T:ContinuousDependence} requires some machinery from the theory of evolution
        equations in a Banach space. The basic method is due to Kato \cite{tK1975}, and
        most of the technical results we use in this section are merely quoted
        from his papers. We find it worthwhile to prove Theorem
        \ref{T:ContinuousDependence} because aside from Kato's work, we have had
        difficulty locating this result in the literature.

    \subsection{An abbreviated proof of Theorem \ref{T:LocalExistence}} \label{SS:AbbreviatedProof}
        As described in Section \ref{SSS:ProofDescription}, we
        produce a sequence of iterates $\lbrace{^{(m)}\mathbf{V}}(t,\snew)\rbrace$ that converges to the
        solution $\mathbf{V}_{sol}(t,\snew).$

        \subsubsection{Smoothing the initial data}
        We begin by smoothing the initial data $\mathring{\mathbf{V}},$ which we assume are of the
        form described in Section \ref{S:IVP}, so that we can work with classical derivatives. Let $\Psi(\snew)$ be a Friedrich's mollifier; i.e. $\Psi \in
        C_c^\infty(\mathbb{R}^3), \  \mbox{supp}(\Psi) \subset \{\snew| \
        |\snew| \leq 1 \}, \Psi \geq 0,$ and $\int \Psi \ d^3\snew = 1.$  For $\epsilon > 0,$
        we set $\Psi_\epsilon(\snew) \overset{\mbox{\tiny{def}}}{=}
        {\epsilon}^{-3} \Psi(\frac{\snew}{\epsilon})$ and
        define $\Psi_\epsilon \mathring{\mathbf{V}}
        \in C^\infty(\mathbb{R}^3)$ by
        \begin{equation}
            \Psi_\epsilon \mathring{\mathbf{V}}(\snew) \overset{\mbox{\tiny{def}}}{=} \int_{\mathbb{R}^3}
            \Psi_\epsilon(\snew- \snew')\mathring{\mathbf{V}}(\snew') \,
            d^3\mathbf{s'}.
        \end{equation}

    The following properties of such a mollification are well-known:

        \begin{align}
            \underset{\epsilon \rightarrow 0^+}{\lim}
                \|\Psi_\epsilon \mathring{\mathbf{V}} - \mathring{\mathbf{V}} \|_{H^N} &=0        \label{E:Mollification1}\\
            \exists \lbrace\epsilon_0 > 0  \land C(\mathring{\mathbf{V}})>0\rbrace
                \owns  0 < \epsilon <
                \epsilon_0 & \Rightarrow \|\Psi_\epsilon \mathring{\mathbf{V}} - \mathring{\mathbf{V}} \|_{L^2}
                \leq \epsilon C(\mathring{\mathbf{V}}) \| \mathring{\mathbf{V}} \|_{H^1}.         \label{E:Mollification2}
        \end{align}
        We will choose below an $\epsilon_0$ that is at least as small as the one
        in \eqref{E:Mollification2}. Once chosen, for a given $m \in \mathbb{N},$ we define

        \begin{align}
            \epsilon_m              &\overset{\mbox{\tiny{def}}}{=} 2^{-m} \epsilon_0                                \label{E:EpsilonDef}\\
            {^{(m)}\mathring{\mathbf{V}}}    &\overset{\mbox{\tiny{def}}}{=} \Psi_{\epsilon_m} \mathring{\mathbf{V}} \label{E:IterateInitialValue1}\\
            {^{(m)}\mathring{\mathbf{W}}}    &\overset{\mbox{\tiny{def}}}{=} \Psi_{\epsilon_m} \mathring{\mathbf{W}}, \label{E:IterateInitialValue2}
        \end{align}
        where $\mathring{\mathbf{W}}$ denotes the first 5 components of $\ringV.$

            By Sobolev embedding, by the assumptions on the initial data
            $\mathring{\mathbf{V}}$, and by the mollification properties above,
            $\exists \lbrace \Lambda>0 \land \epsilon_0 > 0 \rbrace$ (at least as
            small as the $\epsilon_0$ in \eqref{E:Mollification2}) such that

        \begin{align}
            \| \mathring{\mathbf{V}} - {^{(0)} \mathring{\mathbf{V}}} \|_{H^N} &\leq C_{\bar{\mathcal{O}}_2} \frac{\Lambda}{4}          			\label{E:SmoothingApproximation}\\
            \| \mathbf{V} - {^{(0)} \mathring{\mathbf{V}}} \|_{H^N} &\leq \Lambda \Rightarrow \mathbf{V}(\mathbb{R}^3)
                \subset \bar{\mathcal{O}}_2                                                 \label{E:WellDefined}\\
            \|{^{(m)} \mathring{\mathbf{V}}} - {^{(0)} \mathring{\mathbf{V}}}\|_{H^N} & \leq
                C_{\bar{\mathcal{O}}_2} \frac{\Lambda}{2}
                \ \mbox{holds for} \ m \geq 0,                                \label{E:IterateIntialValueEstimate}
        \end{align}
        where $C_{\bar{\mathcal{O}}_2}$ is defined in \eqref{E:UniformPositivity}.

        \begin{remark}                                                                          \label{R:ExtraDerivativeNeeded}
            It is a standard result that if $\epsilon > 0$ and
            $N'$ is any real number, then $\Psi_\epsilon
            \mathring{\mathbf{V}} \in H_{\Vb}^{N'}(\mathbb{R}^3).$ We will make use of this remark below, for
            in the local existence proof, we will need to
            differentiate the equations \eqref{E:IterateInhomogeneous1} - \eqref{E:IterateInhomogeneousl4}
            $N$ times and utilize Sobolev estimates; since several terms from these
            undifferentiated equations already contain one derivative of
            the smoothed initial data, our estimates will involve
            $\|{^{(0)} \mathring{\mathbf{V}}}\|_{H_{\Vb}^{N+1}}.$ See e.g. 
            \eqref{E:bHNEstimate} and \eqref{E:SobolevClaim3}.
        \end{remark}

        \begin{remark}                                                              \label{R:LambdaRemark}
            If we are considering initial data $\mathring{\widetilde{\mathbf{V}}}$ in a small enough $H^N$
            neighborhood $\mathcal{N}$ of the initial data $\ringV,$ we can
            use a \emph{fixed} smoothed function ${^{(0)}
            \mathring{\mathbf{V}}}$ in place of each ${^{(0)}
            \mathring{\widetilde{\mathbf{V}}}}$ in Proposition \ref{P:UniformTime} below, and choose $\Lambda$ to be
            uniform over the neighborhood. For what then enters into the proof of local existence
            for the initial data $\mathring{\widetilde{\mathbf{V}}}$ are the quantities $\|{^{(0)}
            \mathring{\mathbf{V}}}\|_{H_{\Vb}^{N+1}}$ and
            $\|{^{(m)} \mathring{\widetilde{\mathbf{V}}}} - {^{(0)} \mathring{\mathbf{V}}}
            \|_{H^N},$ and the latter quantity is easily
            controlled by the inequality
            \begin{align}                                       \label{E:SmoothedInitialDataInequality}
                \|{^{(m)} \mathring{\widetilde{\mathbf{V}}}} - {^{(0)}
                \mathring{\mathbf{V}}}  \|_{H^N} \leq \|{^{(m)} \mathring{\widetilde{\mathbf{V}}}} -
                \mathring{\widetilde{\mathbf{V}}}\|_{H^N}
                + \|\mathring{\widetilde{\mathbf{V}}} - \ringV\|_{H^N} + \| \ringV - {^{(0)}
                \mathring{\mathbf{V}}}\|_{H^N};
            \end{align}
            once we fix an appropriately chosen smoothed function ${^{(0)}
            \mathring{\mathbf{V}}}$ and a corresponding $\Lambda$ satisfying \eqref{E:SmoothingApproximation}
            and \eqref{E:WellDefined}, we may independently adjust
            the mollification of each $\mathring{\widetilde{\mathbf{V}}}$ belonging to $\mathcal{N}$ so that the right-hand side of
            \eqref{E:SmoothedInitialDataInequality} is $\leq
            C_{\bar{\mathcal{O}}_2} \Lambda/2$ for $m \geq 0.$ This estimate would then enter into
            our proof in inequality \eqref{E:IterateInequality}. We also note that this remark is relevant for Corollary
            \ref{C:ExistenceInterval} above.
        \end{remark}

        \subsubsection{Defining the iterates}       \label{SS:IterateDefinition}
        Consider the iteration scheme described in Section \ref{SSS:ProofDescription}. The components of the iterates are denoted by 
        ${^{(m)}\mathbf{V}}=\left({^{(m)}\Ent},{^{(m)}P},\cdots,{^{(m)}\psi_3} 
        \right),$ and we use the notation ${^{(m)}\mathbf{W}}$ to denote the first five components of ${^{(m)}\mathbf{V}}.$ Linear 
        existence theory implies
        that each iterate ${^{(m+1)}\mathbf{V}}$ is a well-defined, smooth function with
    		$\|{^{(m+1)}\mathbf{V}}(t)- {^{(0)}\mathring{\mathbf{V}}} \|_{H^N} < \infty$ for $0 \leq t \leq T_m.$
        Here, by \eqref{E:WellDefined}, $[0,T_m]$ is any time interval
        for which $\mid\mid\mid {^{(m)}\mathbf{V}}- {^{(0)}\mathring{\mathbf{V}}}
        \mid\mid\mid_{H^N,T_m} \leq \Lambda$ holds.

        \subsubsection{The uniform time estimate} \label{SSS:UniformTime}
        As discussed in Section \ref{SSS:ProofDescription}, we show the
        existence of a fixed $T_*>0$ such that $\mid\mid\mid {^{(m)}\mathbf{V}} - {^{(0)}\mathring{\mathbf{V}}}
        \mid\mid\mid_{H^N,T_*} \leq \Lambda$ for $m \in \mathbb{N},$ thus ensuring that each iterate
        is defined for a uniform amount of time and remains inside
        of $\bar{\mathcal{O}}_2.$ We state a slightly stronger
        version of this result as a proposition:
        \begin{proposition}                                                          \label{P:UniformTime}
            Let $\Lambda$ denote the constant defined in \eqref{E:SmoothingApproximation} - \eqref{E:IterateIntialValueEstimate}. 
            Then there exist $T_* >0$ and $L>0$ such that each of the iterates ${^{(m)}\mathbf{V}}(t,\snew)$ satisfies
                \begin{subequations}
                    \begin{align}
                        \mid\mid\mid {^{(m)}\mathbf{V}} - {^{(0)}\mathring{\mathbf{V}}}\mid\mid\mid_{H^N,T_*} \leq \Lambda        \label{E:UniformTime1a}\\
                        \mid\mid\mid \partial_t\big({^{(m)}\mathbf{V}} \big) \mid\mid\mid_{H^{N-1},T_*} \leq L.       \label{E:UniformTime1b}
                    \end{align}
                \end{subequations}
					\end{proposition}
             
            \noindent {\bf Proof.} We proceed in our proof of Proposition \ref{P:UniformTime} by induction on $m,$ noting that 
            ${^{(0)}\mathbf{V}(t,\snew)} 
            \overset{\mbox{\tiny{def}}}{=} {^{(0)}\mathring{\mathbf{V}}}(\snew)$
            satisfies \eqref{E:UniformTime1a} and \eqref{E:UniformTime1b} with
            any $T_{*} > 0$ and any positive number $L$. We thus assume that
            $^{(m)}\mathbf{V}$ satisfies \eqref{E:UniformTime1a} and \eqref{E:UniformTime1b}
            without first specifying the values of $T_*$ or $L.$ At the end of the
            proof, we will show that we can choose such a $T_*$ and an
            $L,$ both independent of $m,$ such that energy estimates imply the inductive
            step. To obtain the estimates stated in the proposition, it is
            convenient
            to work not with the iterates themselves, but with the difference
            between the iterate and the smoothed initial value. Thus,
            referring to the notation defined in \eqref{E:IterateInitialValue1} and \eqref{E:IterateInitialValue2} ,
            for each $m \in \mathbb{N}$ we define
            \begin{align}                                                       \label{E:DotVmUniformTime}
                \dot{\mathbf{V}}(t,\snew) &\overset{\mbox{\tiny{def}}}{=} {^{(m+1)}\mathbf{V}}(t,\snew) -
                    {^{(0)}\mathring{\mathbf{V}}}(\snew)                                  \\
                \dot{\mathbf{W}}(t,\snew) &\overset{\mbox{\tiny{def}}}{=} {^{(m+1)}\mathbf{W}}(t,\snew)
                    - {^{(0)}\mathring{\mathbf{W}}}(\snew) \\
                \widetilde{\mathbf{V}} &\overset{\mbox{\tiny{def}}}{=}{^{(m)}\mathbf{V}}       \label{E:WidetildeVUniformTime}.
            \end{align}

            We have used the notation $\dot{\mathbf{V}}$ and $\widetilde{\mathbf{V}}$ suggestively: it follows from
            the the definition of the iterates, definition \eqref{E:DotVmUniformTime}, and
            definition \eqref{E:WidetildeVUniformTime} that
            $\dot{\mathbf{V}}$ is a solution to the EOV
            \eqref{E:EOV1} - \eqref{E:WidetildeR} defined by the bgs $\widetilde{\mathbf{V}}$ with initial data
            $\dot{\mathbf{V}}(0,\snew) = {^{(m+1)}\mathring{\mathbf{V}}}(\snew) -
            {^{(0)}\mathring{\mathbf{V}}}(\snew).$ Our notation \eqref{E:DotVmUniformTime} - \eqref{E:WidetildeVUniformTime}
            is therefore consistent with our notation for the EOV introduced in Section \ref{SS:EOV}.
            Recalling also the notation \eqref{E:bdef} and \eqref{E:ldef}
            introduced in Section \ref{SS:EOV}, the inhomogeneous terms in the EOV satisfied by $\dot{\mathbf{V}}$
            are given by $(\mathbf{b},\mathbf{l})=(f,g,\cdots,l^{(4)}),$ where for $j=1,2,3,$

            \begin{align}
                f &= - \widetilde{U}^k \partial_k  [{^{(0)}\mathring{{\Ent}}}]                               \label{E:IterateInhomogeneous1}\\
                g &= -\widetilde{U}^k \partial_k  [{^{(0)}\mathring{P}}] \ - \
                    \widetilde{Q} \partial_k  [{^{(0)}\mathring{U}^k}] \ + \
                    (4\widetilde{P}-3\widetilde{Q})\widetilde{U}^\mu \widetilde{\psi}_{\mu}             \label{E:IterateInhomogeneous2}\\
                h^{(j)} &= -(\widetilde{R} + \widetilde{P}) \widetilde{U}^k \partial_k [{^{(0)}
                    \mathring{U}^j}] \ - \widetilde{\Pi}^{kj} \partial_k [{^{(0)}\mathring{P}}]
                    + (3\widetilde{P} - \widetilde{R})\widetilde{\Pi}^{\mu j}
                    \widetilde{\psi}_{\mu}   	                          \label{E:IterateInhomogeneous3}\\
                l^{(0)} &= \kappa^2 \widetilde{\phi} + \widetilde{R} - 3\widetilde{P} -
                    \partial^k[{^{(0)}\mathring{\psi}_k}]                                               \label{E:IterateInhomogeneousl0}\\
                l^{(j)} &= \partial^j [{^{(0)}\mathring{\psi}_0}]            \label{E:IterateInhomogeneouslj}\\
                l^{(4)} & = \widetilde{\psi}_0.                                                         \label{E:IterateInhomogeneousl4}
            \end{align}

        As explained in Section \ref{SS:EOV}, for each \emph{spatial} derivative multi-index $\vec{\alpha}$ with
        \\ $0 \leq |\vec{\alpha}| \leq N,$ we may differentiate the EOV
        with inhomogeneous terms $(\mathbf{b},\mathbf{l})$
        to which $\dot{\mathbf{V}}$ is a solution, obtaining that $\partial_{\vec{\alpha}}
        \dot{\mathbf{V}}$ is also a solution to the EOV defined
        by the \emph{same} bgs $\widetilde{\mathbf{V}}$ with inhomogeneous terms
        $(\mathbf{b}_{\vec{\alpha}}, \partial_{\vec{\alpha}}\mathbf{l}).$
        The inhomogeneous terms $\mathbf{b}_{\vec{\alpha}}$ are given by
        \begin{align}
            \mathbf{b}_{\vec{\alpha}} \overset{\mbox{\tiny{def}}}{=}
                    A^0 \partial_{\vec{\alpha}} \left((A^0)^{-1}\mathbf{b} \right) +
                    \mathbf{k}_{\vec{\alpha}},                                                                \label{E:balphaDef}
        \end{align}
        where
        \begin{align}
            \mathbf{k}_{\vec{\alpha}} \overset{\mbox{\tiny{def}}}{=} A^0 \left[(A^0)^{-1}A^k \partial_k (\partial_{\vec{\alpha}}\dot{\mathbf{W}})
            -\partial_{\vec{\alpha}}
            \left((A^0)^{-1} A^k \partial_k \dot{\mathbf{W}} \right)\right]                              \label{E:kalphaDef}
        \end{align}
        for $0 \leq |\vec{\alpha}| \leq N.$ Note that we have suppressed the dependence
        of the $A^{\nu}(\cdot)$ on $\widetilde{\mathbf{V}}.$

        As discussed in Section \ref{SS:JdotUniformPositiveDefinite}, we
        will use energy currents to control $\mid\mid\mid\dot{\mathbf{V}}\mid\mid\mid_{H^N,T}.$  We state here as a lemma an important
        differential inequality that allows us to proceed with our desired
        Sobolev estimates. Its proof is based on the key properties 
        of energy currents described in Section \ref{SS:EnergyCurrents} and the divergence theorem.

        \begin{lemma}                                                                   \label{L:DifferentialInequality}
        (See Fig. \ref{Fi:DODCone})
        Suppose $r \geq T>0.$ For $0 \leq t \leq T,$ let \\ $\Sigma_{t,r-t} \overset{\mbox{\tiny{def}}}{=}
        \{x \in \mathcal{M} |x^0=t, x^k x_k \leq r - t \} $ denote the
        Euclidean sphere of radius $r-t$ centered at $(t,0,0,0)$ in the flat
        hypersurface $\{x^0=t\},$ and let \\ $M_{t,r} \overset{\mbox{\tiny{def}}}{=} \{x \in \mathcal{M}| 0
        \leq x^0 \leq t, x^k x_k = r - x^0 \}$ denote the mantle of the
        past directed, truncated light cone with lower base $\Sigma_{0,r}$
        and upper base $\Sigma_{t,r-t}.$ Let $\dot{\mathbf{V}}$ be a solution to the
        EOV \eqref{E:EOV1} - \eqref{E:WidetildeR} defined by the bgs $\widetilde{\mathbf{V}},$
        and assume that \\ $\widetilde{\mathbf{V}}([0,T] \times \mathbb{R}^3) \subset
            \bar{\mathcal{O}}_2.$ Let $\dot{J}$ be the energy current \eqref{E:EnergyCurrent} for the variation
            $\dot{\mathbf{V}}$ defined by the bgs $\widetilde{\mathbf{V}},$ and define $\mathscr{E}(t;r) \overset{\mbox{\tiny{def}}}{=}
        \Big(\underset{\Sigma_{t,r-t}}{\int} \dot{J}^0(t,\snew) \,
        d^3\snew \Big)^{1/2}.$  Then

    \begin{equation}
            2\mathscr{E}(t;r)\frac{d}{dt}\mathscr{E}(t;r) \leq
                \underset{\Sigma_{t,r-t}}{\int} \partial_\mu
                \dot{J}^\mu(t,\snew) \, d^3\snew.
    \end{equation}
    \end{lemma}

    \begin{remark}
        We note that our use of $t$ in the statement of Lemma \ref{L:DifferentialInequality} as a constant value taken on by the 
        generic spacetime coordinate $x^0$ is inconsistent with our usual notational convention for spacetime
        points defined in Section \ref{SS:SpacetimeConventions}, in which $t$ and $x_0$ are both used
        in the same manner as generic coordinate variables.
    \end{remark}

        \begin{proof}
            By the divergence theorem, we have that

            \begin{align}                                                       \notag
                 &\mathscr{E}^2(t;r) = \underset{\Sigma_{t,r-t}}{\int}
                    \dot{J}^0(t,\snew) \, d^3\snew = \underset{\Sigma_{0,r}}{\int} \dot{J}^0(0,\snew)
                 \,  d^3\snew\ \\
                 & - \ \underset{M_{t,r}}{\int} \langle \hat{n}(x),\dot{J}(x)\rangle_E \, d\mathcal{H}(x)
                 \ + \ \int_{t'=0}^{t'=t} \Big(\underset{\Sigma_{t',r-t'}}{\int}
                    \partial_\mu \dot{J}^\mu(t',\snew) \, d^3\snew \Big) \,
                    d{t'}. \label{E:DifferentialInequality}
            \end{align}

            \begin{figure}[htp]
            		\centering
                \includegraphics{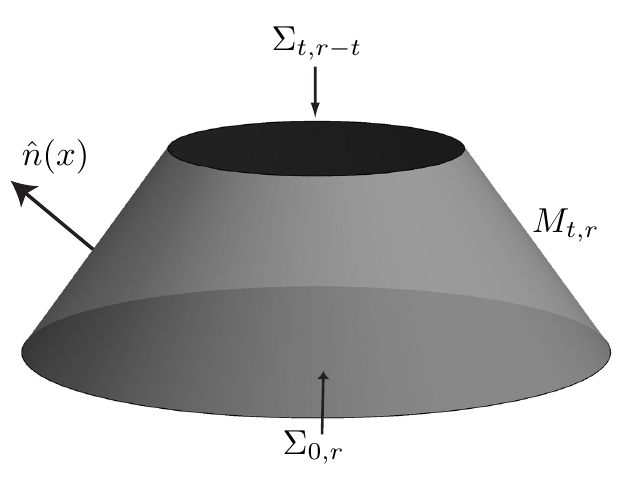}
                \caption{The Surfaces of Integration in Lemma \ref{L:DifferentialInequality}} \label{Fi:DODCone}
             \end{figure}

            Here, $\hat{n}(x)$ is the Euclidean outer normal at $x \in M_{t,r}$ to the
            mantle of truncated cone, $\langle\hat{n}(x),\dot{J}(x)\rangle_E$ denotes
            the Euclidean inner product of $\hat{n}(x)$ and $\dot{J}(x)$ as
            vectors in $\mathbb{R}^4,$ and $\mathcal{H}$ is the Hausdorff measure
            on the mantle of the cone. For each normal vector
            $\hat{n}(x),$ let $_{\hat{n}(x)}\xi$ denote the co-vector
            belonging to $T^*_x \mathcal{M}$ such that $_{\hat{n}(x)}\xi(X) =
            {_{\hat{n}(x)}\xi_{\mu}} X^{\mu} = \langle\hat{n}(x),X\rangle_E$ holds for every
            $X \in T_x \mathcal{M}.$ By the positivity condition
            \eqref{E:PositiveDefinite}, co-vectors $\xi$ belonging to
            $\mathcal{I}^{*+}_x$ satisfy $\xi_{\mu} \dot{J}^{\mu} (\dot{\mathbf{V}},\dot{\mathbf{V}}) > 0$ for all non-zero 
            variations
            $\dot{\mathbf{V}}.$ Since for each $x \in M_{t,r},$ the co-vector $_{\hat{n}(x)}\xi$ belongs to the
            boundary of $\mathcal{I}^{*+}_x$, which is the positive component of the cone $C^*_{x,l},$ continuity
            in the variable $\xi$ implies that $\langle\hat{n}(x),\dot{J}(x)\rangle_E \ =
            {_{\hat{n}(x)}\xi_{\mu}} \dot{J}^{\mu}(\dot{\mathbf{V}},\dot{\mathbf{V}})
            \geq 0$ holds for $x \in M_{t,r}.$ Furthermore, if $t_1 < t_2$, then $M_{t_1,r} \subset M_{t_2,r}.$ From
            these facts it follows that $-\underset{M_{t,r}}{\int}\langle\hat{n}(x),\dot{J}(x)\rangle_E
            \,d\mathcal{H}(x)$ is a decreasing function of $t$ on $[0,T].$  Lemma \ref{L:DifferentialInequality} now follows
            from differentiating each side of \eqref{E:DifferentialInequality} with respect to
            $t$ and accounting for this decreasing term. Fig. \ref{Fi:DODCone} illustrates the
            setup in $\mathbb{R}^{1+2},$ where the vertical direction represents positive values of $t.$
        \end{proof}

    Returning to the proof of the proposition and recalling that
    we are using definitions \eqref{E:DotVmUniformTime} and \eqref{E:WidetildeVUniformTime} to
    define $\dot{\mathbf{V}}$ and $\widetilde{\mathbf{V}},$ we let $\dot{J}_{\vec{\alpha}}$ denote the energy current for the variation
    $\partial_{\vec{\alpha}} \dot{\mathbf{V}}$ defined by the bgs $\widetilde{\mathbf{V}}.$
    For notational convenience, we allow $\vec{\alpha}$ to take on the value $\vec{\mathbf{0}},$
    in which case $\dot{J}_{\vec{\mathbf{0}}}$ is defined to be the energy current in the variation
    $ \dot{\mathbf{V}}$ defined by the bgs $\widetilde{\mathbf{V}}.$

    As in Lemma \ref{L:DifferentialInequality}, we define for any $T_*>0$ and $r>T_*$ the
    following functions of $t$ on $[0,T_*]:$
    \begin{align}
        \mathscr{E}_{\vec{\alpha}}(t;r)&\overset{\mbox{\tiny{def}}}{=}\Big(\underset{\Sigma_{t,r-t}}{\int}
            {\dot{J}_{\vec{\alpha}}}^0(t,\snew) \, d^3\snew\Big)^{1/2} \notag \\
        E(t;r;N) &\overset{\mbox{\tiny{def}}}{=} \Big(\sum_{0 \leq |\vec{\alpha}| \leq N}
            \mathscr{E}_{\vec{\alpha}}^2(t;r)\Big)^{\frac{1}{2}} 
            = \Big( \sum_{0 \leq |\vec{\alpha}| \leq N} \underset{\Sigma_{t,r-t}}{\int}
            {\dot{J}_{\vec{\alpha}}}^0(t,\snew) \, d^3 \snew \Big)^{\frac{1}{2}} .
    \end{align}
    Then with $C_{\bar{\mathcal{O}}_2}$ defined in \eqref{E:UniformPositivity}, we have that
    \begin{align}
        C_{\bar{\mathcal{O}}_2}E^2(t;r;N) &\leq \|\dot{\mathbf{V}}(t)\|_{H^N(\Sigma_{t,r-t})}^2
         \leq C_{\bar{\mathcal{O}}_2}^{-1} E^2(t;r;N) \label{E:NormRelationship}.
    \end{align}
    Additionally, by Lemma \ref{L:DifferentialInequality}, we have the following inequality
    for $0 \leq t \leq T_*:$

    \begin{align}
        2E(t;r;N) \frac{d}{dt}E(t;r;N)                                  \label{E:DifferentialInequality2}
            \leq \sum_{0 \leq |\vec{\alpha}| \leq N}{} \underset{\Sigma_{t,r-t}}{\int} \partial_\mu
            \left({\dot{J}_{\vec{\alpha}}}^\mu(t,\snew)\right) \, d^3\snew.
    \end{align}

    The technically cumbersome aspect of the proof of Proposition \ref{P:UniformTime} is bounding the
    right-hand side of \eqref{E:DifferentialInequality2} by a
    constant times $E(t;r;N) + E^2(t;r;N),$
    which then allows us to use Gronwall's inequality to exponentially bound from above the
    growth of $E(t;r;N)$ in $t.$ We prove some of the technical points in
    lemmas \ref{L:EnergyCurrentDivergenceL1SobolevEstimate} and \ref{L:L} below, so as to not disrupt the
    main argument. The keys to proofs of lemmas \ref{L:EnergyCurrentDivergenceL1SobolevEstimate} and \ref{L:L}
    are Sobolev-Moser calculus inequalities, special versions of which are stated in the Appendix.
    In the following argument, $C=C(N,\bar{\mathcal{O}}_2,\|{^{(0)}\mathring{\mathbf{V}}}\|_{H_{\Vb}^{N+1}},\Lambda,L),$
    even though we have not yet chosen $L.$ By Lemma \ref{L:EnergyCurrentDivergenceL1SobolevEstimate}, we have that
    \begin{align}                                                                                                   \label{E:SobolevMoserBound}
        \sum_{0 \leq |\vec{\alpha}| \leq N}{} \underset{\Sigma_{t,r-t}}{\int}
            \partial_\mu \left({\dot{J}_{\vec{\alpha}}}^\mu(t,\snew)\right)
            \, d^3\snew & \leq C \cdot \big[\|\dot{\mathbf{V}}(t) \|_{H^N(\Sigma_{t,r-t})} + \|\dot{\mathbf{V}}(t) \|_{H^N(\Sigma_{t,r-t})}^2\big] \notag \\
        & \leq C \cdot \big[C^{-1/2}_{\bar{\mathcal{O}}_2}E(t;r;N) + C^{-1}_{\bar{\mathcal{O}}_2}
            E^2(t;r;N)\big],
    \end{align}
    where in the second inequality we have used
    \eqref{E:NormRelationship}. Combining \eqref{E:DifferentialInequality2} with
    \eqref{E:SobolevMoserBound}, and applying Gronwall's inequality, we have for $0
    \leq t \leq T_*$ that
    \begin{align}                                                                           \label{E:EquivalentNormBound}
        E(t;r;N)  \leq  \big[E(0;r;N)
            + C\cdot(2C_{\bar{\mathcal{O}}_2})^{-1/2}t
            \big] \cdot \big[\mbox{exp}\big(C\cdot(2C_{\bar{\mathcal{O}}_2})^{-1}t\big)\big],
    \end{align}
    and consequently by \eqref{E:NormRelationship}, that
    \begin{align}                                                                                                    \label{E:ConicalDotVEstimate}
        \|\dot{\mathbf{V}}(t)\|_{H^N(\Sigma_{t,r-t})} &\leq {C_{\bar{\mathcal{O}}_2}^{-1}}
            \big[\|\dot{\mathbf{V}}(0)\|_{H^N(\Sigma_{0,r})} + Ct\big] \cdot
            \big[\mbox{exp}\big(Ct\big)\big].
    \end{align}
    Letting $r \to \infty,$ taking the
    $\sup$ over $t \in [0,T_*],$ and using \eqref{E:IterateIntialValueEstimate}, we
    have that

    \begin{align}                                                                                               \label{E:IterateInequality}
        &\mid\mid\mid\dot{\mathbf{V}}\mid\mid\mid_{H^N,T_*} \leq
            C^{-1}_{\bar{\mathcal{O}}_2} \big[\|\dot{\mathbf{V}}(0)\|_{H^N} + CT_*\big] \cdot \mbox{exp}\big(CT_*\big) \notag \\
        &\leq \big[\Lambda/2+ C(N,\bar{\mathcal{O}}_2,\|{^{(0)}\mathring{\mathbf{V}}}\|_{H_{\Vb}^{N+1}},\Lambda,L)\cdot T_* \big]
            \cdot \mbox{exp}\big(C(N,\bar{\mathcal{O}}_2,\|{^{(0)}\mathring{\mathbf{V}}}\|_{H_{\Vb}^{N+1}},\Lambda,L) \cdot T_*\big).
    \end{align}

    To make a viable choice of $L,$
    \begin{align}
        \mbox{we first \emph{assume} that right-hand side of \eqref{E:IterateInequality} is} \ \leq\Lambda,    \label{E:Assumption}
    \end{align}
    which implies the inductive step \eqref{E:UniformTime1a} for ${^{(m+1)}\mathbf{V}}$.
    Using assumption \eqref{E:Assumption} as a hypothesis, Lemma \ref{L:L} implies that there exists
    an $L(N,\bar{\mathcal{O}}_2,\|{^{(0)}\mathring{\mathbf{V}}}\|_{H_{\Vb}^{N}},\Lambda) >0$ such that
    \begin{align}                                                                               \label{E:AprioriPartialtHNMinusOneBound}
            \mid\mid\mid \partial_t\big({^{(m+1)}\mathbf{V}} \big) \mid\mid\mid_{H^{N-1},T_*} \leq L.
        \end{align}
    For this fixed choice of $L,$ we can implicitly solve for a $T_* >
    0$ such that the right-hand side of
    inequality \eqref{E:IterateInequality} is in fact $\leq \Lambda,$ thus justifying the assumption \eqref{E:Assumption}
    and the conclusion \eqref{E:AprioriPartialtHNMinusOneBound},
    thereby closing the induction argument. This completes the proof of Proposition \ref{P:UniformTime}.
    \begin{flushright}
    $\Box$
    \end{flushright}

		\subsection{Proofs of Lemma \ref{L:EnergyCurrentDivergenceL1SobolevEstimate} and Lemma \ref{L:L}}
    We now state and prove the two technical lemmas quoted in the
    proof of Proposition \ref{P:UniformTime}.

    \begin{lemma}                                                               \label{L:EnergyCurrentDivergenceL1SobolevEstimate}
        Assume the hypotheses and notation of Proposition
        \ref{P:UniformTime}. In addition, assume that \\ $\|\partial_t \widetilde{\mathbf{V}}(t)\|_{H^{N-1}} \leq L.$ Then
        \begin{align} \label{E:EnergyCurrentDivergenceL1Estimate}
            \|\partial_\mu
            {\dot{J}_{\vec{\alpha}}}^\mu(t)\|_{L^1(\Sigma_{t,r-t})} \leq
            C(N,\bar{\mathcal{O}}_2,\|{^{(0)}\mathring{\mathbf{V}}}\|_{H_{\Vb}^{N+1}},\Lambda,L)
            \cdot \left(\|\dot{\mathbf{V}}(t)\|_{H^N} + \|\dot{\mathbf{V}}(t)\|_{H^N}^2\right).                     
        \end{align}
 		\end{lemma}
    \begin{proof}
        We use here the definitions \eqref{E:DotVmUniformTime} and \eqref{E:WidetildeVUniformTime}
        from Proposition \ref{P:UniformTime}. Recall that $\partial_{\vec{\alpha}}
        \dot{\mathbf{V}}$ is a solution to the EOV defined by the
        bgs $\widetilde{\mathbf{V}}$ with inhomogeneous terms $(\mathbf{b}_{\vec{\alpha}}, \partial_{\vec{\alpha}}\mathbf{l}),$
        and that $\dot{J}_{\vec{\alpha}}$ is the energy current
        for $\partial_{\vec{\alpha}} \dot{\mathbf{V}}$ defined by the
        bgs $\widetilde{\mathbf{V}}.$ Furthermore,
        \begin{align}
            \|\widetilde{\mathbf{V}} - {^{(0)}\mathring{\mathbf{V}}}\|_{H^N} \leq \Lambda                     \label{E:InductionAssumption}
        \end{align}
        holds by the induction assumption from the proposition.

        By \eqref{E:EnergyCurrentDivergence} and Remark \ref{R:HigherOrderEnergyCurrents},
        the expression for $\partial_\mu {\dot{J}_{\vec{\alpha}}}^\mu$ consists of terms that
        are either precisely linear or precisely quadratic in
        the components of the variation $\partial_{\vec{\alpha}}
        \dot{\mathbf{V}}.$ The coefficients of the quadratic variation
        terms are smooth functions with arguments $\widetilde{\mathbf{V}}$ and
        $D\widetilde{\mathbf{V}}.$ Examining the particular form of
        these coefficients and using the fact that $\widetilde{\mathbf{V}}([0,T]\times\mathbb{R}^3) \subset \bar{\mathcal{O}}_2,$
        we see that their $L^{\infty}$ norm is bounded by $C(\bar{\mathcal{O}}_2)
        \|D\widetilde{\mathbf{V}}\|_{L^{\infty}}.$ By
        assumption, $\|\widetilde{\mathbf{V}} - {^{(0)}\mathring{\mathbf{V}}}\|_{H^N} \leq
        \Lambda$ and $\|\partial_t \widetilde{\mathbf{V}}(t)\|_{H^{N-1}}
        \leq L.$ Therefore, by Sobolev embedding, $\|D\widetilde{\mathbf{V}}\|_{L^{\infty}} \leq
        C(N,\|{^{(0)}\mathring{\mathbf{V}}}\|_{H_{\Vb}^{N}},\Lambda,L).$
        These facts imply that the $L^1(\Sigma_{t,r-t})$ norm of the terms involving the quadratic
        variations is bounded from above by $C(N,\bar{\mathcal{O}}_2,\|{^{(0)}\mathring{\mathbf{V}}}\|_{H_{\Vb}^{N}},
        \Lambda,L)\|\partial_{\vec{\alpha}} \dot{\mathbf{V}}\|_{L^2(\Sigma_{t,r-t})}^2.$

        The coefficients of the linear variation terms are
        linear combinations of products of the components of
        $(\mathbf{b}_{\vec{\alpha}}, \partial_{\vec{\alpha}}\mathbf{l}),$
        where $\mathbf{b}_{\vec{\alpha}}$ is defined in \eqref{E:balphaDef},
        with smooth functions, the arguments of which are the
        components of $\widetilde{\mathbf{V}}.$ Since \\ $\widetilde{\mathbf{V}}([0,T]\times\mathbb{R}^3) \subset \bar{\mathcal{O}}_2,$
        the smooth functions of $\widetilde{\mathbf{V}}$ are bounded in $L^{\infty}$ by
        $C(\bar{\mathcal{O}}_2).$ Therefore, by the
        Cauchy-Schwarz integral inequality for $L^2,$ the
        $L^1(\Sigma_{t,r-t})$ norm of the terms depending
        linearly on the variations is bounded from above by
        $C(\bar{\mathcal{O}}_2)\|(\mathbf{b}_{\vec{\alpha}}, \partial_{\vec{\alpha}}\mathbf{l})\|_{L^2}
        \|\partial_{\vec{\alpha}}
        \dot{\mathbf{V}}\|_{L^2(\Sigma_{t,r-t})}.$ To complete the proof of
        \eqref{E:EnergyCurrentDivergenceL1Estimate}, it remains to show
        that for $0 \leq |{\vec{\alpha}}| \leq N,$ we have that
        \begin{align}       \label{E:ComponentsinL2}
            \|(\mathbf{b}_{\vec{\alpha}}, \partial_{\vec{\alpha}}\mathbf{l})\|_{L^2} \leq
            C(N,\bar{\mathcal{O}}_2,\|{^{(0)}\mathring{\mathbf{V}}}\|_{H_{\Vb}^{N+1}},\Lambda) \cdot \big(1 + \|\dot{\mathbf{W}}\|_{H^N}\big).
        \end{align}
        The proof of \eqref{E:ComponentsinL2} will follow easily
        from the propositions given in the Appendix.

        Concerning ourselves with the $\|\mathbf{b}_{\vec{\alpha}}\|_{L^2}$ estimate first, we claim that the term $A^0 \partial_{\vec{\alpha}}
        \left((A^0)^{-1}\mathbf{b} \right)$ from
        \eqref{E:balphaDef} satisfies
        \begin{align}                                                                           \label{E:SobolevClaim}
            \| A^0 \partial_{\vec{\alpha}} \left((A^0)^{-1}\mathbf{b} \right) \|_{L^2} \leq
            C(N,\bar{\mathcal{O}}_2,\|{^{(0)}\mathring{\mathbf{V}}}\|_{H_{\Vb}^{N+1}},\Lambda).
        \end{align}
        We repeat for clarity that $\mathbf{b}=(f,g,h^1,h^2,h^3),$
        where the scalar-valued quantities $f,g,h^1,h^2,h^3$ are
        defined in \eqref{E:IterateInhomogeneous1} - \eqref{E:IterateInhomogeneous3}.
        Since $\|A^0(\widetilde{\mathbf{V}})\|_{L^{\infty}} \leq
        C(\bar{\mathcal{O}}_2),$ to prove \eqref{E:SobolevClaim}, it suffices to control the
        $L^2$ norm of $\partial_{\vec{\alpha}} \big((A^0)^{-1}\mathbf{b} \big).$
        Using Proposition \ref{P:CompositionProductSobolevMoser} and Remark \ref{R:SobolevCalculusRemark},
        with $(A^0)^{-1}$ playing the role of $F$ in the proposition and $\mathbf{b}$
        playing the role of $G,$ we have that
        \begin{align} \label{E:A0bHNEstimate}
            \|\partial_{\vec{\alpha}} \big((A^0)^{-1}\mathbf{b} \big)\|_{L^2} \leq \|(A^0)^{-1}\mathbf{b}\|_{H^N} \leq
            C(N,\bar{\mathcal{O}}_2,\|{^{(0)}\mathring{\mathbf{V}}}\|_{H_{\Vb}^{N}},\Lambda)\|\mathbf{b}\|_{H^N}.
        \end{align}
        Furthermore, Proposition \ref{P:CompositionProductSobolevMoser} and Remark \ref{R:SobolevCalculusRemark} imply that
        \begin{align}                                           \label{E:bHNEstimate}
            \|\mathbf{b}\|_{H^N} \leq
            C(N,\bar{\mathcal{O}}_2,\|{^{(0)}\mathring{\mathbf{V}}}\|_{H_{\Vb}^{N+1}},\Lambda).
        \end{align}
        Combining \eqref{E:A0bHNEstimate} with \eqref{E:bHNEstimate} proves
        \eqref{E:SobolevClaim}.

        We next claim that the $\mathbf{k}_{\vec{\alpha}}$ from
        \eqref{E:kalphaDef} satisfy
        \begin{align}                                                               \label{E:SobolevClaim2}
            \|\mathbf{k}_{\vec{\alpha}}\|_{L^2} \leq
            C(N,\bar{\mathcal{O}}_2,\|{^{(0)}\mathring{\mathbf{V}}}\|_{H_{\Vb}^{N}},\Lambda)\|\dot{\mathbf{W}}\|_{H^N}.
        \end{align}
        Again, since $\|A^0(\widetilde{\mathbf{V}})\|_{L^{\infty}} \leq
        C(\bar{\mathcal{O}}_2),$ to prove \eqref{E:SobolevClaim2}, it suffices to control the $L^2$ norm of
        $(A^0)^{-1}A^k \partial_k (\partial_{\vec{\alpha}}\dot{\mathbf{W}}) -\partial_{\vec{\alpha}}
        \left((A^0)^{-1} A^k \partial_k \dot{\mathbf{W}} \right).$ By
        Proposition \ref{P:SobolevMissingDerivativeProposition} and Remark \ref{R:SobolevMissingDerivativeRemark},
        with $(A^0)^{-1}A^k$ playing the role of $F$
        in the proposition, and $\partial_k \dot{\mathbf{W}}$ playing the
        role of $G,$ we have that
        \begin{align}
            &\|(A^0)^{-1}A^k \partial_k (\partial_{\vec{\alpha}}\dot{\mathbf{W}}) -\partial_{\vec{\alpha}}
                \big((A^0)^{-1} A^k \partial_k \dot{\mathbf{W}} \big)\|_{L^2} \\
            &\leq C(N,\bar{\mathcal{O}}_2,\|{^{(0)}\mathring{\mathbf{V}}}\|_{H_{\Vb}^{N}},\Lambda) \|\nabla^{(1)} \dot{\mathbf{W}}\|_{H^{N-1}}, \notag
        \end{align}
        from which \eqref{E:SobolevClaim2} immediately follows.

        To finish the proof of \eqref{E:ComponentsinL2}, we will show that
        \begin{align}         \label{E:SobolevClaim3}
            \|l^{(z)}\|_{H^N} \leq C(N,\bar{\mathcal{O}}_2,\|{^{(0)}\mathring{\mathbf{V}}}\|_{H_{\Vb}^{N+1}},\Lambda)
            \qquad (z=0,1,2,3,4).
        \end{align}

        For $l^{(1)},l^{(2)},l^{(3)},l^{(4)}$ defined in
        \eqref{E:IterateInhomogeneouslj} and \eqref{E:IterateInhomogeneousl4}, the claim is trivial.
        To estimate the component $l^{(0)},$ defined in
        \eqref{E:IterateInhomogeneousl0},
        we first rewrite
        \begin{align} \label{E:Rewritel0}
            l^{(0)} = \kappa^2 (\widetilde{\phi} - \bar{\phi}) + (\widetilde{R} -
            \bar{R}) - 3(\widetilde{P} - \bar{P}) - \partial^k[{^{(0)}\mathring{\psi}_k}],
        \end{align}
        where $\bar{P} \overset{\mbox{\tiny{def}}}{=} e^{4\bar{\phi}}
        \bar{p}$ and $\bar{R} \overset{\mbox{\tiny{def}}}{=} e^{4\bar{\phi}}
        \mathscr{R}(e^{-4\bar{\phi}}\bar{P},\bar{{\Ent}}),$
        the function $\mathscr{R}$ is defined in \eqref{E:R}, and $\bar{p}$
        and $\bar{{\Ent}}$ are constants defined in Section \ref{S:IVP}. In equation \eqref{E:Rewritel0}, we
        have made use of \eqref{E:phiBar}, which is the assumption that $\kappa^2 \bar{\phi} + \bar{R} - 3\bar{P}=0.$ Since
        \begin{align}
            &\|\kappa^2 (\widetilde{\phi} - \bar{\phi})\|_{H^N} + 3\|(\widetilde{P} - \bar{P})\|_{H^N}
            + \|\partial^k[{^{(0)}\mathring{\psi}_k}]\|_{H^N} \leq
            C(\|{^{(0)}\mathring{\mathbf{V}}}\|_{H_{\Vb}^{N+1}},\Lambda),
        \end{align}
        we only need to show that
        \begin{align}
            \|\widetilde{R} - \bar{R}\|_{H^N} \leq
            C(N,\bar{\mathcal{O}}_2,\|{^{(0)}\mathring{\mathbf{V}}}\|_{H_{\Vb}^{N}},\Lambda).                                                \label{E:SobolevforR}
        \end{align}
        This follows immediately from definition \eqref{E:WidetildeR}, Proposition \ref{P:SobolevTaylor}, and Remark 
        \ref{R:SobolevTaylorCalculusRemark}.
        
   			Inequality \eqref{E:ComponentsinL2} now follows from combining \eqref{E:balphaDef}, \eqref{E:SobolevClaim}, 
   			\eqref{E:SobolevClaim2}, and \eqref{E:SobolevClaim3}; this completes the proof of
        \eqref{E:EnergyCurrentDivergenceL1Estimate}.
    \end{proof}

    \begin{lemma}                                                                                           \label{L:L}
         Assume the hypotheses and notation of Proposition
        \ref{P:UniformTime}. Also assume the induction hypothesis $\mid\mid\mid {^{(m)}\mathbf{V}} -
        {^{(0)}\mathring{\mathbf{V}}}\mid\mid\mid_{H^N,T_*} \leq \Lambda$ from Proposition
        \ref{P:UniformTime}. Assume further that $\mid\mid\mid{^{(m+1)}\mathbf{V}} -
        {^{(0)}\mathring{\mathbf{V}}}\mid\mid\mid_{H^N,T_*} \leq
        \Lambda.$ Then
        \begin{align}
            \mid\mid\mid \partial_t\big({^{(m+1)}\mathbf{V}} \big) \mid\mid\mid_{H^{N-1},T_*} \leq
            L(N,\bar{\mathcal{O}}_2,\|{^{(0)}\mathring{\mathbf{V}}}\|_{H_{\Vb}^{N}},\Lambda).                              \label{E:LCombinedEsitmate}
        \end{align}
        \end{lemma}
        
        \begin{proof}
            By Remark \ref{R:Invertible}, we may solve for
            $\partial_t({^{(m+1)}\mathbf{W}}):$
            \begin{align}
                \partial_t\big({^{(m+1)}\mathbf{W}}\big)=
                \big(A^0({^{(m)} \mathbf{V}})\big)^{-1} \big(\mathbf{b}
                -A^k({^{(m)} \mathbf{V}})\partial_k({^{(m+1)}\mathbf{W}})\big),                              \label{E:IterateTimeDerivative}
            \end{align}
            where the function $\mathbf{b}$ denotes the
            inhomogeneous terms from the linearized EN$_{\kappa}$ equations
            satisfied by ${^{(m+1)}\mathbf{W}};$ i.e.,
            $\mathbf{b}{=}\boldsymbol{\mathfrak{B}}({^{(m)}\mathbf{V}}),$ where
            \begin{align}                                                                                                                   \label{E:BFunctiondef}
                \boldsymbol{\mathfrak{B}}(\cdot) \overset{\mbox{\tiny{def}}}{=}
                \left(\mathfrak{F}(\cdot),\mathfrak{G}(\cdot),
                \mathfrak{H}^{(1)}(\cdot),\mathfrak{H}^{(2)}(\cdot),\mathfrak{H}^{(3)}(\cdot)\right)
            \end{align}
            is an array-valued function, the scalar-valued functions $\mathfrak{F},\mathfrak{G},\cdots,\mathfrak{H}^{(3)}$ are defined
            in \eqref{E:Inhomogeneousf} - \eqref{E:Inhomogeneoushj}, and the $A^{\mu}(\cdot)$ are defined
            in \eqref{E:MatrixDef}.

            Using the hypotheses of the lemma, we apply Proposition
            \ref{P:CompositionProductSobolevMoser} and Remark \ref{R:SobolevCalculusRemark}
            to the right-hand side of \eqref{E:IterateTimeDerivative}, concluding that
            \begin{align}                                                           \label{E:LEstimatePart1}
                \mid\mid\mid \partial_t\big({^{(m+1)}\mathbf{W}} \big) \mid\mid\mid_{H^{N-1},T_*}
                \leq L(N,\bar{\mathcal{O}}_2,\|{^{(0)}\mathring{\mathbf{V}}}\|_{H_{\Vb}^{N}},\Lambda).
            \end{align}
            Likewise, an argument similar to the one used to prove \eqref{E:SobolevClaim3} gives that
            \begin{align}   \label{E:LEstimatePart2}
                \mid\mid\mid \partial_t\big({^{(m+1)}\phi},{^{(m+1)}\psi_0},\cdots,{^{(m+1)}\psi_3} \big) \mid\mid\mid_{H^{N-1},T_*}
                \leq L(N,\bar{\mathcal{O}}_2,\|{^{(0)}\mathring{\mathbf{V}}}\|_{H_{\Vb}^{N}},\Lambda).
            \end{align}
            Combining \eqref{E:LEstimatePart1} and \eqref{E:LEstimatePart2} proves \eqref{E:LCombinedEsitmate}.
        \end{proof}

        \subsubsection{Uniqueness and $H^{N-1}-$Lipschitz-continuous dependence on initial data.} \label{SS:UniqunessandHNMinusOneDependence}
        We now prove the prove the uniqueness of $H_{\bar{\mathbf{V}}}^{N}$ solutions to
        the EN$_{\kappa}$ system and show that the solution is an $H^{N-1}-$Lipschitz-continuous function 
        of the initial data. Let $\mathring{\mathbf{V}}$ denote initial data that launch a solution $\mathbf{V}$ of the EN$_{\kappa}$
        system as furnished by the existence aspect of Theorem \ref{T:LocalExistence}.
        Let $\delta,B_{\delta}(\ringV),T',$ and $K(N,\bar{\mathcal{O}}_2,\|{^{(0)}\mathring{\mathbf{V}}}\|_{H_{\Vb}^{N+1}}, \Lambda, 	
        \delta)$ be as in corollaries
        \ref{C:ExistenceInterval} and \ref{C:UniformNormBound}. Assume that the initial data $\mathring{\widetilde{\mathbf{V}}}$ belong to $B_{\delta},$
        and let $\widetilde{\mathbf{V}}$ be a solution of the
        EN$_{\kappa}$ system with initial data $\mathring{\widetilde{\mathbf{V}}}$
        existing on the interval $[0,T']$ as furnished by Corollary \ref{C:ExistenceInterval}. We now define
        \begin{align}                                                               \label{E:dotVUnique}
            \dot{\mathbf{V}}\overset{\mbox{\tiny{def}}}{=}\widetilde{\mathbf{V}} - \mathbf{V}.
        \end{align}
        It follows from definition \eqref{E:dotVUnique} that $\dot{\mathbf{V}}$ is a solution to the EOV
        \eqref{E:EOV1} - \eqref{E:WidetildeR} defined by the bgs $\widetilde{\mathbf{V}}$
        with inhomogeneous terms given by (for $j=1,2,3$)

        \begin{align}
            f &= (U^{\mu} - \widetilde{U}^{\mu}) \partial_{\mu} {\Ent}          \label{E:UniquenessInhomogeneousf}\\
            g &= (U^{\mu} - \widetilde{U}^{\mu}) \partial_{\mu} P
                + \big[Q U_k/U^0 - \widetilde{Q}
                \widetilde{U}_k/\widetilde{U}^0\big] \partial_0 U^k                  \label{E:UniquenessInhomogeneousg}\\
             & \ \ + (Q - \widetilde{Q}) \partial_k U^k +
             ( 4\widetilde{P} - 3\widetilde{Q} ) \widetilde{U}^{\mu} \widetilde{\psi}_{\mu}
             - (4P - 3Q) U^{\mu} \psi_{\mu}                   \notag \\
            h^{(j)} &= \big[(R+P)U^{\mu} - (\widetilde{R} + \widetilde{P})\widetilde{U}^{\mu}\big]\partial_{\mu}U^j
                + (\Pi^{\mu j} - \widetilde{\Pi}^{\mu j}) \partial_{\mu} P \\
                & \ \ + (3\widetilde{P}-\widetilde{R})
                \widetilde{\Pi}^{\mu j} \widetilde{\psi}_{\mu}
                - (3P - R) \Pi^{\mu j} \psi_{\mu}                                     \label{E:UniquenessInhomogeneoush}\\
            l^{(0)} & = {\kappa}^2(\widetilde{\phi} - \phi) + (\widetilde{R} - 3 \widetilde{P}) - (R - 3P) \label{E:UniquenessInhomogeneousl0}\\
            l^{(j)} & = 0                                      \label{E:UniquenessInhomogeneouslj}\\
            l^{(4)} & = \widetilde{\psi}_0 - \psi_0,                                           \label{E:UniquenessInhomogeneousl4}
    \end{align}
    and we denote them using the abbreviated notation $\mathbf{b}$ and $\mathbf{l}$ defined in \eqref{E:bdef} and \eqref{E:ldef}.

            By combining Proposition \ref{P:CompositionProductSobolevMoser}, Remark \ref{R:SobolevCalculusRemark},
            Proposition \ref{P:SobolevTaylor}, and Remark \ref{R:SobolevTaylorCalculusRemark}
            (noting the particular manner in which the inhomogeneous terms depend
            on the difference of functions of $\mathbf{V}$ and $\widetilde{\mathbf{V}}$), we have that
            \begin{align} \label{E:InhomogeneousContinuousDependenceL2Estimate}
                \mid\mid\mid(\mathbf{b},\mathbf{l})\mid\mid\mid_{H^{N-1},T'} \leq C(N,\bar{\mathcal{O}}_2,K)
                \mid\mid\mid\dot{\mathbf{V}}\mid\mid\mid_{H_{\Vb}^{N-1},T'}.
            \end{align}
            Without providing details, we reason as in our proof of Proposition \ref{P:UniformTime}, using
            \eqref{E:InhomogeneousContinuousDependenceL2Estimate} in place of
            \eqref{E:bHNEstimate} and \eqref{E:SobolevClaim3} to arrive at the following bound:

            \begin{align} \label{E:ContinuousDependenceInequality1}
                \mid\mid\mid\dot{\mathbf{V}}\mid\mid\mid_{H_{\Vb}^{N-1},T'} &\leq
        C^{-1}_{\bar{\mathcal{O}}_2}\|\dot{\mathbf{V}}(0)\|_{H^{N-1}} \cdot
        \mbox{exp}\big(C(N,\bar{\mathcal{O}}_2,K) \cdot T'\big),
       \end{align}
        where $\dot{\mathbf{V}}(0) \overset{\mbox{\tiny{def}}}{=} \mathring{\widetilde{\mathbf{V}}} - \ringV.$

      We now observe that \eqref{E:ContinuousDependenceInequality1}
      implies both the uniqueness statement in Theorem \ref{T:LocalExistence} and the $H^{N-1}-$Lipschitz-continuous dependence on 
      the initial data mentioned in Remark \ref{R:ContinuousDependenceHolder}.

      \begin{remark}
        We cannot obtain an estimate analogous to \eqref{E:ContinuousDependenceInequality1}
        by using the $H_{\Vb}^N$ norm in place of the $H_{\Vb}^{N-1}$ norm; the
        inhomogeneous terms \eqref{E:UniquenessInhomogeneousf} -
        \eqref{E:UniquenessInhomogeneousl4} already contain one
        derivative of $\mathbf{V},$ and therefore cannot be bounded in the $H_{\Vb}^N$
        norm. However, for $N'<N,$ we can obtain an estimate for the $H_{\Vb}^{N'}$ norm
        by combining Proposition \ref{P:SobIterpolation}, \eqref{E:ContinuousDependenceInequality1} and the uniform bound
        provide by the constant $K.$
        The inequality we obtain is
        \begin{align}                   \label{E:ContinuousDependenceInequality4}
        \mid\mid\mid\dot{\mathbf{V}}\mid\mid\mid_{H_{\Vb}^{N'},T'} &\leq
        C \|\dot{\mathbf{V}}(0)\|^{1 - N'/N}_{H^{N'}} \cdot
        \mbox{exp}\left(CT'\right),
        \end{align}
        where in \eqref{E:ContinuousDependenceInequality4}, $C=C(N',N,\bar{\mathcal{O}}_2,K).$
        \end{remark}

      \begin{remark}                                                                          \label{R:DOI}
        The estimate \eqref{E:ContinuousDependenceInequality1} is a limiting version of the ``conical" estimate
        \begin{align} \label{E:ContinuousDependenceInequality3}
            \|\dot{\mathbf{V}}(t)\|_{H^{N-1}(\Sigma_{t,r-t})} &\leq
            C^{-1}_{\bar{\mathcal{O}}_2}\|\dot{\mathbf{V}}(0)\|_{H^{N-1}(\Sigma_{0,r})} \cdot \mbox{exp}\left(Ct\right),
        \end{align}
        where we are using notation defined in Lemma \ref{L:DifferentialInequality}.
        A proof of \eqref{E:ContinuousDependenceInequality3} can be
        constructed using arguments similar to the ones used in our proof of \eqref{E:ConicalDotVEstimate}.
        Inequality \eqref{E:ContinuousDependenceInequality3} shows that two solutions that agree on $\Sigma_{0,r}$ also
        agree on $\Sigma_{t,r-t}.$ By translating the cone from Lemma
        \ref{L:DifferentialInequality} so that its
        lower base is centered at the spacetime point $x,$ we may produce a translated version of the inequality.
        Thus, we observe that a \emph{domain of dependence} for $x \in \mathcal{M}$ is
        given by the solid backward light cone in $\mathcal{M}$ with vertex at $x;$ i.e., the past
        (relative to $x$) behavior of a solution to the EOV outside
        of this cone does not influence behavior of the solution at $x.$
        Similarly, a \emph{domain of influence} of $x$ is the solid
        forward light cone with vertex at $x;$ the behavior of a
        solution at $x$ does not influence the future (relative to $x$)
        behavior of the solution outside of this cone, a fact which justifies our claim made in Section \ref{SS:Speeds of Propagation}
        that the fastest speed of propagation in the EN$_{\kappa}$ system is the speed of light. In \cite{dC2000}, Christodoulou gives an 
        advanced discussion of these and related topics for hyperbolic PDEs derivable from a Lagrangian.
    \end{remark}

    This completes our abbreviated proof of Theorem $\ref{T:LocalExistence}.$ 
    \begin{flushright}
    $\Box$
    \end{flushright}

    \subsection{Proof of Theorem \ref{T:ContinuousDependence}.} \label{SS:ContinuousDependence}
    We now provide a detailed proof of Theorem \ref{T:ContinuousDependence}.
    
    \subsubsection{The setup}
    Let $\lbrace\ringVm\rbrace$ be the sequence of initial data from the hypotheses of Theorem \ref{T:ContinuousDependence} converging
    in $H_{\Vb}^N$ to $\ringV$ . By corollaries \ref{C:ExistenceInterval} and \ref{C:UniformNormBound}, for all large $m,$
    the initial data $\ringVm$ and $\mathring{\mathbf{V}}$
    launch unique solutions $\Vm$ and $\mathbf{V}$ respectively to \eqref{E:ENkappa1} -
    \eqref{E:ENkappa9} that exist on a common interval $[0,T']$ and that have the property \\ $\mathbf{V}([0,T']\times\mathbb{R}^3)\subset
    \bar{\mathcal{O}}_2, \Vm([0,T']\times\mathbb{R}^3) \subset
    \bar{\mathcal{O}}_2.$  Furthermore, for all large $m,$ with
    \\ $K=K(N,\bar{\mathcal{O}}_2,\|{^{(0)}\mathring{\mathbf{V}}}\|_{H_{\Vb}^{N+1}},\Lambda,\delta),$
    we have the uniform (in $m$) bounds
    \begin{align}                                                           \label{E:UniformSobolev}
            \mid\mid\mid \mathbf{V} \mid\mid\mid_{H_{\Vb}^N,T},\mid\mid\mid \partial_t
            \mathbf{V} \mid\mid\mid_{H^{N-1},T},\mid\mid\mid \Vm \mid\mid\mid_{H_{\Vb}^N,T'},\mid\mid\mid \partial_t
            \Vm \mid\mid\mid_{H^{N-1},T'} \ < \ K,
    \end{align}
    where $[0,T]$ is the interval of existence for $\mathbf{V}$ furnished by Theorem \ref{T:LocalExistence}. In this section, we will show 
    that for all large $m, \Vm$ exists on $[0,T]$ and that
    \begin{align} \label{E:HNContinuousDependence}
        \lim_{m \to \infty} \mid\mid\mid \Vm - \mathbf{V}
        \mid\mid\mid_{H^N,T} = 0.
    \end{align}

    The proof we give here is inspired by a similar proof given by Kato in \cite{tK1975}.
    We use results and terminology from the theory of abstract evolution equations in Banach
    spaces, an approach that streamlines the argument. We also freely use results from the
    theory of integration in Banach spaces; a detailed discussion of this theory may be found in
    \cite{kY1980}. We begin by rewriting the linearization of the EN$_{\kappa}$
    system around $\mathbf{V}$ and $\Vm$ as abstract evolution equations in the affine Banach space
    $H_{\bar{\mathbf{V}}}^N(\mathbb{R}^3).$ In this form, the linearized systems
    are respectively written as
    \begin{align}
        \partial_t \mathbf{Z} + \mathcal{A}(\mathbf{V}) \mathbf{Z} =
        \mathbf{f}(\mathbf{V}) \label{E:AbstractEvolutionV} \\
        \partial_t \mathbf{Z} + \mathcal{A}(\Vm) \mathbf{Z} = \mathbf{f}(\Vm),
        \label{E:AbstractEvolutionVm}
    \end{align}
    where $\mathbf{f}(\cdot)$ is a smooth function on $\mathcal{O},$
    and $\mathbf{f}(\bar{\mathbf{V}})=\mathbf{0}.$
    Here, the symbol $\mathbf{Z}$ stands for all 10 components of
    a solution to a linearized system, and the operator
    $\mathcal{A}(\cdot)$ is a first order spatial differential operator with
    coefficients that depend smoothly on its arguments. We state
    for clarity that the first 5 components of the inhomogeneous terms $\mathbf{f}(\Vm)$ are
    given by $\left(A^0(\Vm)\right)^{-1}\cdot(\mathfrak{F}(\Vm),\mathfrak{G}(\Vm),
    \cdots,\mathfrak{H}^{(3)}(\Vm))^{Transpose},$ where the
    matrix-valued function $A^0(\cdot)$ is defined in
    \eqref{E:A0Def} and the scalar-valued functions $\mathfrak{F},\mathfrak{G},\cdots \mathfrak{H}^3$ are defined
    in \eqref{E:Inhomogeneousf} - \eqref{E:Inhomogeneoushj}.

    We will make use of the pseudodifferential operator
    \begin{align}
        \mbox{\textnormal{S}}\overset{\mbox{\tiny{def}}}{=}(1-\Delta)^{N/2},
    \end{align}
    which is an isomorphism between $H^N$ and $L^2;$ i.e., $\mbox{\textnormal{S}} \in \mathcal{L}(H^N,L^2)$
    and \\ $\mbox{\textnormal{S}}^{-1} \in \mathcal{L}(L^2, H^N).$

    \subsubsection{Technical estimates}
    In this section, we provide some technical lemmas that will be needed 
    in our proof of Theorem \ref{T:ContinuousDependence}. For certain function spaces $X,$ there exist evolution operators
    \begin{align}
        \mathcal{U}(t,t'), \ \mathcal{U}^m(t,t') : X \rightarrow X
    \end{align}
    defined on $\triangle_{T'} \overset{\mbox{\tiny{def}}}{=} \lbrace 0 \leq t' \leq t \leq T' \rbrace$
    that map solutions (belonging to the space $X$) of the
    corresponding \emph{homogeneous} version of the linearized systems \eqref{E:AbstractEvolutionV} and \eqref{E:AbstractEvolutionVm}
    at time $t'$ to solutions at
    time $t.$ The relevant spaces in our discussion are $X=L^2$ and $X=H^N.$
    In the following three lemmas, we describe the properties of the operators $\mathcal{U}(t,t')$
    and $\mathcal{U}^m(t,t').$ Complete proofs are given in \cite{tK1970},
    \cite{tK1973}, and \cite{tK1975}; rather than repeating them,
    we instead attempt to provide some insight into how the proofs relate
    to the methods described in this paper.

    \begin{lemma} \label{L:BoundedLinearOperators}
        $\mathcal{U}(\cdot,\cdot)$ and $\mathcal{U}^m(\cdot,\cdot)$ (for
        $m \geq 0$) are strongly-continuous maps from $\triangle_{T'}$ into
        $\mathcal{L}(L^2)\cap\mathcal{L}(H^N).$ Furthermore, there exists a $C(K) > 0$ such that
        $\mid\mid\mid \mathcal{U} \mid\mid\mid_{L^2,\triangle_{T'}}, \mid\mid\mid
        \mathcal{U}^m \mid\mid\mid_{L^2,\triangle_{T'}}, \mid\mid\mid \mathcal{U}
        \mid\mid\mid_{H^N,\triangle_{T'}}, \mid\mid\mid
        \mathcal{U}^m \mid\mid\mid_{H^N,\triangle_{T'}} < C(K).$
    \end{lemma}

    \begin{remark}
        Lemma \ref{L:BoundedLinearOperators} is essentially a
        consequence of the fact that the uniform bound \eqref{E:UniformSobolev} for
        $\mathbf{V}$ and $\Vm$ allows for uniform Sobolev
        estimates to be made on the $L^2$ (or $H^N$) norm of $L^2$ (or $H^N$) solutions to the
        linearized equations.
    \end{remark}

    \begin{remark}                                          \label{R:HNInhomogeneousTerms}
        By the regularity result of $\mathbf{V}$ furnished by Theorem \ref{T:LocalExistence}, Corollary
        \ref{C:SobolevCorollary}, Remark \ref{R:SobolevCalculusRemark}, and \eqref{E:phiBar}, the right-hand side
        \eqref{E:AbstractEvolutionV} is an element of
        $C^0([0,T],H^N).$ Given a function $\mathring{\mathbf{Z}} \in H_{\Vb}^N,$
        it follows from Lemma \ref{L:BoundedLinearOperators} and
        standard linear theory (via Duhamel's principle) that there exists a unique solution
        $\mathbf{Z} \in C^0([0,T],H_{\Vb}^N)$ to \eqref{E:AbstractEvolutionV} with initial data
        equal to $\mathring{\mathbf{Z}}.$ An 
        analogous result holds for solutions to \eqref{E:AbstractEvolutionVm}.
    \end{remark}

    \begin{lemma} \label{L:StrongConvergenceL2HN}
        $\mathcal{U}^m(t,t')$ converges to $\mathcal{U}(t,t')$
        strongly in $\mathcal{L}(L^2)$ as $m \to \infty.$ Furthermore, the strong
        convergence is uniform on $\triangle_{T'}.$
    \end{lemma}
    \begin{remark}
        By smoothing the initial data, a solution $\mathbf{Z} \in C^0([0,T],L^2)$ to either $\partial_t \mathbf{Z} + 
        \mathcal{A}(\mathbf{V}) \mathbf{Z} = 0$ or $\partial_t \mathbf{Z} + \mathcal{A}(\Vm) \mathbf{Z} = \mathbf{f}(\Vm)$ can be realized 
        as the limit (in the norm $\mid\mid\mid \cdot \mid\mid\mid_{L^2,T}$) of a sequence
        $\lbrace \mathbf{Z}^k \rbrace \subset C^0([0,T],H^N).$ Therefore, to prove Lemma \ref{L:StrongConvergenceL2HN},
        one only needs to check that given initial data $\mathring{\mathbf{Z}} \in H^N,$ we have that
        \begin{align}                                                                                                                           \label{E:StrongConvergence}
            \lim_{m \to \infty} \mid\mid\mid \big(\mathcal{U}^m(\cdot,0) - \mathcal{U}(\cdot,0)\big)
            \mathring{\mathbf{Z}}\mid\mid\mid_{L^2,T}=0.
        \end{align}
         Based on Lemma \ref{L:BoundedLinearOperators} and \eqref{E:ContinuousDependenceInequality4},
         which can be used to show that for $N' < N$ we have $\Vm \rightarrow \mathbf{V}$ in 
         $C^0([0,T'],H_{\Vb}^{N'})\cap C_b^1([0,T'] \times \mathbb{R}^3),$
         \eqref{E:StrongConvergence} follows from the method of energy currents.
    \end{remark}

    \begin{lemma} \label{L:OperatorBL2Norm}
        There exist operator-valued functions $\mathcal{B}:[0,T] \rightarrow \mathcal{L}(L^2), \\ \Bm :[0,T'] \rightarrow 
        \mathcal{L}(L^2)$
        such that
        \begin{align}
            \mbox{\textnormal{S}}\mathcal{A}(\mathbf{V}(t)) \mbox{\textnormal{S}}^{-1} = \mathcal{A}(\mathbf{V}(t)) + \mathcal{B}(t) 							\\
            \mbox{\textnormal{S}} \mathcal{A}(\Vm(t)) \mbox{\textnormal{S}}^{-1} = \mathcal{A}(\Vm(t)) + \Bm(t).
        \end{align}

        Furthermore, for all $t'$ with $0 \leq t' \leq T',$ $\mathcal{B}$ and $\Bm$ satisfy the estimates
        \begin{align}
            \mid\mid\mid \Bm - \mathcal{B} \mid\mid\mid_{L^2,t'} & \leq C(K) 
            \mid\mid\mid \Vm - \mathbf{V} \mid\mid\mid_{H^N,t'} \\
           \mbox{and} \ \mid\mid\mid \mathcal{B} \mid\mid\mid_{L^2,T}, \ \mid\mid\mid \Bm \mid\mid\mid_{L^2,T'} &\leq C(K).
        \end{align}
    \end{lemma}
    
    \begin{remark}
    	A modern treatment of commutator estimates that can be used to prove Lemma \ref{L:OperatorBL2Norm} is located
    	in \cite{MT1991}.
    \end{remark}

		Our proof of Theorem \ref{T:ContinuousDependence} also requires the following lemma, whose
		simple proof is based on Duhamel's principle:

    \begin{lemma}               \label{L:Duhamel}
        Let $\mathring{\mathbf{Z}} \in H_{\Vb}^N,$ 
        and let $\mathbf{Z} \in C^0([0,T],H_{\Vb}^N)$ denote the
        unique solution to \eqref{E:AbstractEvolutionV} with initial data
        equal to $\mathring{\mathbf{Z}}$ as furnished by Remark \ref{R:HNInhomogeneousTerms}. Then
        for \\ $t \in [0,T],$ $\mbox{\textnormal{S}}(\mathbf{Z} - \Vb)$
        satisfies the Duhamel formula
        \begin{align} \label{E:Duhamel}
            \mbox{\textnormal{S}}(\mathbf{Z}(t)-\Vb) =
            \mathcal{U}(t,0)\mbox{\textnormal{S}}(\mathring{\mathbf{Z}} - \Vb) & - \int_0^t
            \mathcal{U}(t,t') \mathcal{B}(t')
            \mbox{\textnormal{S}}(\mathbf{Z}(t') - \Vb) \, dt' \notag \\
            &+ \int_0^t \mathcal{U}(t,t') \mbox{\textnormal{S}}\mathbf{f}(\mathbf{V}(t')) \, dt'. 
        \end{align}
        An analogous result holds for the linearization of the EN$_{\kappa}$ system around $\Vm.$
    \end{lemma}

    \begin{proof} We apply $\mbox{\textnormal{S}}$ to each side of the equation
    satisfied by $\mathbf{Z} - \Vb$ and use Lemma \ref{L:OperatorBL2Norm} to arrive
    at the equation
        \begin{align}                                                       \label{E:SZEquation}
            \partial_t \big[\mbox{\textnormal{S}} (\mathbf{Z} - \Vb) \big] + \mathcal{A}(\mathbf{V})
            \mbox{\textnormal{S}} (\mathbf{Z} - \Vb) =
            - \mathcal{B}(\mathbf{V})\mbox{\textnormal{S}}(\mathbf{Z} - \Vb)
            + \mbox{\textnormal{S}}\mathbf{f}(\mathbf{V}).
        \end{align}
        Thus, $\mbox{\textnormal{S}}(\mathbf{Z} - \Vb)$ is a solution to
        the same linear equation that $\mathbf{Z}$ solves, except
        the inhomogeneous terms for
        $\mbox{\textnormal{S}}(\mathbf{Z} - \Vb)$ are given by the
        right-hand side of \eqref{E:SZEquation} and the initial
        data are given by $\mbox{\textnormal{S}}(\mathring{\mathbf{Z}} - \Vb).$
        Equation \eqref{E:Duhamel} now follows from Duhamel's principle. Note that
        $\mbox{\textnormal{S}}\mathbf{f}(\mathbf{V})$ is well-defined since
        $\mathbf{f}(\bar{\mathbf{V}})=\mathbf{0},$ and we can therefore apply Proposition \ref{P:SobolevTaylor} and Remark
        \ref{R:SobolevTaylorCalculusRemark} to conclude that $\mathbf{f}(\mathbf{V}) \in H^N.$
    \end{proof}

    \subsubsection{Proof of Theorem \ref{T:ContinuousDependence}}
    We will now demonstrate
    \eqref{E:HNContinuousDependence} by providing a proof of the equivalent
    statement
    \begin{align} \label{E:HNContinuousDependenceEquivalent}
        \lim_{m \to \infty} \mid\mid\mid \mbox{\textnormal{S}}\left(\Vm -
        \mathbf{V} \right) \mid\mid\mid_{L^2,T} = 0.
    \end{align}
    Lemma \ref{L:Duhamel} implies the following equality, valid
    for $0 \leq t \leq T':$
    \begin{align} \label{E:DuhmaelDifference}
       &\mbox{\textnormal{S}}\big(\Vm(t) - \mathbf{V}(t) \big) =
        \mathcal{U}^m(t,0)\mbox{\textnormal{S}}\big(\ringVm - \ringV \big)
        + \big(\mathcal{U}^m(t,0) - \mathcal{U}(t,0)\big)\mbox{\textnormal{S}}(\ringV - \Vb) \notag \\
        &+ \int_0^t \mathcal{U}(t,t') \mathcal{B}(t') \mbox{\textnormal{S}}(\mathbf{V}(t') - \Vb) \, dt'-
        \int_0^t \mathcal{U}^m(t,t') \Bm(t')
        \mbox{\textnormal{S}}(\Vm(t') - \Vb)\, dt' \notag \\
        &+ \int_0^t \mathcal{U}^m(t,t')\mbox{\textnormal{S}}\mathbf{f}(\Vm(t')) \, dt' -
        \int_0^t \mathcal{U}(t,t') \mbox{\textnormal{S}}\mathbf{f}(\mathbf{V}(t')) \, dt'.
    \end{align}

    By Lemma \ref{L:BoundedLinearOperators}, we have that
    \begin{align} \label{E:InititalDataEvolutionHNBound}
       \mid\mid\mid \mathcal{U}^m(t,0)\mbox{\textnormal{S}}\big(\ringVm - \ringV
        \big)\mid\mid\mid_{L^2,T'} \leq C(\mbox{\textnormal{S}},K) \mid\mid \ringVm - \ringV \mid\mid_{H^N}.
    \end{align}
    Furthermore, if we define
    \begin{align} \label{E:SupnormOperatorDifference}
    	\sup_{t \in [0,T']} \|\big(\mathcal{U}^m(t,0) - \mathcal{U}(t,0)\big)\mbox{\textnormal{S}}(\ringV - \Vb)\|_{L^2} 
    	\overset{\mbox{\tiny{def}}}{=} d_m,
    \end{align}
    then Lemma \ref{L:StrongConvergenceL2HN} implies that 
    \begin{align} \label{E:dmConvergestoZero}
    	\lim_{m \to \infty} d_m = 0.
    \end{align}

    We now rewrite the second line of \eqref{E:DuhmaelDifference} as
    \begin{align} \label{E:DuhamelDifferenceRewritten}
        &\int_0^t (\mathcal{U}(t,t') - \mathcal{U}^m(t,t'))
        \mathcal{B}(t')\mbox{\textnormal{S}}(\mathbf{V}(t') - \Vb)
        \,dt' \notag \\
        + &\int_0^t \mathcal{U}^m(t,t')
        \left(\mathcal{B}(t')-\Bm(t')\right)\mbox{\textnormal{S}}(\mathbf{V}(t') - \Vb) \,dt' \notag \\
        + &\int_0^t \mathcal{U}^m(t,t')
        \Bm(t')\mbox{\textnormal{S}}\left(\mathbf{V}(t') - \Vm(t') \right)
        \,dt'. 
    \end{align}

    By \eqref{E:UniformSobolev}, Lemma \ref{L:BoundedLinearOperators} and Lemma \ref{L:OperatorBL2Norm}, for 
    $0 \leq t \leq T_* \leq T',$ the $L^2$ norms of the second and third integrals in \eqref{E:DuhamelDifferenceRewritten} are each 
    bounded from above by \\ $C(\mbox{\textnormal{S}},K)T_*\mid\mid\mid \Vm - \mathbf{V}\mid\mid\mid_{H^N,T*}.$ 

    We similarly split the third line of
    \eqref{E:DuhmaelDifference} into two terms and use \eqref{E:UniformSobolev}, Lemma \ref{L:BoundedLinearOperators}, Proposition 
    \ref{P:SobolevTaylor}, and Remark \ref{R:SobolevTaylorCalculusRemark} to bound the $L^2$ 
    norm of one of them, namely $\int_0^t \mathcal{U}^m(t,t') \mbox{\textnormal{S}}\big(\mathbf{f}(\Vm(t') - 
    \mathbf{f}(\mathbf{V}(t'))\big) \, dt',$ from above by \\
    $C(\mbox{\textnormal{S}},K)T_*\mid\mid\mid\Vm - \mathbf{V}\mid\mid\mid_{H^N,T*}.$

    Combining these estimates
    with \eqref{E:InititalDataEvolutionHNBound} and \eqref{E:SupnormOperatorDifference}, we take the $L^2$
    norm of each side of \eqref{E:DuhmaelDifference} followed by the $\sup$ over $t \in [0,T_*]$ to arrive at
    the inequality
    \begin{align}
        \mid\mid\mid\mbox{\textnormal{S}}\left(\Vm - \mathbf{V}
        \right)\mid\mid\mid_{L^2,T_*} &\leq C(\mbox{\textnormal{S}},K) \|\ringVm - \ringV\|_{H^N}  
        + d_m \notag \\
        &+ C(\mbox{\textnormal{S}},K) T_*\mid\mid\mid \Vm -
        \mathbf{V}\mid\mid\mid_{H^N,T_*} \notag \\
        &+ \int_0^{T_*} \sup_{t \in [0,T_*]}\|(\mathcal{U}(t,t') - \mathcal{U}^m(t,t'))
        \mathcal{B}(t')\mbox{\textnormal{S}}(\mathbf{V}(t') - \Vb)\|_{L^2} \,dt' \notag \\
        &+ \int_0^{T_*} \sup_{t \in [0,T_*]}\|(\mathcal{U}^m(t,t') - \mathcal{U}(t,t'))
        \mbox{\textnormal{S}}\mathbf{f}(\mathbf{V}(t'))\|_{L^2} \, dt'. 
    \end{align}

     We now choose $T_*$ small enough so that
     \begin{align} \label{E:TStarSmall}
        C(\mbox{\textnormal{S}},K)T_*\mid\mid\mid \Vm - \mathbf{V}\mid\mid\mid_{H^N,T_*}
        \leq \frac{1}{2}\mid\mid\mid\mbox{\textnormal{S}}\left(\Vm - \mathbf{V}
        \right)\mid\mid\mid_{L^2,T_*},
     \end{align}
     from which it follows that
     \begin{align} 
        \mid\mid\mid\mbox{\textnormal{S}}\left(\Vm - \mathbf{V}
        \right)&\mid\mid\mid_{L^2,T_*} \ \leq 2C(\mbox{\textnormal{S}},K) \|\ringVm - \ringV\|_{H^N} + 2d_m \notag \\ 
        &+ 2\int_0^{T_*} \sup_{t \in [0,T_*]} \|\left(\mathcal{U}^m(t,t') - \mathcal{U}(t,t')\right)
        \mathcal{B}(t')\mbox{\textnormal{S}}(\mathbf{V}(t') - \Vb)\|_{L^2} \,dt' \notag \\ 
        &+ 2\int_0^{T_*} \sup_{t \in [0,T_*]} \|\left(\mathcal{U}^m(t,t') - \mathcal{U}(t,t')\right)\mbox{\textnormal{S}}\big(\mathbf{f}({\mathbf{V}}(t')) - \mathbf{f}(\bar{\mathbf{V}})\big)\|_{L^2}
        \,dt',  \label{E:HNContinuousEstimate}
    \end{align}
    \noindent where in \eqref{E:HNContinuousEstimate}, we have used the fact that $\mathbf{f}(\bar{\mathbf{V}}) = \mathbf{0}.$

    By \eqref{E:UniformSobolev}, Lemma \ref{L:BoundedLinearOperators}, Lemma
    \ref{L:OperatorBL2Norm}, Proposition
    \ref{P:CompositionProductSobolevMoser}, Remark \ref{R:SobolevCalculusRemark}, Proposition \ref{P:SobolevTaylor}, and Remark
    \ref{R:SobolevTaylorCalculusRemark}, the two integrands in
    \eqref{E:HNContinuousEstimate}, viewed as functions of $t',$ are uniformly bounded by
    $C(K)$ on $[0,T_*].$ Furthermore, by Lemma
    \ref{L:StrongConvergenceL2HN}, the integrands converge to $0$ pointwise in $t'$ as $m \to
    \infty.$ Therefore, by the dominated convergence theorem, the
    two integrals in \eqref{E:HNContinuousEstimate} converges to $0$ as $m \to \infty.$
    Recalling that by hypothesis we have that $\lim_{m \to \infty}\|\ringVm -
    \ringV\|_{H^N} = 0,$ and also using \eqref{E:dmConvergestoZero}, we conclude that
    \begin{align}
        \lim_{m \to \infty}&\mid\mid\mid\mbox{\textnormal{S}}\left(\Vm - \mathbf{V}
        \right)\mid\mid\mid_{L^2,T_*}=0.
    \end{align}

    To extend this argument to the interval $[0,2T_*],$ let
    $\epsilon > 0$ and choose $m_0$ large enough so that $m \geq m_0$
    implies that $\mid\mid\mid \Vm - \mathbf{V}
    \mid\mid\mid_{H^N,T_*}< {\epsilon} / \big(4C(\mbox{\textnormal{S}},K)\big).$ Starting from
    time $T_*,$ we may argue as above to show that
    \begin{align}
        \limsup_{m \to \infty} \sup_{t \in [T_*,2T_*]}\|{\textnormal{S}}\left(\Vm - \mathbf{V}
        \right)\|_{H^N} \leq \frac{1}{2}\epsilon.
    \end{align}
    Thus, we can choose $m_1 \geq m_0$
    such that $\mid\mid\mid {\textnormal{S}}\left(\Vm - \mathbf{V}
    \right) \mid\mid\mid_{L^2,2T_*} \leq \epsilon$ when $m \geq m_1.$  Continuing in this
    manner, we may inductively extend this argument to the interval
    $[0,T'].$ We state for emphasis that the size of $T_*$ required to satisfy
    the inequality \eqref{E:TStarSmall} depends only on $C(\mbox{\textnormal{S}},K).$ Consequently, the length of the time
    interval of extension $T_*$ may be chosen to be the same at each step in the
    induction.

    We now show that this argument can be extended to the entire
    interval $[0,T]$ on which $\mathbf{V}$ exists. Define
    \begin{align}
    T_{max} \overset{\mbox{\tiny{def}}}{=}  \sup \lbrace T' \ | \
        &\mathbf{V} \ \mbox{and the} \ \Vm \ \mbox{exist on the interval} \notag \\
        &[0,T'] \ \mbox{for all large} \ m \ \mbox{and}
    \ \lim_{m \to \infty} \mid\mid\mid \Vm - \mathbf{V} \mid\mid\mid_{H^N,T'} = 0 \rbrace. \notag
    \end{align}
    We will show that the assumption $T_{max} < T$ leads to a
    contradiction.

    By Theorem \ref{T:LocalExistence} and Corollary \ref{C:ExistenceInterval}, for each $t \in
    [0,T],$ there exist an $H^N$ neighborhood $B_{\delta_t}(\mathbf{V}(t))$ of $\mathbf{V}(t)$ with positive radius
    $\delta_t$ and a $\Delta_t>0$ such that initial data
    belonging to $B_{\delta_t}(\mathbf{V}(t))$ launch a unique solution\footnote{This solution may escape
    $\bar{\mathcal{O}}_2,$ but this is not a difficulty since $\bar{\mathcal{O}}_2 
    \Subset \mathcal{O};$ the solution still has some ``room" left to evolve and remain in a compactly supported
    convex subset of $\mathcal{O}.$} that exists
    on the interval $[t,t+ \Delta_t]$ (the term ``initial" here refers to the time $t$). By
    continuity, $\mathbf{V}([0,T])$ is a compact subset of
    $H_{\Vb}^N.$ Therefore, there exist $\delta > 0$
    and $\Delta > 0$ such that initial data belonging to $B_{\delta}(\mathbf{V}(t))$
    launch a unique solution that exists on the interval $[t,t+
    \Delta].$ Furthermore, by Corollary \ref{C:UniformNormBound}, there exists a $C_{uniform}>0$ such that for any 
    ``initial" data $\mathring{\widetilde{\mathbf{V}}}$ contained in the ball
    $B_{\delta}(\mathbf{V}(t)),$ the corresponding solution $\widetilde{\mathbf{V}}$ to the 
    EN$_{\kappa}$ system satisfies the bounds
    \begin{align}
			\mid\mid\mid \widetilde{\mathbf{V}} \mid\mid\mid_{H_{\widetilde{\mathbf{V}}}^N,[t,t+\Delta]} &\leq C_{uniform} \notag \\
			\mid\mid\mid \partial_t \widetilde{\mathbf{V}} \mid\mid\mid_{H^{N-1},[t,t+\Delta]} &\leq C_{uniform}. \label{E:UniformSobolevBounds}
		\end{align} 
    We emphasize that $\delta$ and $\Delta$ \emph{are
    independent of $t$ belonging to $[0,T],$} and that $C_{uniform}$ is independent of the data. Note that as a consequence
    of this reasoning, it follows that $\mathbf{V}$ exists on the interval $[0,T + \Delta].$

    The contradiction is now easily obtained: assume that $T_{max} < T.$ Then according to the
    above paragraph, initial data belonging to $B_{\delta}(\mathbf{V}(T_{max} - \frac{1}{2}
    \Delta))$ launch a solution that exists on the interval
    $[T_{max} - \frac{1}{2} \Delta,T_{max} +
    \frac{1}{2}\Delta].$ Furthermore, for all large $m,$ $\Vm(T_{max} - \frac{1}{2}\Delta)$
    is contained in $B_{\delta}(\mathbf{V}(T_{max} - \frac{1}{2} \Delta)).$
    Therefore, for all large $m,$ $\Vm$ can be extended to a
    solution that exists on $[0,T_{max} + \frac{1}{2}\Delta].$ In addition, using \eqref{E:UniformSobolevBounds},
    we have that for all large $m,$ $\mid\mid\mid \Vm \mid\mid\mid_{H_{\widetilde{\mathbf{V}}}^N,T_{max} + \frac{1}{2}\Delta} \leq 
    C_{uniform},$ and $\mid\mid\mid \partial_t \Vm \mid\mid\mid_{H^{N-1},T_{max}+\frac{1}{2}\Delta} \leq C_{uniform}.$
   	Therefore, we can repeat the entire argument given in Section \ref{SS:ContinuousDependence}, substituting
   	$T_{max} + \frac{1}{2}\Delta$ in \eqref{E:UniformSobolev} in place of $T'$ and $T,$ and $C_{uniform}$ in 
   	place of $K,$ to show that $\lim_{m \to \infty} 
   	\mid\mid\mid \Vm - \mathbf{V} \mid\mid\mid_{H^N,T_{max} + \frac{1}{2}\Delta}=0.$ This contradicts the definition of $T_{max}$ and 
   	completes the proof of Theorem \ref{T:ContinuousDependence}. \qquad $\Box$

\section*{Acknowledgments}

This work would not have been possible without many hours of
discussion with and encouragement from Michael Kiessling and A.
Shadi Tahvildar-Zadeh. I would also like to thank Demetrios
Christodoulou for his generous correspondence concerning my
inquiries into his work, and the anonymous referee for some helpful
comments that aided my revision of the first draft. Work supported by NSF Grant DMS-0406951. Any opinions,
conclusions, or recommendations expressed in this material are those of the author
and do not necessarily reflect the views of the NSF.

\setcounter{section}{0}
\setcounter{subsection}{0}
\setcounter{subsubsection}{0}
\setcounter{equation}{0}   
\renewcommand{\thesection}{\Alph{section}}
\renewcommand{\theequation}{\Alph{section}.\arabic{equation}}
\renewcommand{\theproposition}{\Alph{section}.\arabic{proposition}}
\renewcommand{\thecorollary}{\Alph{section}.\arabic{corollary}}
\renewcommand{\thetheorem}{\Alph{section}.\arabic{theorem}}
\renewcommand{\theremark}{\Alph{section}.\arabic{remark}}
\renewcommand{\thelemma}{\Alph{section}.\arabic{lemma}}

\section{Appendix} \label{S:Appendix}

    In this Appendix, we use notation that is as consistent as possible with
    our use of notation in the body of the paper. To conserve space, we refer the reader
    to the literature instead of providing proofs: propositions \ref{P:CompositionProductSobolevMoser} and
    \ref{P:SobolevTaylor} are similar to propositions proved in
    chapter 6 of \cite{lH1997}, while Proposition
    \ref{P:SobolevMissingDerivativeProposition} is proved in
    \cite{sKaM1981}. The corollaries and remarks below are straightforward extensions
    of the propositions. With the exception of Proposition \ref{P:SobIterpolation}, which is a standard Sobolev interpolation
    inequality, the proofs of the propositions given in the literature are commonly based on the following version of the
    Gagliardo-Nirenberg inequality \cite{lN1959}, together with repeated use of H\"{o}lder's inequality and/or Sobolev embedding:
\begin{lemma}                                                                                           \label{L:GN}
    If $i,k \in \mathbb{N}$ with $0 \leq i \leq k,$ and $\mathbf{V}$ is a
    scalar-valued or array-valued function on $\mathbb{R}^d$ satisfying $\mathbf{V} \in
    L^{\infty}(\mathbb{R}^d)$ and $\|\nabla^{(k)} \mathbf{V}\|_{L^2(\mathbb{R}^d)} < \infty,$ then
    \begin{align}
        \| \nabla^{(i)} \mathbf{V} \|_{L^{2k/i}} \leq C(k) \|\mathbf{V}\|_{L^{\infty}}^{1 -
        \frac{i}{k}}\|\nabla^{(k)} \mathbf{V}\|_{L^2}^{\frac{i}{k}}.
    \end{align}
\end{lemma}

\begin{proposition}                                                                                                \label{P:CompositionProductSobolevMoser}
    Let $\mathcal{O}_2 \subset \mathbb{R}^{n}$ be a bounded open set, and let $j,d \in \mathbb{N}$ with $ j > \frac{d}{2}.$
    Let $\mathbf{V}:\mathbb{R}^d \rightarrow \mathbb{R}^n$ be an element of $H^j(\mathbb{R}^d),$ and assume that
    $\mathbf{V}(\mathbb{R}^d) \subset \bar{\mathcal{O}}_2.$ Let $F \in C_b^j(\bar{\mathcal{O}}_2)$
    be a $q \times q$ matrix-valued function,
    and let $G \in H^j(\mathbb{R}^d)$ be a $q \times q$ ($q \times 1)$ matrix-valued (array-valued) function. Then the $q \times q$ 	
    $(q \times 1)$ matrix-valued (array-valued) function $(F \circ \mathbf{V})G$ is an element of $H^j(\mathbb{R}^d)$ and
    \begin{align} \label{E:CompositionProductSobolevMoser}
        \|(F \circ \mathbf{V})G\|_{H^j(\mathbb{R}^d)} \leq C(j,d)|F|_{j,\bar{\mathcal{O}}_2}(1 +
        \|\mathbf{V}\|_{H^j(\mathbb{R}^d)}^j)\|G\|_{H^j(\mathbb{R}^d)}.
    \end{align}
	\end{proposition}

\begin{corollary}                                                                                               \label{C:SobolevCorollary}
    Assume the hypotheses of Proposition
    \ref{P:CompositionProductSobolevMoser} with the following
    changes: \\ $\mathbf{V},G \in C^0([0,T],H^j(\mathbb{R}^d)).$
    Then the $q \times q$ $(q \times 1)$ matrix-valued (array-valued) function $(F \circ \mathbf{V})G$ is an element of
    $C^0([0,T],H^j(\mathbb{R}^d)).$
\end{corollary}

\begin{remark}                                                                                          \label{R:SobolevCalculusRemark}
    We often make use of a slight modification of Proposition \ref{P:CompositionProductSobolevMoser}
    in which the assumption $\mathbf{V} \in H^j(\mathbb{R}^d)$ is replaced with the assumption
    $\mathbf{V} \in H_{\Vb}^j(\mathbb{R}^d),$ where $\bar{\mathbf{V}} \in \mathbb{R}^n$ is a
    constant array. Under this modified assumption, the conclusion of
    Proposition \ref{P:CompositionProductSobolevMoser}
    is modified as follows:
    \begin{align}                                                                                       \label{E:ModifiedSobolevEstimate}
        \|(F \circ \mathbf{V})G\|_{H^j} \leq C(j,d)|F|_{j,\bar{\mathcal{O}}_2}(1 +
        \|\mathbf{V}\|_{H_{\Vb}^j}^j)\|G\|_{H^j}.
    \end{align}
    A similar modification can be made to Corollary \ref{C:SobolevCorollary}.
\end{remark}

\begin{proposition}                                                                                             \label{P:SobolevTaylor}
    Let $\mathcal{O}_2 \subset \mathbb{R}^{n}$ be a bounded open set with convex closure $\bar{\mathcal{O}}_2$, and let $j,d \in 	
    \mathbb{N}$ with $j > \frac{d}{2}.$ Let $F \in C_b^j(\bar{\mathcal{O}}_2)$ be a  scalar or array-valued
    function. Let $\mathbf{V}, \widetilde{\mathbf{V}}: \mathbb{R}^d \rightarrow
    \mathbb{R}^n,$ and assume that $\mathbf{V}, \widetilde{\mathbf{V}}
    \in H^j(\mathbb{R}^d).$ Assume further that 
    $\bar{\mathbf{V}} \in \bar{\mathcal{O}}_2$ and $\mathbf{V}(\mathbb{R}^d), \widetilde{\mathbf{V}}(\mathbb{R}^d) \subset 
    \bar{\mathcal{O}}_2.$ Then $F \circ \mathbf{V} - F \circ \widetilde{\mathbf{V}} \in H^j(\mathbb{R}^d)$ and
    \begin{align}
        \|F \circ \mathbf{V} - F \circ \widetilde{\mathbf{V}}\|_{H^j}
        \leq C(j,d,\|\mathbf{V}\|_{H^j},\|\widetilde{\mathbf{V}}\|_{H^j})|F|_{j+1,\bar{\mathcal{O}}_2} \|\mathbf{V} - 
        \widetilde{\mathbf{V}}\|_{H^j}.
    \end{align}
\end{proposition}

\begin{remark}     \label{R:SobolevTaylorCalculusRemark}
    As in Remark \ref{R:SobolevCalculusRemark}, we may replace the hypotheses $\mathbf{V}, \mathbf{\widetilde{V}} \in H^j(\mathbb{R}^d)$ 
    from Proposition
    \ref{P:SobolevTaylor} with the hypotheses $\mathbf{V}, \mathbf{\widetilde{V}} \in
    H_{\Vb}^j(\mathbb{R}^d),$ where $\Vb$ is a constant array, in which case the conclusion
    of the proposition is:
    \begin{align}                                                                                       \label{E:ModifiedSobolevEstimate2}
        \|(F \circ \mathbf{V}) - (F \circ \mathbf{\widetilde{V}}) \|_{H^j} \leq
        C(j,d,\|\mathbf{V}\|_{H_{\Vb}^j},\|\widetilde{\mathbf{V}}\|_{H_{\Vb}^j})|F|_{j+1,\bar{\mathcal{O}}_2} \|\mathbf{V} -
        \widetilde{\mathbf{V}}\|_{H^j}.
    \end{align}
    Furthermore, a careful analysis of the special case $\widetilde{\mathbf{V}} =
    \bar{\mathbf{V}}$ gives the bound
    \begin{align}
        \|F \circ \mathbf{V} - F \circ \bar{\mathbf{V}}\|_{H^j} \leq C(j,d)|\partial F/\partial \mathbf{V}|_{j-1,\bar{\mathcal{O}}_2}(1 + 
        \|\mathbf{V}\|_{H_{\Vb}^j}^{j-1}) (\|\mathbf{V}\|_{H_{\Vb}^j}),
    \end{align}
    in which we require less regularity of $F$ than we do in the general case.
\end{remark}

\begin{proposition}                                                                             \label{P:SobolevMissingDerivativeProposition}
        Assume the hypotheses of Proposition \ref{P:CompositionProductSobolevMoser} with the following two changes:
        \begin{enumerate}
            \item Assume $j > \frac{d}{2} + 1.$
            \item Assume that $G \in H^{j-1}(\mathbb{R}^d).$
        \end{enumerate}
        Let $k \in \mathbb{N}$ with $1 \leq k \leq j,$ and let $\vec{\alpha}$ be a spatial
        derivative multi-index with $|\vec{\alpha}|=k.$ Then
        \begin{align}       \label{E:SobolevMissingDerivativeProposition}
            \|\partial_{\vec{\alpha}}&\left((F \circ \mathbf{V})G\right) - (F \circ
                \mathbf{V})\partial_{\vec{\alpha}}G\|_{L^2} \notag \\
            &\leq C(j,d)|\partial F/\partial \mathbf{V}|_{j-1,\bar{\mathcal{O}}_2}(\|\mathbf{V}\|_{H^j} +
                \|\mathbf{V}\|_{H^j}^j)\|G\|_{H^{j-1}}. 
        \end{align}
\end{proposition}

    \begin{remark}                                                              \label{R:SobolevMissingDerivativeRemark}
        As in Remark \ref{R:SobolevCalculusRemark}, we may replace the assumption $\mathbf{V} \in
        H^j(\mathbb{R}^d)$ in Proposition \ref{P:SobolevMissingDerivativeProposition}
        with the assumption $\mathbf{V} \in H_{\Vb}^j(\mathbb{R}^d),$ where $\bar{\mathbf{V}}$ is a constant array, in
        which case we obtain
        \begin{align}
             \|\partial_{\vec{\alpha}}&\left((F \circ \mathbf{V})G\right) - (F \circ
             \mathbf{V})\partial_{\vec{\alpha}}G\|_{L^2} \notag \\
             &\leq C(j,d)|\partial F/\partial \mathbf{V}|_{j-1,\bar{\mathcal{O}}_2}(\|\mathbf{V}\|_{H_{\Vb}^j} +
                \|\mathbf{V}\|_{H_{\Vb}^j}^j)\|G\|_{H^{j-1}}. 
        \end{align}
    \end{remark}

\begin{proposition}                                                           \label{P:SobIterpolation}
            Let $N',N \in \mathbb{R}$ be such that $0 \leq N' \leq N,$ and assume that $F \in H^N(\mathbb{R}^d).$ Then
            \begin{align}
                \|F\|_{H^{N'}} \leq C(N',d) \|F\|_{L^2}^{1-
                N'/N}\|F\|_{H^N}^{N'/N}.
            \end{align}
\end{proposition}


\begin{thebibliography}{0}
\bibitem{nAgG2007}
{N. Andersson and G.L. Comer}, {\it Relativistic fluid dynamics: Physics
  for many different scales}, Living Rev.
  Relativity, 10 (2007). [Online article]: cited on 31 Mar 2008,
  {http://relativity.livingreviews.org/Articles/lrr-2007-1/}.
  
\bibitem{uBlK2003}
{U. Brauer and L. Karp}, Local existence of classical solutions for
  the {Einstein-Euler} system using weighted {Sobolev} spaces of fractional
  order, {\it C. R. Math. Acad. Sci. Paris}, {\bf 1} (2007) 49--54.

\bibitem{sC2003}
 S. Calogero, Spherically symmetric steady states of galactic
  dynamics in scalar gravity, {\it Classical and Quantum Gravity}, {\bf 20} (2003)
  1729--1741.

\bibitem{sC2006}
S. Calogero, Global classical
  solutions to the {3D} {Nordstr\"{o}m-Vlasov} system, {\it  Comm. Math. Phys.}, {\bf 266} (2006)  343--353.
  
\bibitem{yCB1958}
Y. F. (Choquet)-Bruhat, Th\'{e}or\`{e}mes d'existence en
  m{\'{e}}canique des fluides relativistes, {\it Bull. Soc. Math. France}, {\bf 86} (1958) 155--175.

\bibitem{dC2000}
D. Christodoulou, {\it The Action Principle and Partial Differential
  Equations} (Princeton University Press, Princeton, NJ, 2000).
  
\bibitem{dC2007}
D. Christodoulou, {\it The Formation of
  Shocks in 3-Dimensional Fluids}, (European Mathematical Society,
  \textnormal{Z\"{u}rich}, Switzerland, 2007).

\bibitem{dC1995}
D. Christodoulou, Self-gravitating relativistic fluids: A two-phase model,
{\it Arch. Rational Mech. Anal.}, {\bf 130} (1995) 343--400.
    
\bibitem{rCdH1966}
R. Courant and D. Hilbert, {\it Methods of Mathematical Physics, Volume
  II}, (Interscience Publishers, New York, London, Sydney, 1966).

\bibitem{cD2000}
C.M. Dafermos, {\it Hyperbolic Conservation Laws in Continuum Physics},
  (Springer-Verlag, New York, 2000).

\bibitem{tDgEF1992}
T. Damour and G. Esposito-Far\`{e}se, Tensor multi-scalar theories
  of gravitation, {\it Classical Quantum Gravity}, {\bf 9} (1992) 2093--2176.
  
\bibitem{aE1917}
A. Einstein, Kosmologische {Betrachtungen} zur allgemeinen
  {Relativit\"{a}tstheorie}, {\it Sitzungsberichte der K{\"{o}}niglich
  Preu{\ss}ische Akademie der Wissenschaften (Berlin)}, {\bf 142-152} (1917)
  235--237.

\bibitem{kF1954}
K. O. Friedrichs, Symmetric hyperbolic linear differential
  equations, {\it Comm. Pure Appl. Math.}, {\bf 7} (1954)
  345--392.

\bibitem{yGsTZ1999}
Y. Guo and S. Tahvildar-Zadeh, Formation of singularities in
  relativistic fluid dynamics and in spherically symmetric plasma dynamics,
  {\it Contemp. Math.}, {\bf 238} (1999) 151--161.

\bibitem{lH1997}
L. H{\"{o}rmander}, {\it Lectures on Nonlinear Hyperbolic Differential
  Equations}, (Springer-Verlag, New York, 1997).

\bibitem{tK1970}
T. Kato, Linear evolution equations of {``Hyperbolic"} type,
{\it J. Fac. Sci. Univ. Tokyo Sect. I}, {\bf 17}
  (1970) 241--258.

\bibitem{tK1973}
T. Kato, Linear evolution
  equations of {``Hyperbolic"} type ii, {\it  J. Math. Soc. Japan}, {\bf 25} (1973) 648--666.

\bibitem{tK1975}
T. Kato, The {Cauchy} problem
  for quasi-linear symmetric hyperbolic systems, {\it Arch. Rational Mech. Anal.}, {\bf 58} (1975) 181--205.

\bibitem{mK2003}
M. Kiessling, The {``Jeans swindle:"} a true story---mathematically
  speaking, {\it  Adv. in Appl. Math.}, {\bf 31} (2003) 132--149.

\bibitem{sKaM1981}
S. Klainerman and A. Majda, Singular limits of quasilinear
  hyperbolic systems with large parameters and the incompressible limit of
  compressible fluids, {\it Comm. Pure Appl. Math.}, {\bf 34}
  (1981) 481--524.

\bibitem{pL2006}
P. Lax, {\it Hyperbolic Partial Differential Equations}, (American
  Mathematical Society, Providence, Rhode Island, 2006).

\bibitem{aM1984}
A. Majda, {\it Compressible Fluid Flow and Systems of Conservation Laws
  in Several Space Variables}, (Springer-Verlag, New York, 1984).

\bibitem{tM1986}
T. Makino, On a local existence theorem for the evolution equation
  of gaseous stars, {\it Patterns and Waves - Qualitative Analysis of Nonlinear
  Differential Equations}, (1986) 459--479.

\bibitem{tMsU1995}
T. Makino and S. Ukai, Local smooth solutions of the {Relativistic
  Euler Equation}, {\it  J. Math. Kyoto Univ.}, {\bf 35} (1995)
  105--114.

\bibitem{tMsU1995b}
T. Makino and S. Ukai, Local smooth
  solutions of the {Relativistic Euler Equation}, ii, {\it  Kodai Math. J.}, 
  {\bf 18} (1995) 365--375.

\bibitem{cMkTjW1973}
C. W. Misner, K. S. Thorne, and J. A. Wheeler, {\it Gravitation}, (W H
  Freeman, San Francisco, 1973).

\bibitem{lN1959}
L. Nirenberg, On elliptic partial differential equations, {\it Ann. Scuola Norm. Sup. Pisa (3)}, {\bf 13} (1959)
  115--162.

\bibitem{gN1913}
G. Nordstr{\"{o}}m, {Zur Theorie der Gravitation vom Standpunkt des
  Relativit{\"{a}}tsprinzips}, {\it Annalen der Physik}, {\bf 42} (1913) 533--554.

\bibitem{fR2004}
F. Ravndal, Scalar gravitation and extra dimensions. {Proceedings}
  of the {Gunnar Nordst\"{o}rm} symposium on theoretical physics,
  {\it Comment. Phys.-Math.}, {\bf 166} (2004) 151--164.

\bibitem{aR1992}
A. D. Rendall, The initial value problem for a class of general
  relativistic fluid bodies, {\it  J. Math. Phys.}, {\bf 33} (1992)
  1047--1053.

\bibitem{dS1999}
D. Serre, {\it Systems of Conservation Laws I: Hyperbolicity, Entropies,
  Shock Waves}, (Cambridge University Press, New York, 1999).

\bibitem{sSsT1993}
S. L. Shapiro and S. A. Teukolsky, Scalar gravitation: A laboratory
  for numerical relativity, {\it Phys. Rev. D}, {\bf 47} (1993) 1529--1540.

\bibitem{MT1991}
Michael E. Taylor, {\it Pseudodifferential Operators and Nonlinear PDE},
(Birkh{\"{a}}user, Boston, 1991).

\bibitem{kY1980}
K. Yosida, {\it Functional Analysis}, (Springer-Verlag, New York, 1980).

\end{thebibliography}
\end{document}